\documentclass[manuscript]{acmart}
\usepackage{booktabs}
\usepackage[linesnumbered,ruled,noend]{algorithm2e}

\SetAlFnt{\small}
\SetAlCapFnt{\small}
\SetAlCapNameFnt{\small}
\SetAlCapHSkip{0pt}
\IncMargin{-\parindent}
\setcopyright{none}
\usepackage{layout}
\usepackage{dsfont}
\usepackage{bbm} 
\usepackage{thmtools}
\usepackage{thm-restate}
\usepackage{cleveref}
\usepackage{newtxtext,newtxmath}
\usepackage[textsize=tiny]{todonotes}

\usepackage{comment}
\usepackage{enumitem}
\usepackage{nicefrac}
\usepackage{pgfplots}
\usepackage{soul}
\usepackage{subcaption}
\usepackage{tabularx}

\usepackage{times}
\usepackage{hyperref} 

\usepackage{array,etoolbox}
\preto\tabular{\setcounter{magicrownumbers}{0}}
\newcounter{magicrownumbers}
\newcommand\rownumber{\stepcounter{magicrownumbers}\arabic{magicrownumbers}}

\newcounter{rownumber}[figure] 
\usepackage{pdflscape}

\DeclareRobustCommand{\citet}[1]{\citeauthor{#1}~\shortcite{#1}}
\newcommand{\cP}{{\sf P}}
\newcommand{\hG}{\hat{G}}
\newcommand{\D}{\mathcal{D}}

\newcommand{\I}{I}
\newcommand{\J}{(G,\V)}

\newcommand{\N}{\mathcal{N}}
\newcommand{\mech}{\mathcal{M}}
\newcommand{\natv}{\texttt{ncl}}
\newcommand{\intv}{\texttt{icl}}
\newcommand{\order}{\text{order}}
\newcommand{\nat}{\text{nat}}
\newcommand{\intt}{\text{int}}
\newcommand{\con}{\text{con}}

\newcommand{\C}{\mathcal{C}}
\newcommand{\segv}{\texttt{iss}}
\newcommand{\sv}{\texttt{isn}}
\newcommand{\X}{\mathcal{X}}

\newcommand{\V}{\mathcal{V}}
\newcommand{\cint}{c_{{\tt int}}}
\newcommand{\Mnat}{\mech_{\nat}}
\newcommand{\Mcon}{\mech_{\con}}
\newcommand{\Mint}{\mech_{\intt}}
\newcommand{\Morder}{\mech_{\order}}

\newcommand{\NP}{{\sf NP}}

\newcommand{\dia}{\hfill{$\diamond$}}
\newcommand{\pG}{(G, \V)}
\newcommand{\gcyc}{\mathcal{C}(\Gamma, \pG)}

\usepackage{boxedminipage}

\newcommand{\decisionprob}[3]{
        \begin{center}
                \begin{boxedminipage}{.99\textwidth}
                       {#1}\\[2pt]
                        \begin{tabular}{ r p{0.8\textwidth}}
                                \textit{~~~~Instance:} & {#2}\\
                                \textit{Question:} & {#3}
                        \end{tabular}
                \end{boxedminipage}
        \end{center}
}

\pgfplotsset{compat=1.18}

\usetikzlibrary{arrows.meta,cd, fit, backgrounds,positioning}
\usetikzlibrary{shapes.geometric}
\usetikzlibrary{decorations.pathreplacing, angles, quotes}
\usetikzlibrary{bending}

\newtheorem{theorem}{Theorem}

\newtheorem{proposition}{Proposition}

\newtheorem{example}{Example}
\newtheorem{lemma}{Lemma}
\newtheorem{remark}{Remark}
\title[International Kidney Exchange Programmes with Country-Specific Parameters]{Complexity and Manipulation of International Kidney Exchange Programmes with Country-Specific Parameters}

\author{Rachael Colley}
\email{rachael.colley@glasgow.ac.uk}
\orcid{0000-0003-1164-4234}
\affiliation{%
  \institution{School of Computing Science, University of Glasgow}
  \city{Glasgow}
  \country{United Kingdom}
}

\author{David Manlove}
\email{david.manlove@glasgow.ac.uk}
\orcid{0000-0001-6754-7308}
\affiliation{%
  \institution{School of Computing Science, University of Glasgow}
  \city{Glasgow}
  \country{United Kingdom}
}

\author{Daniel Paulusma}
\email{daniel.paulusma@durham.ac.uk}
\orcid{0000-0001-5945-9287}
\affiliation{%
  \institution{
Department of Computer Science, Durham University}
  \city{Durham}
  \country{United Kingdom}
}

\author{Mengxiao Zhang}
\email{mengxiao.zhang@auckland.ac.nz}
\orcid{0000-0002-6274-0384}
\affiliation{%
  \institution{School of Computer Science, The University of Auckland}
  \city{Auckland}
  \country{New Zealand}
}

\begin{abstract}
Kidney Exchange Programmes (KEPs) facilitate the exchange of kidneys, and larger pools of recipient-donor pairs tend to yield proportionally more transplants, leading to the proposal of international KEPs (IKEPs).  However, as studied by \citet{mincu2021ip}, practical limitations must be considered in IKEPs to ensure that countries remain willing to participate. Thus, we study IKEPs with country-specific parameters, represented by a tuple $\Gamma$, restricting the selected transplants to be feasible for the countries to conduct, e.g., imposing an upper limit on the number of consecutive exchanges within a country's borders. We provide a complete complexity dichotomy for the problem of finding a feasible (according to the constraints given by $\Gamma$) cycle packing with the maximum number of transplants, for every possible $\Gamma$. We also study the potential for countries to misreport their parameters to increase their allocation.  As manipulation can harm the total number of transplants, we propose a novel individually rational and incentive compatible mechanism $\mathcal{M}_{\text{order}}$.  We first give a theoretical approximation ratio for $\mathcal{M}_{\text{order}}$ in terms of the number of transplants, and show that the approximation ratio of $\Morder$ is asymptotically optimal.  We then use simulations which suggest that, in practice, the performance of $\mathcal{M}_{\text{order}}$ is significantly better than this worst-case ratio.
\end{abstract}

\begin{document}


\maketitle

\section{Introduction}\label{s-intro}

According to the most recent Global Burden of Disease study, in 2021 around 673.7 million people were affected by Chronic Kidney Disease (CKD), and 1.5 million deaths the same year were attributable to 
CKD. 
The most serious cases of CKD lead to End Stage Renal Disease (ESRD), for which the main treatments are either dialysis or transplantation. However, dialysis has major lifestyle implications for patients, and also leads to a patient life expectancy of only five years, as well as carrying a financial burden for healthcare systems~\cite{ASXITCKAL18}.  Transplantation, on the other hand, is more cost-effective~\cite{HSSGSCWKWSKI17}, 
and leads to better survival prospects for a recipient, especially from a living donor rather than a deceased donor. However, a patient with ESRD may have a willing but medically incompatible donor. \textit{Kidney Exchange Programmes} (KEPs) help to increase possibilities for living kidney donation by allowing recipients to swap their willing but incompatible donors in a cyclic fashion to obtain a compatible kidney~\cite{Rap86}. 

Many countries now have regional and national KEPs, with several examples of international KEPs (IKEPs) in Europe~\cite{ENCKEP19a}.  IKEPs involve multiple countries sharing their pools to provide increased opportunities for recipients, particularly those who are long waiting and/or difficult to match. European IKEPs include the Scandiatransplant Exchange Programme (STEP)~\cite{STEP23} 
that involves several Nordic countries, as well as collaboration among Italy, Portugal, and Spain~\cite{valentin2019international}.

STEP is based around a \emph{merged pool} where participating countries combine their own pools of donor-recipient pairs to form a single IKEP pool; and participating countries do not run their own national KEPs~\cite{STEP23}.  On the other hand, the collaboration between Italy, Portugal and Spain involves a \emph{consecutive pool} -- here, the participating countries first optimise on their national KEPs, and then the remaining pairs are combined to form the pool for the IKEP~\cite{valentin2019international}.  Maximising the number of identified transplants via a consecutive pool is individually rational (IR) for each participating country when considering an optimal solution that could be achieved nationally (throughout, we assume that ``optimal'' refers to maximising the number of transplants). However, this is not the case when optimising via a merged pool, as shown in Figure~\ref{fig:ikep}: any mechanism that maximises the number of transplants in the IKEP pool will select the cycle of length~$3$ comprising pairs $h_1$, $j_1$ and $j_2$. This is not IR for country $H$, which could have achieved two transplants (rather than the one following from the cycle of length~$3$) by optimising nationally.

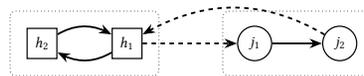
\begin{figure}[b!]
\centering
\resizebox{0.35\columnwidth}{!}{
\begin{tikzpicture}
\begin{scope}[ every node/.style={draw,rectangle,thick,minimum size=7mm}]
    \node  (x1) at (0,0) {$h_1$};
    \node  (x2) at (-2,0) {$h_2$};
\end{scope}
  \begin{scope}[every node/.style={draw,circle,thick,minimum size=8mm}]  
    \node (y1) at (3,0) {$j_1$};
    \node (y2) at (5,0) {$j_2$};
\end{scope}
\begin{scope}[>={Stealth[black]},
              every edge/.style={draw, very thick}]
    \path [->] (x1) edge[bend left] node {} (x2);
    \path [->] (x2) edge[bend left] node {} (x1);
    \path [->] (y1) edge node {} (y2);
    \path [->] (x1) edge[dashed] node {} (y1);
    \path [->] (y2) edge[dashed, bend right] node {} (x1);
\end{scope}
\draw[rounded corners, dotted] (-2.75,0.75) rectangle (0.75, -0.75) {};
\draw[rounded corners, dotted] (5.75,0.75) rectangle (2.25, -0.75) {};
    
\end{tikzpicture}}
\Description{}

\caption{An IKEP pool involving countries $H$ and $J$. 
Country $H$ and Country $J$'s pools 
 comprise pairs $h_1$ and $h_2$, and $j_1$ and $j_2$, respectively.  Full arcs are national arcs. Dashed arrows are international arcs.} 
\label{fig:ikep}
\end{figure}

Regional and national KEPs face many logistical challenges. Typically, the nephrectomies and transplants are scheduled on the same day. For this reason, there are upper bounds on the maximum lengths of exchange cycles, but these vary in different European national KEPs~\cite{ENCKEP19a}.  Another challenge with national KEPs is that large hospitals may be incentivised to hide their easiest-to-match pairs (finding cycles among them ``in house''), while reporting only their hardest-to-match pairs to the national pool~\cite{ashlagi2014free}. Figure~\ref{fig:ikep} also illustrates this if we assume $H$ and $J$ are two hospitals in a national KEP.

IKEPs must overcome several additional complexities. Firstly, countries participating in an IKEP may have different national cycle length limits. Secondly, they must agree on the maximum length~of an \emph{international cycle} that spans two or more countries.  Thirdly, within an international cycle, countries may wish to impose their own limits on \emph{segment size} (i.e., the maximum number of consecutive transplants within their country) and the \emph{number of segments} within their country (i.e., the maximum number of times the cycle enters that country).  
These restrictions may be made for logistical reasons: nephrectomies and transplants may need to be conducted simultaneously, necessitating available resources of a country, which may limit the segment size, and a cycle entering a country may involve increased distance, air transport and complex international import regulations, limiting the segment number. We remark that a country's restrictions on segment size and number only limit segments involving its own pairs.

Countries participating in an IKEP may behave strategically, not only by withholding pairs from the IKEP pool, but also via the declared values of their country-specific, logistically motivated parameters.  This is also problematic, as any benefit to a country behaving strategically will in general be detrimental to the overall social welfare of the IKEP, thus weakening the collaboration. We can assume that countries do not manipulate their own national cycle length limits, since these are usually public knowledge, especially in the case of an existing national KEP. Hence, such strategising would in practice be easily recognised.  However, segment size and segment number are country-specific parameters that are not public information prior to an IKEP being set up. For example, in Figure~\ref{fig:ikep}, if country~$H$ declares that its maximum segment size is~$0$, i.e., that it does not participate in international exchanges, then $H$ receives two transplants instead of only one (if a maximum solution is selected).

 \subsection{Existing Results}\label{s-related}
Multi-agent KEPs have been studied extensively in the kidney exchange literature, mostly from the perspective of incentivising individual hospitals to participate in a national KEP~\cite{ashlagi2014free,toulis2015design,blum2017opting,agarwal2019market}. In the IKEP context, several recent papers focused on the pressing issue of ensuring long-term stability of IKEPs based on merged pools where, as we saw, IR can no longer be guaranteed. As a solution, \citet{klimentova2021fairness} introduced a credit-based system, in which each country is allocated in each round a “fair” target number of kidney transplants. The difference between the allocated number of transplants for a country and its target number yield (positive or negative) credits for each country in the next round. Their model has been further investigated in~\cite{BiroGKPPV20,biro2019generalized,benedek2022computing,benedek2024computing} both from a theoretical and empirical perspective.
\citet{sun2021fair} gave bounds on the number of transplants that each country should receive, leading to an analogous notion of fairness to the credit-based framework described above. \citet{druzsin2024performance} investigated the extent to which countries' level of cooperation impacts the number of transplants each country receives. \citet{BSS24} introduced the ex-post setting, where countries can modify a proposed set of transplants.

\subsection{Our Focus}\label{s-focus}

The papers referenced in Section~\ref{s-related} do not consider any country-specific parameters, such as international segment number and international segment size. 
Indeed, existing IKEPs typically have small international cycle limits at present, which essentially obviates the need for country-specific parameters. 
However, our focus is on future developments in larger IKEPs, as in, e.g., the European Commission's envisioned pan-European IKEP \citep{PanEURO}, where larger international cycle limits may be anticipated, making country-specific parameters more meaningful.  Such parameters have been considered by \citet{mincu2021ip}, who presented a corresponding ILP model for finding an optimal solution that respects a given set of country-specific parameters, represented by a tuple $\Gamma$.  Here, a solution is a \emph{$\Gamma$-cycle packing}, which is a vertex-disjoint set of cycles satisfying the constraints imposed by $\Gamma$.

The central mechanism of \citet{mincu2021ip}, which we refer to as $\Mint$, returns a maximum (size) $\Gamma$-cycle packing for a round in an IKEP with a merged pool. 
In addition, \citet{mincu2021ip}  considered the national mechanism $\Mnat$ that selects a maximum size cycle packing from each country's national pool, as well as the consecutive mechanism $\Mcon$ which first applies $\Mnat$ on each of the national pools and then $\Mint$ on the international pool consisting of all the remaining patient-donor pairs. 

To evaluate their models, \citet{mincu2021ip} performed extensive simulations involving IKEPs with two countries. They considered neither the computational complexity of
finding a maximum $\Gamma$-cycle packing
nor desirable properties of mechanisms, such as IR and incentive compatibility (IC). However, as we illustrated with the example of Figure~\ref{fig:ikep}, ideally mechanisms for IKEPs should be IR and IC (with respect to segment size and segment number) simultaneously. Hence, our research questions are: 

\begin{itemize}
\item[--] \emph{What is the computational complexity of the problem of finding a $\Gamma$-cycle packing with the maximum number of transplants?}
\item[--] \emph{Is it possible to design IR and IC mechanisms for IKEPs with country-specific parameters?}
\item[--] \emph{What is the price of using such mechanisms in terms of the maximum total number of transplants?}
\end{itemize}

 Two additional country-specific parameters were proposed by \citet{mincu2021ip} that we chose not to model. 
 The first allows a country to limit the number of countries involved in an international exchange involving their pairs; we do not include this parameter as it does not allow a country to reduce its logistical burden. The second parameter allows countries to limit the total number of their recipients that can be selected in any international cycle. 
We do not include this parameter either, as
the product of a country's segment size and segment number already provides an upper bound for this additional parameter. 
\citet{mincu2021ip} also allow for non-directed (altruistic) donors who can trigger chains of transplants instead of cycles. 
Non-directed donors are not permitted in several countries, such as  France~\citep{COMBE2022270} and Poland~\citep{LichodziejewskaNiemierko2021}. However, in several other countries, such as the Netherlands and the UK, non-directed donors are allowed~\citep{biro2019building}, and we therefore consider these in future work.

\subsection{Our New Results}\label{s-new}

To answer our research questions, we first need to analyse the computational complexity of finding a maximum $\Gamma$-cycle packing. We do this by fixing the set $\Gamma$ of country-specific parameters as not part of the input, as $\Gamma$ is expected to be stable for an IKEP. It follows immediately from existing work from \citet{abraham2007origin} that for some specific fixed $\Gamma$,  the problem {\sc Max $\Gamma$-Cycle Packing} of computing a maximum $\Gamma$-cycle packing is \NP-hard.  In Section~\ref{sec:complexityMax} we show a {\it complete} complexity dichotomy of  {\sc Max $\Gamma$-Cycle Packing}, so taking into account {\it all} possible values of $\Gamma$. Our dichotomy shows that for most sets $\Gamma$,  {\sc Max $\Gamma$-cycle Packing} is \NP-hard even if only two countries are involved, making it hard for countries to manipulate. 

Despite the above, pool sizes in practice can allow optimal solutions to be computed within seconds or minutes, so the issue of manipulation should certainly not be underestimated.  
We first show in Section~\ref{sec:mech} that $\Mint$ is not IR and also not IC (with respect to segment size and segment number). Thus, it is natural to design a mechanism that is both IR and IC.  The mechanism $\Mnat$ returns a maximum national $\Gamma$-cycle packing, which is the union of maximum $\Gamma$-cycle packings for each of the national pools only. It is readily seen that $\Mnat$ is both IR and IC. However, unfortunately, $\Mnat$ fails to allow the benefits of international collaboration.

We present a novel mechanism in Section~\ref{sec:mech} that is both IC and IR and that does not have the drawback of $\Mnat$. Our new mechanism, $\Morder$, first behaves like $\Mnat$ and computes a maximum national $\Gamma$-cycle packing. Afterwards, it
considers a random ordering of all international cycles not exceeding the international cycle length limit and replaces those that are not $\Gamma$-cycles by smaller ones, which are again considered in some random ordering. We show that the running time of $\Morder$ compares favourably with $\Mint$. Afterwards, we establish that $\Morder$ is indeed IR and IC.  We give a tight upper bound of $\max\{\cint,d^*\}$ on the approximation ratio of $\Morder$, where $\cint$ is the maximum length of any international $\Gamma$-cycle, and $d^*$ is the maximum number of international $\Gamma$-cycles involving any vertex of an international $\Gamma$-cycle.  We also prove that any IR mechanism cannot achieve an approximation ratio less than $\cint$, and any IC mechanism cannot achieve an approximation ratio less than $d^*$. We show how this implies the approximation ratio of $\Morder$ is asymptotically optimal. 

Our above results provide theoretical answers to all three of our research questions. In addition, we also wish to provide some empirical answers. Therefore, in Section~\ref{sec:simulations}, we perform a simulation study that compares the performance of $\Morder$ with $\Mint$ and $\Mnat$, and also with one further mechanism $\Mcon$, which constructs a maximum $\Gamma$-cycle packing on a consecutive pool. We do the latter comparison, as $\Morder$ can also be used for consecutive pools.
Our simulation results suggest that the performance of $\Morder$ greatly exceeds its worst-case upper bound, producing on average $72\%$ of the number of transplants found by (the non-IC and non-IR mechanism) $\Mint$, $91\%$ of the number of transplants found by (the IR but non-IC mechanism) consecutive mechanism $\Mcon$, and $84\%$ more transplants than the number found by $\Mnat$. This gives evidence of the potential benefit of joining an IKEP, and suggests that the price of ensuring IC and IR via our consecutive mechanism $\Morder$ is very modest in comparison with the IR consecutive mechanism $\Mcon$.  

Finally, Section \ref{sec:conc} presents our conclusions and suggestions for future work.

\section{Preliminaries}\label{sec:model}

The merged pool in some round of an IKEP is modelled by {\it compatibility graph}, which is a directed graph $G=(V,A)$ without self-loops or multi-arcs. A vertex of $V$ represents a recipient-donor pair, and an arc $(u,v)$ indicates that the donor of pair $u$ is compatible with the patient of pair $v$ (see also Figure~\ref{fig:ikep}). Below we give further graph terminology; see Section~\ref{sec:mech} for mechanism design terminology.

A {\it (directed) cycle} is a directed graph $C=(V,A)$ with vertices $v_1,\ldots,v_s$ for some integer $s\geq 2$, such that the set $A$ consists of arcs $(v_i,v_{i+1})$ for $i=1,\ldots,s$ where $v_{s+1}=v_1$; we also write $C=\langle v_1,\ldots,v_s \rangle$. The {\it length} of $C$ is its number of arcs.
A {\it cycle packing} of a directed graph $G$ is a set~$\C$ of cycles in $G$ that are pairwise vertex-disjoint. The {\it size} $\Vert \C\Vert$ of~$\C$ is the sum of the lengths of the cycles of~$\C$, or equivalently, the total number of arcs belonging to the cycles in $\C$. A cycle packing $\C$ of a graph $G$ is {\it perfect} if every vertex of $G$ belongs to some cycle in $\C$.  In our context, a cycle packing prescribes a set of kidney transplants in a compatibility graph $G$, and in perfect cycle packing, every patient will receive a kidney. 
Given $G=(V,A)$ and a subset $S\subseteq V$, we let $G[S]$ denote the subgraph of~$G$ {\it induced by} $S$, that is, $G[S]$ is the subgraph of $G$ obtained after deleting the vertices $V\setminus S$ from $G$. 

For an integer $n\geq 1$, let $\V=(V_1,\ldots,V_n)$ be a partition of $V$, that is, every $V_i$ is non-empty, and every two $V_i$ and $V_j$ with $i\neq j$ are disjoint. In our context, $\N=\{1,\ldots,n\}$ is the {\it set of countries} and a set~$V_i$ is the set of recipient-donor pairs belonging to country~$i$. 
We say that the pair $(G,\V)$ is an {\it $n$-partitioned} graph.
An arc of $G$ is {\it international} if its two end-vertices belong to different sets of $\V$; else it is a {\it national} arc.
A cycle $C$ of $G$ is {\it international} if it has at least one international arc; else $C$ is {\it national}.
A path $P$ in an international cycle $C$ is a {\it $V_i$-segment} of $C$ for some $i\in \N$ if $P$ only contains vertices from $V_i$; the {\it size} of $P$ is the number of vertices $|V(P)|$ of $P$. 
With $(G,\V)$, we associate the following country-specific parameters:
\begin{itemize}
\item the \emph{number of countries} $n$;
\item a {\it national cycle (length) limit vector} $\natv$, where for $i\in\N$, $\natv_i$ denotes the maximum national cycle length in $G[V_i]$;
\item an {\it international cycle (length) limit} $\intv$ is the maximum length of an international cycle that is collectively agreed;
\item an {\it international segment size vector} $\segv$, where for $i\in \N$,  
$\segv_i$ denotes the maximum size of a $V_i$-segment in an international cycle in $G$; and
\item an {\it international segment number vector}  $\sv$, where for $i\in \N$,  
$\sv_i$ denotes the maximum number of $V_i$-segments in an international cycle in $G$.
\end{itemize}

We say that $\Gamma= (n, \intv, \natv,\segv, \sv)$ is a {\it set of country-specific parameters} for $n$-partitioned graphs $(G,\V)$.
A national cycle $C$, say with $V(C)\subseteq V_i$ for some $i\in \N$, is a {\it $\Gamma$-cycle} if $C$ has length at most $\natv_i$.
An international cycle $C$ is a {\it $\Gamma$-cycle} if: {\bf(i)} $C$ has length at most $\intv$; {\bf(ii)} for every $i\in \N$, every $V_i$-segment of $C$ has size at most $\segv_i$; and {\bf(iii)} the number of $V_i$-segments of $C$ is at most $\sv_i$.  
A cycle packing $\mathcal{C}$ of $G$ is a {\it $\Gamma$-cycle packing} if every cycle of $\mathcal{C}$ is a $\Gamma$-cycle.
As self-loops are not allowed, cycles cannot have a length of~$1$. Hence, we assume $\intv\neq 1$ and $\natv_i \neq 1$ for every $i\in \N$.

In the example of Figure~\ref{fig:ikep}, we have a $2$-partitioned graph $(G,\V)$ with $\N=\{1,2\}$, where country~$1$ is $H$; country~$2$ is $J$; $V_1=\{h_1,h_2\}$ and $V_2=\{j_1,j_2\}$. Let $\intv=3$, $\segv_1=1$, $\segv_2=2$ and $\sv_1=\sv_2=1$. Now, $C=\langle h_1,j_1,j_2 \rangle$ is an international $\Gamma$-cycle. If $\natv_2=2$, then $D=\langle h_1,h_2 \rangle$ is a national $\Gamma$-cycle packing. There are three $\Gamma$-cycle packings: $\emptyset$, $\{C\}$ and $\{D\}$, none of which is perfect. If country~$1$ sets $\segv_1'=0$, this yields a new set $\Gamma'$ of country-specific parameters. There are only two $\Gamma'$-cycle packings: $\emptyset$ and $\{D\}$. 
From now on, we assume that every country is willing to participate in an international cycle; that is, $\segv_i\geq 1$ and $\sv_i\geq 1$ for every $i\in \N$. Moreover, we write $\intv=\infty$ if we do not put any bound on $\intv$. We also allow $\natv_i$, $\segv_i$ and $\sv_i$ to be $\infty$. 

For an $n$-partitioned graph $(G,\V)$ and set of country-specific parameters, we say that $I= \langle G, \V, \Gamma\rangle $ is an IKEP {\it instance} (corresponding to a round in an IKEP with set of country-specific parameters~$\Gamma$).

\section{The Computational Complexity of Finding Maximum Cycle Packings}\label{sec:complexityMax}

Let $I= \langle G, \V, \Gamma\rangle$ be an IKEP instance.
A $\Gamma$-cycle packing $\C$ is {\it maximum} if $G$ contains no $\Gamma$-cycle packing~${\C}'$ with $\Vert\mathcal{C'} \Vert$ > $\Vert\mathcal{C} \Vert$.
We recall that the mechanism $\Mint$ of \citet{mincu2021ip} returns a maximum $\Gamma$-cycle packing. In this section, we consider the complexity of this optimisation problem. As mentioned in Section~\ref{s-intro}, we do this in a fine-grained manner by fixing $\Gamma$, so $\Gamma$ is not part of the problem input. This yields the problem {\sc Max $\Gamma$-Cycle Packing}, which is to determine a maximum $\Gamma$-cycle packing of an $n$-partitioned graph $(G,\V)$. 
We determine the complexity of {\sc Max $\Gamma$-Cycle Packing} for {\it all} values of $\Gamma$.
All our hardness results hold even for the variant {\sc Perfect $\Gamma$-Cycle Packing}, which is to decide if an $n$-partitioned graph $(G,\V)$ has a perfect $\Gamma$-cycle packing.

\citet{abraham2007origin} established a complete dichotomy for the case of one country (which corresponds to the KEP setting).

\begin{theorem}[\citeauthor{abraham2007origin}, \cite{abraham2007origin}]\label{t-dicho}
For a set of country-specific parameters $\Gamma$ with $n=1$, {\sc Max $\Gamma$-Cycle Packing} is polynomial-time solvable if $\emph{\natv}\in \{0,2,\infty\}$, and else even {\sc Perfect $\Gamma$-Cycle Packing} is \NP-complete.
\end{theorem}
 
Our hardness results are inspired by the gadget used in the proof of Theorem~\ref{t-dicho} for the case where $n=1$ and $\natv=3$. However, we need more involved arguments, as we illustrate by highlighting one of our hardness results. We refer to Appendix~\ref{app:Lems} for its full proof. For $k\geq 0$, we write $\mathbb{N}_{\geq k}=\{k,k+1,\ldots\}$.

\begin{restatable}{lemma}{sampleLem}\label{lem:l=inftyk>=0c=1s>=2n=2}
For a set of country-specific parameters $\Gamma= (n, \emph{\intv}, \emph{\natv},\emph{\segv}, \emph{\sv})$ with 
$n \geq 2$,  
$\emph{\intv}=\infty$, 
$\emph{\natv}_i \in \mathbb{N}_{\geq 2} \cup \{0,\infty\}$ for all $i\in\N$,
$\emph{\sv}_i \in \mathbb{N}_{\geq 2} \cup \{\infty\}$ for all $i\in \N$,
$\emph{\segv}_j\in \mathbb{N}_{\geq 2}$ for some $j\in \N$, and
$\emph{\segv}_i\in \mathbb{N}_{\geq 1} \cup \{\infty\}$ for all $i\in \N\setminus \{j\}$,
{\sc Perfect $\Gamma$-Cycle Packing} is \NP-complete.
\end{restatable}
\begin{proof}[Proof Sketch.]
Let $\Gamma= (n, \intv, \natv,\segv, \sv)$ with 
$n \geq 2$,  
$\intv=\infty$, 
$\natv_i \in \mathbb{N}_{\geq 2} \cup \{0,\infty\}$ for all $i\in\N$,
$\sv_i \in \mathbb{N}_{\geq 1} \cup \{\infty\}$ for all $i\in \N$,
$\segv_j\in \mathbb{N}_{\geq 2}$ for some $j\in \N$, and
$\segv_i\in \mathbb{N}_{\geq 1} \cup \{\infty\}$ for all $i\in \N\setminus \{j\}$.  We assume $\segv_2\in \mathbb{N}_{\geq 2}$ and $m = \segv_2 + 1$.

It is readily seen that \textsc{Perfect $\Gamma$-Cycle Packing}  belongs to \NP. To show \NP-hardness, we reduce from the \NP-complete problem {\sc Perfect $4$-Dimensional Matching}.
The latter problem takes as input
four sets of elements $W$, $X$, $Y$, $Z$, and a set of quadruples $T\subseteq W\times X \times Y \times Z$. 
The question is whether $T$ contains a {\it perfect matching}, which is a subset $M \subseteq T$ such that each $u\in W \cup X \cup Y \cup Z$ is contained in exactly one element of $M$.
A straightforward reduction from the \NP-complete problem  {\sc Perfect $3$-Dimensional Matching} \citep{karp1972} shows that {\sc Perfect $4$-Dimensional Matching} is also \NP-complete. 

        \begin{figure*}[t]
    \centering
\resizebox{0.7\columnwidth}{!}{
\begin{tikzpicture}[triangle/.style = {regular polygon, regular polygon sides=3 }]

\begin{scope}[every node/.style={circle,thick,draw,minimum size=0.9cm}]
    \node[regular polygon, thick, regular polygon sides=3, draw, inner sep=-0.1cm] (a1) at (-5,3.5) {$w_{i,1}^{m-1}$};
    \node (w) at (-5,0) {$w$};
    \node (a2) at (-5,10) {$w_{i,1}^0$};

    \node (a3) at (0,3) {$x_{i,1}^{m+1}$};
    \node[regular polygon, thick, regular polygon sides=3, draw, inner sep=0.005cm] (x) at (0,0) {$x$};
    \node[regular polygon, thick, regular polygon sides=3, draw, inner sep=0.005cm] (a4) at (0,10) {$x_{i,1}^2$};

    \node[regular polygon, thick, regular polygon sides=3, draw,
     inner sep=-0.1cm] (a5) at (5,3.5) {$y_{i,1}^{m-1}$};
    \node (y) at (5,0) {$y$};
    \node (a6) at (5,10) {$y_{i,1}^0$};

    \node (a7) at (10.5,3) {$z_{i,1}^{m+1}$};
    \node[regular polygon, thick, regular polygon sides=3, draw, inner sep=0.005cm] (z) at (10.5,0) {$z$};
    
    \node[regular polygon, thick, regular polygon sides=3, draw, inner sep=0.005cm] (a8) at (10.5,10) {$z_{i, 1}^2$};

    \node[regular polygon, thick, regular polygon sides=3, draw, inner sep=0.005cm] (z1) at (7.25,10) {$z_{1,i}^1$};
    \node[regular polygon, thick, regular polygon sides=3, draw, inner sep=0.005cm] (zs) at (11.5,5) {$z_{i,1}^{m}$};

    \node[regular polygon, thick, regular polygon sides=3, draw,
     inner sep=0.005cm] (y1) at (6,8) {$y_{i,1}^1$};
    \node[regular polygon, thick, regular polygon sides=3, draw, inner sep=0.005cm] (ys) at (6,1.5) {$y_{i,1}^{m}$};
    \node[regular polygon, thick, regular polygon sides=3, draw,
     inner sep=0.005cm] (x1) at (1,5) {$x_{i,1}^{m}$};
    \node[regular polygon, thick, regular polygon sides=3, draw, inner sep=0.005cm] (xs) at (-1.75,10) {$x_{i, 1}^{1}$};

    \node[regular polygon,  regular polygon sides=3, draw,
     inner sep=-0.09cm] (wm-2) at (-4,5) {$w_{i, 1}^{m-2}$};
     \node[regular polygon, thick, regular polygon sides=3, draw,
     inner sep=0.005cm] (x3) at (1,8) {$x_{i, 1}^{3}$};
     \node[regular polygon, thick, regular polygon sides=3, draw,
     inner sep=-0.1cm] (ym-2) at (6,5) {$y_{i,1}^{m-2}$};
     \node[regular polygon, thick, regular polygon sides=3, draw,
     inner sep=0.005cm] (z3) at (11.5,8) {$z_{i,1}^{3}$};
    
    \node[regular polygon, thick, regular polygon sides=3, draw,
     inner sep=0.005cm] (w1) at (-4,8) {$w_{i, 1}^1$};
    \node[regular polygon, thick, regular polygon sides=3, draw, inner sep=0.005cm] (ws) at (-4,1.5) {$w_{i,1}^{m}$};

    \node (k1) at (-3,9) {$w_{i,3}^{1}$};
    \node[regular polygon, thick, regular polygon sides=3, draw, inner sep=-0.05cm] (k2) at (-3,7
    ) {$w_{i,2}^1$};

    \node (k3) at (-3,2.5) {$w_{i,3}^{m}$};
    \node[regular polygon, thick, regular polygon sides=3, draw, inner sep=-0.005cm] (k4) at (-3,0.5
    ) {$w_{i,2}^m$};

    \node (k5) at (-2.75,11.1) {$x_{i,3}^{1}$};
    \node[regular polygon, thick, regular polygon sides=3, draw, inner sep=-0.05cm] (k6) at (-0.5,11.1) {$x_{i,2}^{1}$};

    \node (k7) at (2,6) {$x_{i,3}^{m}$};
    \node[regular polygon, thick, regular polygon sides=3, draw, inner sep=0.005cm] (k8) at (2,4
    ) {$x_{i,2}^{m}$};

    \node (k9) at (7,9) {$y_{i,3}^{1}$};
    \node[regular polygon, thick, regular polygon sides=3, draw, inner sep=-0.05cm] (k10) at (7,7
    ) {$y_{i,2}^{1}$};

    \node (k11) at (7,2.5) {$y_{i,3}^{m}$};
    \node[regular polygon, thick, regular polygon sides=3, draw, inner sep=0.05cm] (k12) at (7,0.5
    ) {$y_{i,2}^{m}$};

    \node (k13) at (6,11) {$z_{i,3}^{1}$};
    \node[regular polygon, thick, regular polygon sides=3, draw, inner sep=-0.05cm] (k14) at (8.25,11) {$z_{i,2}^{1}$};

    \node (k15) at (12.5,6) {$z_{i,3}^{m}$};
    \node[regular polygon, thick, regular polygon sides=3, draw, inner sep=0.005cm] (k16) at (12.5,4) {$z_{i,2}^{m}$};
    
    \node (k17) at (12.5,9) {$z_{i,3}^{3}$};
    \node[regular polygon, thick, regular polygon sides=3, draw, inner sep=-0.05cm] (k18) at (12.5,7) {$z_{i,2}^{3}$};

    \node[inner sep=-0.05cm] (k19) at (7,6) {$y_{i,3}^{m-2}$};
    \node[regular polygon, thick, regular polygon sides=3, draw, inner sep=-0.05cm] (k20) at (7,4) {$y_{i,2}^{m-2}$};

    \node (k21) at (2,9) {$x_{i,3}^{3}$};
    \node[regular polygon, thick, regular polygon sides=3, draw, inner sep=-0.05cm] (k22) at (2,7) {$x_{i,2}^{3}$};

    \node[inner sep=-0.03cm] (k23) at (-2.5,6) {$w_{i,3}^{m-2}$};
    \node[regular polygon, thick, regular polygon sides=3, draw, inner sep=-0.09cm] (k24) at (-2.5,3.75) {$w_{i,2}^{m-2}$};

\end{scope}
 \node (zdot) at (11.5,6.5) {$\cdots$};
  \node (ydot) at (6,6.5) {$\cdots$};
   \node (xdot) at (1,6.5) {$\cdots$};
    \node (wdot) at (-4,6.5) {$\cdots$};

 \begin{scope}[>={Stealth[black]},
               every edge/.style={draw,very thick}]

    \path [->] (w) edge[bend left] node {} (a1);
    \path [->] (ws) edge[bend left] node {} (w);
    \path [->] (a1) edge[bend left] node {} (a2);
    \path [->] (a2) edge[bend left=10] node {} (w1);
    \path [-] (w1) edge node {} (wdot);
    \path [->] (wdot) edge node {} (wm-2);
    \path [->] (wm-2) edge[bend left=10] node {} (a1);
    \path [->] (a1) edge[bend left=10] node {} (ws);

    \path [->] (x) edge[bend left] node {} (a3);
    \path [->] (a3) edge[bend left] node {} (x);
    \path [->] (a3) edge[bend left] node {} (a4);
    \path [->] (a4) edge[bend left=10] node {} (x3);
    \path [-] (x3) edge node {} (xdot);

    \path [->] (x1) edge[bend left=10] node {} (a3);
    \path [->] (xdot) edge node {} (x1);
    \path [->] (a2) edge node {} (xs);
    \path [->] (xdot) edge node {} (x1);
    \path [->] (xs) edge node {} (a4);

    \path [->] (y) edge[bend left] node {} (a5);
    \path [->] (ys) edge[bend left] node {} (y);
    \path [->] (a5) edge[bend left] node {} (a6);
    \path [->] (a6) edge[bend left=10] node {} (y1);

    \path [-] (y1) edge node {} (ydot);
    \path [->] (ydot) edge node {} (ym-2);
    \path [->] (ym-2) edge[bend left=10] node {} (a5);
    
    \path [->] (a5) edge[bend left=10] node {} (ys);

    \path [->] (z) edge[bend left] node {} (a7);
    \path [->] (a7) edge[bend left] node {} (z);
    \path [->] (a7) edge[bend left] node {} (a8);

    \path [->] (a8) edge[bend left=10] node {} (z3);
    \path [-] (z3) edge node {} (zdot);
    \path [->] (zdot) edge node {} (zs);
    \path [->] (zs) edge[bend left=10] node {} (a7);

    \path [->] (a4) edge[] node {} (a6);
    \path [->] (a6) edge[] node {} (z1);
    \path [->] (z1) edge[] node {} (a8);
    \path [->] (a8) edge[bend right=45] node {} (a2);

    \path [->] (k1) edge[] node {} (w1);
    \path [->] (w1) edge[] node {} (k2);
    \path [->] (k1) edge[bend left=15] node {} (k2);
    \path [->] (k2) edge[bend right=15] node {} (k1);

    \path [->] (k3) edge[] node {} (ws);
    \path [->] (ws) edge[] node {} (k4);
    \path [->] (k3) edge[bend left=15] node {} (k4);
    \path [->] (k4) edge[bend right=15] node {} (k3);

    \path [->] (k5) edge[] node {} (xs);
    \path [->] (xs) edge[] node {} (k6);
    \path [->] (k5) edge[bend left=15] node {} (k6);
    \path [->] (k6) edge[bend right=15] node {} (k5);

    \path [->] (k7) edge[] node {} (x1);
    \path [->] (x1) edge[] node {} (k8);
    \path [->] (k7) edge[bend left=15] node {} (k8);
    \path [->] (k8) edge[bend right=15] node {} (k7);

    \path [->] (k9) edge[] node {} (y1);
    \path [->] (y1) edge[] node {} (k10);
    \path [->] (k9) edge[bend left=15] node {} (k10);
    \path [->] (k10) edge[bend right=15] node {} (k9);

    \path [->] (k11) edge[] node {} (ys);
    \path [->] (ys) edge[] node {} (k12);
    \path [->] (k11) edge[bend left=15] node {} (k12);
    \path [->] (k12) edge[bend right=15] node {} (k11);

    \path [->] (k13) edge[] node {} (z1);
    \path [->] (z1) edge[] node {} (k14);
    \path [->] (k13) edge[bend left=15] node {} (k14);
    \path [->] (k14) edge[bend right=15] node {} (k13);

    \path [->] (k15) edge[] node {} (zs);
    \path [->] (zs) edge[] node {} (k16);
    \path [->] (k15) edge[bend left=15] node {} (k16);
    \path [->] (k16) edge[bend right=15] node {} (k15);

    \path [->] (k17) edge[bend left=15] node {} (k18);
    \path [->] (k18) edge[bend right=15] node {} (k17);

    \path [->] (k19) edge[bend left=15] node {} (k20);
    \path [->] (k20) edge[bend right=15] node {} (k19);

    \path [->] (k21) edge[bend left=15] node {} (k22);
    \path [->] (k22) edge[bend right=15] node {} (k21);

    \path [->] (k23) edge[bend left=15] node {} (k24);
    \path [->] (k24) edge[bend right=15] node {} (k23);
    \path [->] (k23) edge node {} (wm-2);
    \path [->] (wm-2) edge node {} (k24);

    \path [->] (k19) edge node {} (ym-2);
    \path [->] (ym-2) edge node {} (k20);

    \path [->] (k21) edge node {} (x3);
    \path [->] (x3) edge node {} (k22);

    \path [->] (k17) edge node {} (z3);
    \path [->] (z3) edge node {} (k18);

\end{scope}
\end{tikzpicture}}
\Description{}
    \caption{The gadget that is used in the proof of Lemma~\ref{lem:l=inftyk>=0c=1s>=2n=2}.}
    \label{fig:l=inftyk>=0c=1s>=2n=2*}
\end{figure*}
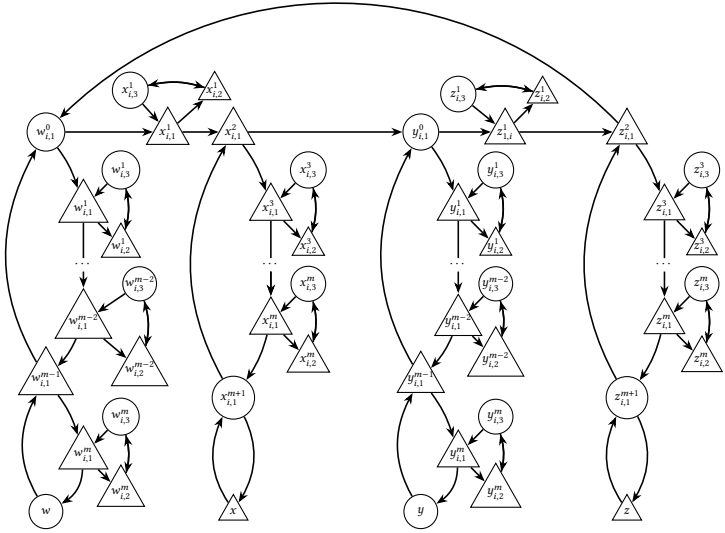

Let $(W,X,Y,Z,T)$ be an instance of {\sc Perfect $4$-Dimensional Matching}.
We construct an instance $(G,\V)$
of \textsc{Perfect $\Gamma$-Cycle Packing} as follows. 
    We first create a vertex for each $u \in W \cup X \cup Y \cup Z$. We create a gadget for every $t_i = (w, x, y, z) \in T$, and then all remaining vertices and edges are given per gadget, as depicted in Figure~\ref{fig:l=inftyk>=0c=1s>=2n=2*}.  
  We set $$V_1 = \{w, w_{i,1}^0, x_{i,1}^{m+1}, y, y_{i,1}^0, z_{i,1}^{m+1}, w_{i,3}^{q}, x_{i,3}^{p}, y_{i,3}^{q}, z_{i,3}^{p}\mid  t_i \in T, q \in \{1,\ldots, m-2, m\}, p \in \{1\}\cup\{3, \ldots, m\}\},$$ 
  and $V_2 = V \setminus V_1$. For every $i\in \N\setminus \{1,2\}$, we introduce an isolated vertex (which we subsequently ignore). Countries~$1$ and~$2$ are represented by circular and triangular vertices in Figure~\ref{fig:l=inftyk>=0c=1s>=2n=2*}, respectively. 

We first prove that every $\Gamma$-cycle~$C$ is contained within a single gadget.  
For a contradiction, suppose $C$ spans over more than one gadget. Hence, $C$ must have a path that enters and exits some gadget $t_i$. We may assume without loss of generality that $C$ enters via $w$ and exits via $x$. This implies that $C$ contains the path 
$
[w, w_{i,1}^{m-1}, w_{i,1}^0, x_{i,1}^1, x_{i,1}^2, \cdots, x_{i,1}^m, x_{i,1}^{m+1}, x].
$
This path contains the $V_2$-segment $[x_{i,1}^1,\ldots,  x_{i,1}^m]$ of size~$m$. As $m = \segv_2 + 1$, we find that $C$ is not a $\Gamma$-cycle. Similar reasoning can be made when entering or exiting a gadget via different vertices. Thus, every $\Gamma$-cycle is contained within a gadget.

We claim that $G$ has a perfect $\Gamma$-cycle packing if and only if $(W,X,Y,Z,T)$ has a perfect matching~$M$. 

First suppose that $(W,X,Y,Z,T)$ has a perfect matching $ M $. We construct a perfect $\Gamma$-cycle packing by considering each $ t_i \in T $ one by one. If $ t_i \in M $, we select the following cycles and add them to $\C$: 
\[\begin{array}{l}
\langle w, w_{i,1}^{m-1}, w_{i,1}^m \rangle,  
\langle x, x_{i,1}^{m+1} \rangle, 
\langle y, y_{i,1}^{m-1}, y_{i,1}^m \rangle, 
\langle z, z_{i,1}^{m+1} \rangle, 
\langle w_{i,1}^0, x_{i,1}^1, x_{i,1}^2, y_{i,1}^0, z_{i,1}^1, z_{i,1}^2 \rangle, 
\langle w_{i,3}^{m}, w_{i,2}^{m} \rangle, \\
\langle w_{i,1}^q, w_{i,2}^{q}, w_{i,3}^{q} \rangle, 
\langle x_{i,3}^{1}, x_{i,2}^{1} \rangle, 
\langle x_{i,1}^p, x_{i,2}^{p},
x_{i,3}^{p} \rangle,
\langle y_{i,3}^{m}, y_{i,2}^{m} \rangle, 
\langle y_{i,1}^q, y_{i,2}^{q}, y_{i,3}^{q} \rangle 
\langle z_{i,3}^{1}, z_{i,2}^{1} \rangle, 
\langle z_{i,1}^p, z_{i,2}^{p},
z_{i,3}^{p} \rangle, 
\end{array}
\]
for $ q \in [1, m-2] \text{ and } p \in [3,  m]$. Thus, all the nodes within the gadget are covered by some $\Gamma$-cycle.  If $t_i\notin M$, then the following cycles are added to $\C$:
\[
\begin{array}{c}
\langle w_{i,1}^0, w_{i,1}^1, \cdots, w_{i,1}^{m-1} \rangle, 
\langle x_{i,1}^2, \cdots, x_{i,1}^m, x_{i,1}^{m+1} \rangle, 
\langle y_{i,1}^0, y_{i,1}^1, \cdots, y_{i,1}^{m-1} \rangle,  
\langle z_{i,1}^2, \cdots, z_{i,1}^m, z_{i,1}^{m+1} \rangle, 
\langle w_{i,3}^{q}, w_{i,2}^q \rangle, \\
\langle w_{i,1}^m, w_{i,2}^{m}, w_{i,3}^{m} \rangle, 
\langle x_{i,3}^{p}, x_{i,2}^{p} \rangle, 
\langle x_{i,1}^1, x_{i,2}^{1},
x_{i,3}^{1} \rangle,
\langle y_{i,3}^{q}, y_{i,2}^{q} \rangle, 
\langle y_{i,1}^m y_{i,2}^{m}, y_{i,3}^{m} \rangle, 
\langle z_{i,3}^{p}, z_{i,2}^{p} \rangle, 
\langle z_{i,1}^1, z_{i,2}^{1},
z_{i,3}^{1} \rangle,
\end{array}
\]
for $ q \in [1, m-2] \text{ and } p \in [3,  m]$. In this second case, all of the gadget's vertices have been selected in some $\Gamma$-cycle, apart from $w$, $x$, $y$, and $z$.  As $M$ is a perfect matching, every vertex $u \in W\cup X \cup Y\cup Z$ will be covered by some gadget exactly once. Therefore, $\C$ is a perfect $\Gamma$-cycle packing of $G$.

Conversely suppose $G$ has a perfect $\Gamma$-cycle packing $\C$. We construct a perfect matching $M$ of $(W,X,Y,Z,T)$. We defer the full proof to Appendix~\ref{app:Lems}. There, we show that if $\C$ contains the cycle $\langle w, w_{i,1}^{m-1}, w_{i,1}^m \rangle$, then  $\C$ covers the vertices $w,x,y$, and $z$ via the gadget corresponding to $t_i$. Moreover, if $\langle w, w_{i,1}^{m-1}, w_{i,1}^m \rangle\notin \C$, we show that $\C$ cannot cover the vertices $w,x,y$, and $z$ with cycles within the gadget corresponding to $t_i$. Hence, a perfect matching $M$ of $(W,X,Y,Z,T)$ can be constructed by adding those quadruples~$t_i$ to $M$ whose four elements $w,x,y,z$ are all selected by the gadget corresponding to $t_i$. 
\end{proof}
 
We now state our dichotomy; see Appendix~\ref{app:complexity} for its proof. First, we recall from Section~\ref{sec:model} that we assume that for a set of country-specific parameters~$\Gamma$ we have $\intv\neq 1$; $\natv \neq 1$; and $\segv_i\geq 1$ and $\sv_i\geq 1$ for all $i\in N$.
We also note that the polynomial cases follow, in the same way as the polynomial cases in Theorem~\ref{t-dicho}, from straightforward reductions to either a maximum matching problem in a general graph or a maximum weight perfect matching problem in a bipartite graph. The hard cases are proven in a series of lemmas using similar arguments as in the proof of Lemma~\ref{lem:l=inftyk>=0c=1s>=2n=2}.

\begin{restatable}{theorem}{dicho}\label{thm:dico}
For a set of country-specific parameters $\Gamma$, {\sc Max $\Gamma$-Cycle Packing} is polynomial-time solvable if $\Gamma$ belongs to one of the following seven cases:
\begin{enumerate}
\item $n=1$ and $\emph{\natv}_1\in \{0,2,\infty\}$ (Theorem~\ref{t-dicho}~\cite{abraham2007origin});
\item $n \geq 2$, $\emph{\intv}=0$, $\emph{\natv}_i\in \{0,2,\infty\}$  for all $i \in \N$;
\item $n \geq 2$,  $\emph{\intv}=2$, and  $\emph{\natv}_i \in \{0,2\}$ for all $i \in \N $;
\item $n=2$, $\emph{\intv}=3$,  $\emph{\segv}=(1,1)$ and $\emph{\natv}_i \in \{0,2\}$ for all $i \in \N$; 
\item $n=2$, $\emph{\intv} \in \mathbb{N}_{\geq 4}\cup\{\infty\}$,   $\emph{\segv}=(1,1)$ and there is some $j \in\N$ such that $\emph{\sv}_j=1$ and $\emph{\natv}_i \in \{0,2\}$ for all $i \in \N$; 
\item $n \geq 2$,  $\emph{\intv}=\infty$,  $\emph{\natv}=\{0\}^n$, $\emph{\segv}=\{1\}^n$, $\emph{\sv}=\{\infty\}^n$; or
\item $n \geq 2$,   $\emph{\intv} =\infty$,  $\emph{\natv}=\emph{\segv}=\emph{\sv}=\{\infty\}^n$;
\end{enumerate}
and else even {\sc Perfect $\Gamma$-Cycle Packing} is \NP-complete.
\end{restatable}

Unlike Theorem~\ref{t-dicho}, Theorem~\ref{thm:dico} shows  even if, e.g.,
$\Gamma= (2, \intv, \natv,\segv, \sv)$ with 
$\intv=\segv_1=\segv_2=\sv_1=\sv_2=\infty$ and $\natv=(2,\infty)$, {\sc Perfect Cycle Packing} is already \NP-complete. 
Thus, combining the parameter values $2$ and $\infty$ for which
{\sc Max $\Gamma$-Cycle Packing} is polynomial-time solvable for $n=1$, i.e., in the KEP setting (see Theorem~\ref{t-dicho}) can lead to \NP-complete cases in the IKEP setting, already for $n=2$.

We now briefly return to the three mechanisms $\Mint$, $\Mnat$, $\Mcon$ introduced in Section~\ref{s-intro}. Theorem~\ref{thm:dico} indicates exactly for which sets $\Gamma$ of country-specific parameters for a given IKEP, $\Mint$, $\Mnat$, $\Mcon$ can return a $\Gamma$-cycle packing in polynomial time.

\section{Manipulation with Country-Specific Parameters}\label{sec:mech}

In this section, we will present and analyse our IR and IC mechanism for international kidney exchange that respects country-specific parameters. In order to analyse our mechanism, we first provide various existence results for mechanisms as well as
several lower bounds on approximation ratios.

\subsection{Existence Results and Lower Bounds}
Recall that $I= \langle G, \V, \Gamma\rangle $ is an IKEP {\it instance} involving $n$ countries with a set of country-specific parameters $\Gamma$; so $(G,{\mathcal V})$ is the $n$-partitioned compatibility graph for a certain round.  We let $\D_I$  be the set of all $\Gamma$-cycle packings of~$G$. 
Note that $\D_I$ contains the empty set as a $\Gamma$-cycle packing as well. 
A {\it (randomised) mechanism} $\mech$ takes an IKEP instance~$I$ as input and returns exactly one $\Gamma$-cycle packing from $\D_I$ according to some probability distribution~$p$ defined on $\D_I$. Note that the output of a mechanism $\mech$ may not be a maximum $\Gamma$-cycle packing. That is, $\mech$ does not necessarily solve {\sc Max $\Gamma$-Cycle Packing} (but could be modified to do so if desired). 

Let $I= \langle G, \V, \Gamma\rangle $ be an IKEP instance with $\D_I=\{{\mathcal C}_1,\ldots, {\mathcal C}_r\}$ for some $r\geq 1$.
For a $\Gamma$-cycle~$C$, we define the {\it utility} $u_i(C)$ for country~$i$ as the number of arcs $(u,v)$ of $C$ with $v\in V_i$. 
That is,  $u_i(C)$ is the number of kidney transplants for country~$i$ if $C$ is used.
For ${\mathcal C}\in \D_I$, we define the {\it utility}  for country~$i$ as $u_i(\C)=\sum_{C\in \C}u_i(C)$.
In other words, $u_i(\C)$ is the total number of kidney transplants for country~$i$ if $\C$ is used. For a mechanism $\mech$ with probability distribution $p$ on instance $I$, the {\it expected utility} for country~$i$ is $U_i(\mech(I))=p_1(\C_1)u_i(\C_1) + \ldots + p_r(\C_r)u_i(\C_r)$, and the {\em expected social welfare}  $\mbox{SW}(\mech(I))= \sum_{i=1}^n U_i(\mech(I))$ is the sum of all the expected utilities.
 
Let $I= \langle G, \V, \Gamma\rangle $ be an IKEP instance. 
Recall that the size $\|\C\|$ of a $\Gamma$-cycle packing of $G$ is the sum of the lengths of the cycles of ${\mathcal C}$, that is, the total number of kidney transplants if the IKEP uses $\|\C\|$.
We let {\tt opt}$(I)$ denote the size of a maximum $\Gamma$-cycle packing of $G$. To measure the performance of a mechanism~$\mech$, we say that $\mech$ provides an {\em $\alpha$-approximation ratio} if for every IKEP instance $I$: 
$$\frac{{\tt opt}(I)}{\mbox{SW}(\mech(I))} \leq \alpha.$$
We now give some desirable properties that mechanisms for IKEP instances may or may not have, starting with the following natural property:

\begin{itemize}
\item {\bf efficiency:} a mechanism $\mech$ is {\em efficient} if for every $I=\langle G, \V, \Gamma\rangle$, $\mbox{SW}(\mech(I))={\tt opt}(I)$.
\end{itemize}

Let $I= \langle G, \V, \Gamma\rangle $ be an IKEP instance. 
A $\Gamma$-cycle packing $\C$ of $G$ is {\em maximal} if $G$ has no $\Gamma$-cycle packing $\C'$ with $\C'\supsetneq \C$.  We can relax the efficiency requirement by only demanding the following:

\begin{itemize}
\item {\bf maximality:} a mechanism $\mech$ is {\em maximal} if for every $I=\langle G, \V, \Gamma\rangle$, $\mech(I)$ always returns a maximal $\Gamma$-cycle packing of $G$.
\end{itemize}

We can even relax further by demanding:

\begin{itemize}
\item {\bf nonemptiness:} a mechanism $\mech$ is {\em nonempty} if for every $I=\langle G, \V, \Gamma\rangle$, $\mech(I)$ always returns a nonempty $\Gamma$-cycle packing of $G$.
\end{itemize}

Let  $I= \langle G, \V, \Gamma\rangle $ be an IKEP instance. 
Recall that the graph $G[V_i]$ is the subgraph of $G$ induced by $V_i$, so $G[V_i]$ only contains vertices from country~$i$.  
For each country~$i$, we let ${\tt nat}_i(I)$ denote the size of a maximum $\Gamma$-cycle packing in $G[V_i]$. By requiring the following property,  no country has an incentive to leave the IKEP and run a KEP on its own:

\begin{itemize}
\item {\bf individual rationality (IR):} a mechanism $\mech$ is {\it individual rational} if for every $I=\langle G, \V, \Gamma\rangle$ it holds that for every country~$i$, $U_i(\mech(I))\geq {\tt nat}_i(I)$.
\end{itemize}

Finally, we define incentive compatibility. We do this only with respect to the international segment size and the international segment number. Moreover, we assume that due to resource limits, countries can only misreport their international segment size and international segment number by reporting {\it lower} values than their actual values. 
Now, let $I= \langle G, \V, \Gamma\rangle $ be an IKEP instance with $\Gamma=(n,\natv, \intv, \segv, \sv)$. We say that $\Gamma' \leq_i \Gamma$ if $\Gamma'=(n,\natv, \intv, \segv', \sv')$ with
 $\segv'_h=\segv_h$ if $h\neq i$ and $\segv'_i\leq \segv_i$ and $\sv'_h=\sv_h$ if $h\neq i$ and $\sv'_i\leq \sv_i$.

\begin{itemize}
\item {\bf incentive compatibility (IC):} a mechanism $\mech$ is {\it incentive compatible} if for every~$I=\langle G, \V, \Gamma\rangle$ it holds that for every country~$i$ and every $\Gamma' \leq_i \Gamma$, $U_i(\mech(I))\geq U_i(\mech(\langle G, \V, \Gamma'\rangle)$.
\end{itemize}

Our aim is to design mechanisms that are IR and IC, and to determine how these two requirements impact the approximation ratio. 
We start with two propositions.

\begin{proposition}\label{o-efir}
There exists no mechanism that is both efficient and IR. 
\end{proposition}

\begin{proof}
We construct an IKEP instance $I= \langle G, \V, \Gamma\rangle$ as follows. Let $V_1=\{u_1,u_2\}$ and $V_2=\{v_1,v_2\}$. Let $G$ consist of two cycles $C=\langle u_1, u_2, v_1\rangle$ and $D=\langle v_1,v_2\rangle$ that intersect in $v_1$ only. See also Figure~\ref{fig:noeffIR}.
Choose $\Gamma$ such that $\{C\}$ and $\{D\}$ are both $\Gamma$-cycle packings. Now, every efficient mechanism will return $C$, whereas every IR mechanism will return $D$.
\end{proof}

\begin{figure}[t]
\centering
    \begin{subfigure}[t]{0.35\columnwidth}
    \centering
    \resizebox{\columnwidth}{!}{
    \begin{tikzpicture}[triangle/.style = {regular polygon, regular polygon sides=3 }]
    
    \begin{scope}[every node/.style={circle,thick,draw,minimum size=0.35cm}]  
        \node (u1) at (0,0) {$u_1$};
        \node (u2)[right=of u1] {$u_2$};
    \end{scope}
    
    \begin{scope}[every node/.style={rectangle,thick,draw,minimum size=0.6cm}]
        \node (v1)[right=of u2] {$v_1$};
        \node (v2)[right=of v1] {$v_2$};
    \end{scope}

    \begin{scope}[>={Stealth[black]}, every edge/.style={thick,draw}]
        % cycle C
        \path [->] (u1) edge[bend left] node {} (u2);
        \path [->] (u2) edge[bend left] node {} (v1);
        \path [->] (v1) edge[bend left=22] node {} (u1);

        % cycle D
        \path [->] (v1) edge[bend left] node {} (v2);
        \path [->] (v2) edge[bend left] node {} (v1);
    \end{scope}
    
    \end{tikzpicture} 
    }
    \Description{}
    \caption{The example from Proposition~\ref{o-efir}. }
    \label{fig:noeffIR}
    \end{subfigure}
    \begin{subfigure}[t]{0.5\columnwidth}
    \centering
    \resizebox{\columnwidth}{!}{
    \begin{tikzpicture}[triangle/.style = {regular polygon, regular polygon sides=3 }]
    \begin{scope}[every node/.style={circle,thick,draw,minimum size=0.35cm}]  
        \node (u1) at (6,0) {$u_1$};
        \node (u2) at (2,0.8) {$u_2$};
        \node (u3) at (4,0.8) {$u_3$};
        \node (u4) at (6,0.8) {$u_4$};
        \node (u5) at (4,-0.8) {$u_5$};
        \node (u6) at (8,-0.8) {$u_6$};
    \end{scope}
    
    \begin{scope}[every node/.style={rectangle,thick,draw,minimum size=0.6cm}]
        \node (v1) at (8,0.8) {$v_1$};
        \node (v2) at (10,0.8) {$v_2$};
        \node (v3) at (2,-0.8) {$v_3$};
        \node (v4) at (6,-0.8) {$v_4$};
        \node (v5) at (10,-0.8) {$v_5$};
    \end{scope}

    \begin{scope}[>={Stealth[black]}, every edge/.style={thick,draw}]
        % cycle C
        \path [->] (u1) edge[bend left=6] node {} (u2);
        \path [->] (u2) edge[bend left] node {} (u3);
        \path [->] (u3) edge[bend left] node {} (u4);
        \path [->] (u4) edge[bend left] node {} (v1);
        \path [->] (v1) edge[bend left] node {} (v2);
        \path [->] (v2) edge[bend left=6] node {} (u1);

        % cycle D
        \path [->] (u1) edge[bend right=6] node {} (v3);
        \path [->] (v3) edge[bend right] node {} (u5);
        \path [->] (u5) edge[bend right] node {} (v4);
        \path [->] (v4) edge[bend right] node {} (u6);
        \path [->] (u6) edge[bend right] node {} (v5);
        \path [->] (v5) edge[bend right=6] node {} (u1);
    \end{scope}
    \end{tikzpicture} 
    }
    \Description{}
    \caption{The example from Proposition~\ref{p-ic}.}
    \label{fig:emptyIC}   
    \end{subfigure}

    \vspace{0.2cm}
    \begin{subfigure}{0.65\columnwidth}
    \resizebox{\columnwidth}{!}{
    \begin{tikzpicture}[triangle/.style = {regular polygon, regular polygon sides=3 }]
    
    \begin{scope}[every node/.style={circle,thick,draw,minimum size=0.6cm,inner sep=0pt}]  
        \node (u1) at (0,0) {$u_1$};
        \node (u2) at (4,0) {$u_2$};
        \node (u3) at (8,0) {$u_{r-1}$};
        \node (u4) at (10,0) {$u_r$};
    \end{scope}
    
    \begin{scope}[every node/.style={rectangle,thick,draw,minimum size=0.6cm}]
        \node (v1) at (-1,1) {};
        \node (v2) at (3,1) {};
        \node (v3) at (7,1) {};
    \end{scope}
    
    \begin{scope}[every node/.style={diamond,thick,draw,minimum size=0.8cm}]
        \node (w1) at (1,1) {};
        \node (w2) at (5,1) {};
        \node (w3) at (9,1) {};
    \end{scope}

    \node (dot1) at (0,1) {$\cdots$};
    \node (dot2) at (4,1) {$\cdots$};
    \node (dot3) at (6,0) {$\cdots$};
    \node (dot4) at (8,1) {$\cdots$};
    
    \begin{scope}[>={Stealth[black]}, every edge/.style={thick,draw}]
        \path [->] (u1) edge (u2);
        \path [->] (u2) edge (dot3);
        \path [->] (dot3) edge (u3);
        \path [->] (u3) edge (u4);
        \path [->] (u4) edge [bend left=14] (u1);

        \path [->] (u1) edge (v1);
        \path [->] (v1) edge (dot1);
        \path [->] (dot1) edge (w1);
        \path [->] (w1) edge (u1);

        \path [->] (u2) edge (v2);
        \path [->] (v2) edge (dot2);
        \path [->] (dot2) edge (w2);
        \path [->] (w2) edge (u2);

        \path [->] (u3) edge (v3);
        \path [->] (v3) edge (dot4);
        \path [->] (dot4) edge (w3);
        \path [->] (w3) edge (u3);
    \end{scope}
    
    \end{tikzpicture}
    }
    \Description{}
    \caption{The example from Proposition~\ref{p-alpha1}.}
    \label{fig:IRalpha}    
    \end{subfigure}
    \caption{The three examples from Propositions~\ref{o-efir}, \ref{p-alpha1} and~\ref{p-ic}.}
\end{figure}
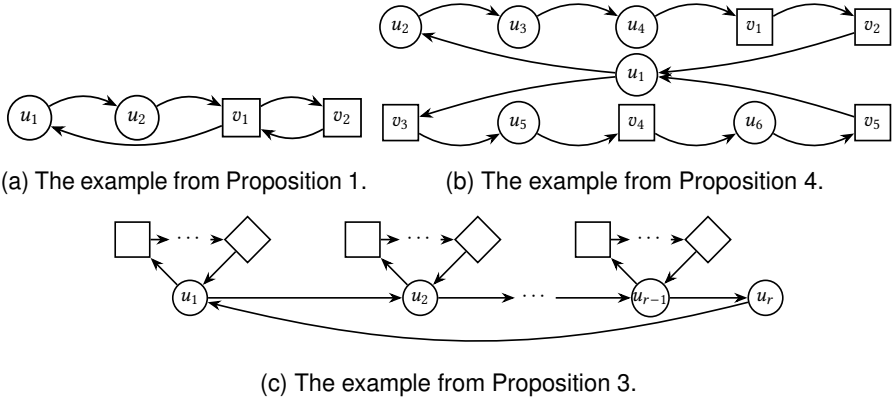

\begin{proposition}\label{p-ir}
There exists a mechanism that is both maximal and IR. 
\end{proposition}

\begin{proof}
For an IKEP instance $I= \langle G, \V, \Gamma\rangle$, the mechanism selects a maximum national $\Gamma$-cycle packing $\C$ of $G$ and then extends $\C$ to a maximal $\Gamma$-cycle packing $\C'$ of $G$ by exhaustively adding (international) $\Gamma$-cycles that are vertex-disjoint from the already selected $\Gamma$-cycles.
\end{proof}

Let $I= \langle G, \V, \Gamma\rangle $ be an IKEP instance.
Recall that a $\Gamma$-cycle  in $G$ is international if it has at least one international arc and that else it is national.
Let $c_{{\tt nat}}$ and $c_{{\tt int}}$ denote the length of a largest national $\Gamma$-cycle and the length of a largest international $\Gamma$-cycle, respectively, in a graph $G$ in an IKEP instance. Note that $c_{{\tt nat}} \leq \max\{\natv_i\; |\; 1\leq i\leq n\}$ and that 
$c_{{\tt int}}\leq \intv$.

We now observe the following lower bound for the approximation ratio of IR mechanisms. As can be seen from the proof of Proposition~\ref{p-alpha1}, this lower bound even holds for near-perfect IKEP instances; here, an IKEP instance is {\it near-perfect} if $G$ has a {\it near-perfect} $\Gamma$-cycle packing ${\mathcal C}$, that is, there exists exactly one vertex in $G$ that belongs to no cycle in ${\mathcal C}$.

\begin{proposition}\label{p-alpha1}
No IR mechanism can provide an $\alpha$-approximation ratio with $\alpha\leq (1-\epsilon)c_{{\tt int}}$, for any $\epsilon>0$.
\end{proposition}

\begin{proof}
We construct an IKEP instance $I= \langle G, \V, \Gamma\rangle$ as follows. Let $C=\langle u_1,u_2, \cdots, u_r\rangle$, for some $r\geq 2$,  be a national cycle with only vertices from $V_1$. For $i\in \{1,\ldots,r-1\}$, let $D_i$ be an international cycle of length $\intv \geq 2$ intersecting with $C$ in $u_i$ only. Assume $C$, $D_1,\ldots,D_{r-1}$ are all $\Gamma$-cycles. See also Figure~\ref{fig:IRalpha}.
Note that $c_{{\tt nat}}=r$.  Every IR mechanism $\mech$ must pick $C$, so $\mbox{SW}(\mech(I))=c_{{\tt nat}}$, but ${\tt opt}(I)=(c_{{\tt nat}}-1)\cdot c_{{\tt int}}$. Hence, $\mech$ has an $\alpha$-approximation ratio with $\alpha \geq \frac{c_{{\tt nat}}-1}{c_{{\tt nat}}}\cdot c_{{\tt int}} > (1-\epsilon)c_{{\tt int}}$ for any $\epsilon>0$ if $c_{{\tt nat}}$ is chosen large enough.
\end{proof}

We now show that IC mechanisms may need to return the empty set as the $\Gamma$-cycle packing for an IKEP instance~$I$. As a consequence, IC mechanisms cannot be maximal (or efficient) either. 

\begin{proposition}\label{p-ic}
There exists no mechanism that is both nonempty and IC. 
\end{proposition}
\begin{proof}
We construct an IKEP instance $I= \langle G, \V, \Gamma\rangle$ as follows. Let $V_1=\{u_1,\ldots,u_6\}$ and $V_2=\{v_1,\ldots,v_5\}$. Let $G$ consist of two cycles $C=\langle u_1,u_2,u_3,u_4,v_1,v_2\rangle$ and $D=\langle u_1,v_3,u_5,v_4,u_6,v_5\rangle$ that both contain only vertex $u_1$. We assume $\intv=6$, $\segv=(4,2)$ and $\sv=(3,3)$, so both $C$ and $D$ are $\Gamma$-cycles. See also Figure~\ref{fig:emptyIC}.
Note that $\D_I=\{\emptyset, \{C\}, \{D\}\}$.

Let $\mech$ be a nonempty IC mechanism that selects $C$ with probability $p$ and that selects $D$ with probability $q$. As $\mech$ is nonempty, $p+q=1$. First assume $p<1$. We observe that $U_1(\mech(I))=p\cdot u_1(\{C\}) + q \cdot u_1(\{D\})=4p+3q = 4p + 3(1-p) = p+3<4$. However, if country~$1$ misreports its segment number as 
$\sv_1'=1$ (or as $\sv_1'=2$), then $\mech$ will return $C$ and country~1 has expected utility~$4$.
As $\mech$ is IC, this means that $p=1$ and thus $q=0$. We now have $U_2(\mech(I))=p\cdot u_2(\{C\}) + q \cdot u_2(\{D\}) = 2$. However, if country~$2$ misreports its segment size as $\segv_2'=1$, then $\mech$ will return $D$ and country~2 has expected utility~$3$, a contradiction.
\end{proof}

Let $I= \langle G, \V, \Gamma\rangle $ be an IKEP instance. 
The {\it international cycle degree} $d(C)$ of $C$ is the number of international $\Gamma$-cycles in $G$ that share at least one vertex with $C$.
So, the two international $\Gamma$-cycles in the proof of Proposition~\ref{p-ic} both have international cycle degree~$2$, as each of them shares a vertex with itself and the other one. We let $d^*$ be the maximum international cycle degree over all international $\Gamma$-cycles of $G$. 
We can now show a lower bound for the approximation ratio of IC mechanisms.

\begin{theorem}\label{t-alpha2}
No IC mechanism can provide an $\alpha$-approximation ratio with $\alpha \leq (1-\epsilon)d^*$, for any $\epsilon>0$.
\end{theorem}

\begin{proof}
We construct an IKEP instance $I= \langle G, \V, \Gamma\rangle$ as follows. Let $\N=\{1,\ldots,n+1\}$.
Let $L$ be an arbitrarily large integer. Let $\intv=1+L+2(n-1)=L-1+2n$; $\segv=\{L\}^{n+1}$; and $\sv=\{2\}^{n+1}$.

Let $G$ consist of $n$ international $\Gamma$-cycles $C_1,\ldots,C_n$, which all have exactly one vertex~$x\in V_{n+1}$ in common.
Cycle $C_1$ consists of $x$, followed by $L$ consecutive vertices from country~$1$, followed by an $(n-1)$-sequence of one vertex from countries $2,\ldots,n$, respectively, followed by another $(n-1)$-sequence of one vertex from countries $2,\ldots,n$, respectively. For $i\in \{2,\ldots,n\}$, cycle $C_i$ is constructed in the same way except that every vertex of country~$1$ is replaced by a vertex from country~$i$, and vice versa. So, each $C_i$ has length $1+L+2(n-1)$. See also Figure~\ref{f-alpha2}.
We note that $d^*=n$, and moreover that each $C_i$ is indeed a $\Gamma$-cycle. Hence, $\D_I=\{\emptyset,\{C_1\},\dots,\{C_n\}\}$.

Consider an arbitrary mechanism $\mech$ with probability distribution $p$ on $\D_I$. For $i\in \{1,\ldots,n\}$, let $p_i=p(\{C_1\})$. 
We may assume without loss of generality that $p_1 \leq \frac{1}{n}$.
This means that the expected utility for country~1 is at most 
$$p_1L+(1-p_1)\cdot 2=p_1(L-2)+2 \leq \frac{L-2}{n}+2 = \frac{L-2+2n}{n},$$
where the last inequality follows from the fact that $L$ is large, so $L>2$ holds. Now suppose country~1 misreports its segment number as $\sv_1'= 1$. Let $\Gamma'$ be the new set of country-specific parameters. Now $C_1$ is the only $\Gamma'$-cycle, as $C_2,\ldots,C_n$ all have two $V_1$-segments. Let $q$ be the probability that $C_1$ is selected by $\mech$ for $I'= \langle G, \V, \Gamma'\rangle$.  To ensure IC, we must have $qL \leq \frac{L-2+2n}{n}$, so $q \leq \frac{L-2+2n}{Ln}$. Hence, $\mech$ has an $\alpha$-approximation ratio with 
$$\alpha \geq \frac{{\tt opt}(I')}{\mbox{SW}(\mech(I'))} = \frac{|C_1|}{q|C_1|} = \frac{1}{q} \geq \frac{Ln}{L-2+2n} = \frac{n}{\frac{2n-2}{L} + 1}>(1-\epsilon)d^*,$$
for any $\epsilon>0$ if $L > > n$ is chosen large enough.
\end{proof}

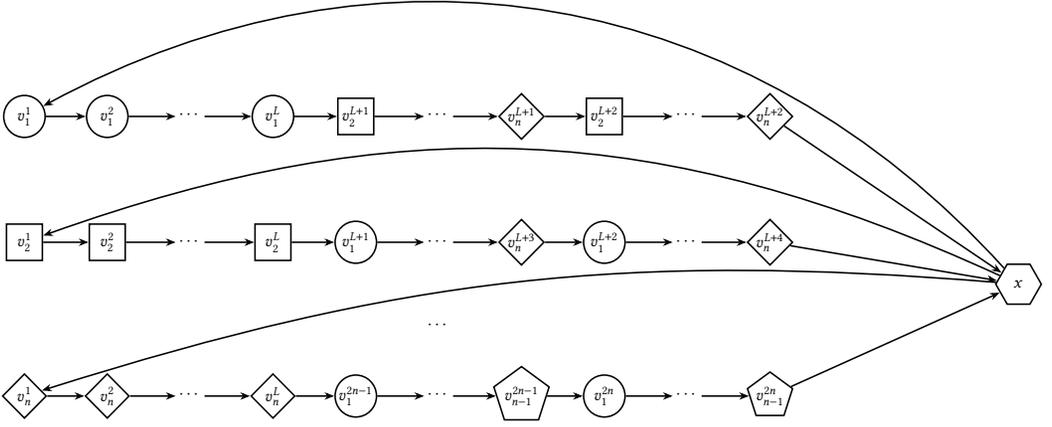
\begin{figure}[t]
\centering
\resizebox{\columnwidth}{!}{
\begin{tikzpicture}[pentagon/.style = {regular polygon, regular polygon sides=5 }, hexagon/.style = {regular polygon, regular polygon sides=6 }]
    
    % First Row
    \node (dot1) at (6,0) {{\huge$\cdots$}};
    \node (dot2) at (15,0) {{\huge$\cdots$}};
    \node (dot3) at (24,0) {{\huge$\cdots$}};
    
    \begin{scope}[every node/.style={circle,ultra thick,draw,minimum size=1.5cm,inner sep=0pt,font=\huge}]
        \node (a1) at (0,0) {$v_1^1$};
        \node (a2) at (3,0) {$v_1^2$};
        \node (aM) at (9,0) {$v_1^L$};
    \end{scope}

    \begin{scope}[every node/.style={rectangle,ultra thick,draw,minimum size=1.3cm,inner sep=0pt,font=\huge}]
        \node (b1) at (12,0) {$v_2^{L+1}$};
        \node (b2) at (21,0) {$v_2^{L+2}$};
    \end{scope}
    
    \begin{scope}[every node/.style={diamond,ultra thick,draw,minimum size=1.65cm,inner sep=0pt,font=\huge}]
        \node (z1) at (18,0) {$v_n^{L+1}$};
        \node (z2) at (27,0) {$v_n^{L+2}$};
    \end{scope}
    
    % Second Row
    \node (dot4) at (6,-4.5) {{\huge$\cdots$}};
    \node (dot5) at (15,-4.5) {{\huge$\cdots$}};
    \node (dot6) at (24,-4.5) {{\huge$\cdots$}};
    
    \begin{scope}[every node/.style={rectangle,ultra thick,draw,minimum size=1.3cm,inner sep=0pt,font=\huge}]
        \node (b3) at (0,-4.5) {$v_2^1$};
        \node (b4) at (3,-4.5) {$v_2^2$};
        \node (bM) at (9,-4.5) {$v_2^L$};
    \end{scope}
    
    \begin{scope}[every node/.style={circle,ultra thick,draw,minimum size=1.5cm,inner sep=0pt,font=\huge}]
        \node (a4) at (12,-4.5) {$v_1^{L+1}$};
        \node (a5) at (21,-4.5) {$v_1^{L+2}$};
    \end{scope}

    \begin{scope}[every node/.style={diamond,ultra thick,draw,minimum size=1.65cm,inner sep=0pt,font=\huge}]
        \node (z3) at (18,-4.5) {$v_n^{L+3}$};
        \node (z4) at (27,-4.5) {$v_n^{L+4}$};
    \end{scope}

    % Third Row
    \node (dot7) at (15,-7.5) {{\huge$\cdots$}};

    \begin{scope}[every node/.style={hexagon,ultra thick,draw,minimum size=1.65cm,inner sep=0pt,font=\huge}]
        \node (star) at (36,-6) {$x$};
    \end{scope}
    
    % Fourth Row
    \node (dot8) at (6,-10) {{\huge$\cdots$}};
    \node (dot9) at (15,-10) {{\huge$\cdots$}};
    \node (dot10) at (24,-10) {{\huge$\cdots$}};
    
    \begin{scope}[every node/.style={diamond,ultra thick,draw,minimum size=1.6cm,inner sep=0pt,font=\huge}]
        \node (z5) at (0,-10) {$v_n^1$};
        \node (z6) at (3,-10) {$v_n^2$};
        \node (z7) at (9,-10) {$v_n^L$};
    \end{scope}

    \begin{scope}[every node/.style={circle,ultra thick,draw,minimum size=1.5cm,inner sep=0pt,font=\huge}]
        \node (a6) at (12,-10) {$v_1^{2n-1}$};
        \node (a7) at (21,-10) {$v_1^{2n}$};
    \end{scope}

    \begin{scope}[every node/.style={pentagon,ultra thick,draw,minimum size=1.65cm,inner sep=0pt,font=\huge}]
        \node (y1) at (18,-10) {$v_{n-1}^{2n-1}$};
        \node (y2) at (27,-10) {$v_{n-1}^{2n}$};
    \end{scope}

    \begin{scope}[>={Stealth[black]}, every edge/.style={ultra thick,draw}]
    % Edges for first row
    \path [->,ultra thick] (a1) edge (a2);
    \path [->,ultra thick] (a2) edge (dot1);
    \path [->,ultra thick] (dot1) edge (aM);
    \path [->,ultra thick] (aM) edge (b1);
    \path [->,ultra thick] (b1) edge (dot2);
    \path [->,ultra thick] (dot2) edge (z1);
    \path [->,ultra thick] (z1) edge (b2);
    \path [->,ultra thick] (b2) edge (dot3);
    \path [->,ultra thick] (dot3) edge (z2);
    \path [->,ultra thick] (z2) edge (star);
    \path [->,ultra thick, bend right=38] (star) edge (a1);
    
    % Edges for second row
    \path [->,ultra thick] (b3) edge (b4);
    \path [->,ultra thick] (b4) edge (dot4);
    \path [->,ultra thick] (dot4) edge (bM);
    \path [->,ultra thick] (bM) edge (a4);
    \path [->,ultra thick] (a4) edge (dot5);
    \path [->,ultra thick] (dot5) edge (z3);
    \path [->,ultra thick] (z3) edge (a5);
    \path [->,ultra thick] (a5) edge (dot6);
    \path [->,ultra thick] (dot6) edge (z4);
    \path [->,ultra thick] (z4) edge (star);
    \path [->,ultra thick, bend right=22] (star) edge (b3);
    
    % Edges for third row
    \path [->,ultra thick] (z5) edge (z6);
    \path [->,ultra thick] (z6) edge (dot8);
    \path [->,ultra thick] (dot8) edge (z7);
    \path [->,ultra thick] (z7) edge (a6);
    \path [->,ultra thick] (a6) edge (dot9);
    \path [->,ultra thick] (dot9) edge (y1);
    \path [->,ultra thick] (y1) edge (a7);
    \path [->,ultra thick] (a7) edge (dot10);
    \path [->,ultra thick] (dot10) edge (y2);
    \path [->,ultra thick] (y2) edge (star);
    \path [->,ultra thick, bend right=11] (star) edge (z5);
    \end{scope}
    
\end{tikzpicture}
}
\Description{}
\caption{The example used to prove the lower bound for the approximation ratio in Theorem~\ref{t-alpha2}.}\label{f-alpha2}
\end{figure}

 We note that the proof of Theorem \ref{t-alpha2} makes use of long cycles in the construction of $G$. These may be reasonable within IKEPs where individual countries only have to perform shorter segments as part of a longer cycle. However, a potential requirement for the overall simultaneity of transplants may still pose a challenge. Yet, if we were to consider chains, which can be longer in practice due to simultaneity not being required, then such a construction could be more applicable.

\subsection{Our IR and IC Mechanism}

We will now give a mechanism, which we call $\mech_{{\tt order}}$, and which we prove is both IR and IC. Another important feature of $\mech_{{\tt order}}$ will be that its algorithmic implementation runs in polynomial time as long as we can find a maximum national $\Gamma$-cycle packing in polynomial time and $\intv$ is a constant.

Let $I= \langle G, \V, \Gamma\rangle$ be an IKEP instance.
We say that a mechanism $\mech$ {\it pre-selects} a national $\Gamma$-cycle packing ${\mathcal C}$ of $G$ if $\mech$ always selects ${\mathcal C}$ for every IKEP instance $\langle G, \V, \Gamma'\rangle$ with the same national cycle limit vector as $I$. 
We let $\X(G)$ consist of all international cycles of $G$ of length at most $\intv$. Note that two cycles in $\X(G)$ may have common vertices and also that $\X(G)$ may contain cycles that are not $\Gamma$-cycles. Hence, $\X(G)$ may not be a $\Gamma$-cycle packing of $G$. Let $C$ be a cycle in $\X(G)$ that is not a $\Gamma$-cycle. We say that $C$ is {\it turnable} if $C$ contains a subgraph $D$ that is an international $\Gamma$-cycle of $G$ with at most one $V_i$-segment for every $i\in \{1,\ldots,n\}$. Such a $\Gamma$-cycle $D$ is a {\it substitute} for $C$. 

We now present our mechanism $\mech_{{\tt order}}$.
From its algorithmic description,
Algorithm~\ref{alg:mech}, we see that it always outputs a $\Gamma$-cycle packing (which might be empty). Moreover, 
from Step~1 of Algorithm~\ref{alg:mech}, we see that $\mech_{{\tt order}}$ satisfies IR.
Before proving that $\mech_{{\tt order}}$ is also IC and determining an upper bound for its approximation ratio, we first give some intuition behind this proof by giving three examples and a running time analysis. Examples~1 and~2 show how $\mech_{{\tt order}}$ maintains IC if the instance has no turnable cycles. Example~3 provides intuition if the instance has turnable cycles by showing that we cannot allow substitutes with two (or more) $V_i$-segments for some country~$i$.

\begin{algorithm}[t]
\caption{The mechanism $\mech_{{\tt order}}$}\label{alg:mech}
\Description{}
\KwIn{An IKEP instance $I= \langle G, \V, \Gamma\rangle$}
\KwOut{A $\Gamma$-cycle packing of $G$}
\begin{itemize}
\item []
\end{itemize}
\vspace*{-0.4cm}
{\bf Step 1.} Pre-select a maximum national $\Gamma$-cycle packing ${\mathcal C}_{{\tt nat}}$. 

\smallskip
\noindent
{\bf Step 2.} Construct a graph $\hG$ by deleting all the vertices of the cycles in ${\mathcal C}_{{\tt nat}}$ from $G$. 

\smallskip
\noindent
{\bf Step 3.} Consider a random order $\Pi$ of the cycles in $\X({\hG})$.  

\smallskip
\noindent
{\bf Step 4.} Find a cycle packing ${\mathcal C}$ of $\hG$ 
by exhaustively adding vertex-disjoint cycles from~$\Pi$. 

\smallskip
\noindent
{\bf Step 5.} For each $C\in {\mathcal C}$, do as follows:
\begin{itemize}
\item if $C$ is a $\Gamma$-cycle, keep $C$
\item if $C$ is not a $\Gamma$-cycle but turnable, replace $C$ with a randomly selected substitute for $C$
\item else delete $C$ from ${\mathcal C}$.
\end{itemize}
\noindent
{\bf Step 6.} Return ${\mathcal C}_{{\tt nat}}\cup {\mathcal C}$.
\end{algorithm}

\begin{example}
 \textup{
Consider the IKEP instance $I$ in the proof of Proposition~\ref{p-ic}. So $\intv=6$, $\segv=(2,2)$ and $\sv=(4,2)$, and both $C$ and $D$ are $\Gamma$-cycles, implying that $\mech_{{\tt order}}$ may return $\{C\}$ or $\{D\}$, each with probability $\frac{1}{2}$. Hence, the expected utility for country~1 is $\frac{1}{2}\cdot 4+\frac{1}{2}\cdot 3=3\frac{1}{2}$ and the expected utility for country~2 is $\frac{1}{2}\cdot 2+\frac{1}{2}\cdot 3=2\frac{1}{2}$.
If country~$1$ misreports its segment number as $\sv_1'=1$, then $\mech_{{\tt order}}$ will either return the empty cycle packing with probability $\frac{1}{2}$, or it will return $\{C\}$ with probability $\frac{1}{2}$. Hence, the expected utility for country~1 becomes $\frac{1}{2}\cdot 4=2\leq 3\frac{1}{2}$. Similarly, if country~$2$ misreports its segment size as $\segv_2'=1$, then $\mech_{{\tt order}}$ will either return the empty cycle packing with probability $\frac{1}{2}$, or it will return $\{D\}$ with probability $\frac{1}{2}$. Hence, the expected utility for country~$2$ becomes $\frac{1}{2}\cdot 3=1\frac{1}{2}\leq 2\frac{1}{2}$. In other words, neither country can benefit from misreporting.} \dia
   
\end{example}

\begin{example}
 \textup{
Consider the IKEP instance $I$ in the proof of Theorem~\ref{t-alpha2}. We recall that each $C_i$ is a $\Gamma$-cycle. Hence, $\mech_{{\tt order}}$ will return each $C_i$ with probability $\frac{1}{n}$, and the expected utility for each country~$i$ is equal to $\frac{L-2+2n}{n}$ if $i\in \{1,\ldots,n\}$ and equal to $1$ for $i=n+1$.
If country~1 misreports its segment number as $\sv_1'= 1$, then $\mech_{{\tt order}}$ will return either the empty cycle packing with probability $\frac{n-1}{n}$, or $\{C_1\}$ with probability $\frac{1}{n}$.
Hence, the expected utility for country~1 becomes $\frac{1}{n}\cdot L=\frac{L}{n} \leq \frac{L-2+2n}{n}$, showing that country~1 cannot benefit from misreporting. We also note that $\mech_{{\tt order}}$ has an $\alpha$-approximation ratio with $\alpha \geq n$, which is in line with the lower bound given in Theorem~\ref{t-alpha2}. 
}\dia
\end{example}

\begin{example}
\textup{
We construct an IKEP instance $I= \langle G, \V, \Gamma\rangle$ as follows. Let $\N=\{1,2\}$.
Let $L$ be an arbitrary integer. We set $\intv=L+4$ and $\segv=(L,1)$ and $\sv=(2,2)$.
We obtain $G$ from a cycle $D=\langle u_1,u_2,\cdots, u_L,v_1,u_{L+1},v_2,u_{L+2}\rangle$ after adding the arcs $(v_1,u_1)$ and $(u_{L+2},v_1)$ as well. Let $V_1=\{u_1,\ldots,u_{L+2}\}$ and $V_2=\{v_1,v_2\}$.
See also Figure~\ref{f-ex3}.
We observe that $G$ has exactly three cycles, namely $D$ and the cycles $C_1=\langle u_1,u_2, \cdots, u_L,v_1\rangle$ and $C_2=\langle v_1,u_{L+1},v_2,u_{L+2}\rangle$. The first one, $D$, 
is not a $\Gamma$-cycle, as it contains a $V_1$-segment of size~$L+1>L$, but $C_1$ and $C_2$ are $\Gamma$-cycles. }

\textup{As $G$ has no national $\Gamma$-cycles, we have ${\mathcal C}_{{\tt nat}}=\emptyset$ and $\hG=G$.
As $D$ has length~$\intv$, and $C_1$ and $C_2$ are $\Gamma$-cycles, we find that $\X(G)=\{C_1,C_2,D\}$.
So, $\mech_{{\tt order}}$ may still select $D$ in Step~4. However, if it does, then it must replace $D$ in Step 5 with its only substitute $C_1$. }

\textup{We now consider country~1. The expected utility of country~1 is $\frac{1}{3}\cdot L+ \frac{1}{3}\cdot L + \frac{1}{3}\cdot 2 = \frac{2}{3} L + \frac{2}{3}$.
If country~$1$ misreports its segment number as $\sv_1'=1$, then the expected utility for country~1 is $\frac{1}{3}\cdot L + \frac{1}{3}\cdot L = \frac{2}{3} L$. Hence, country~1 cannot benefit from misreporting.}

\textup{
We also note that it is essential to consider only substitutes with at most one $V_i$-segment per country. In order to see this, suppose that we would modify $\mech_{{\tt order}}$ by allowing substitutes with at most two $V_i$-segments. We consider the same IKEP instance $I$. Now $D$ has both $C_1$ and $C_2$ as its two substitutes. This means that the expected utility of country~1 is $\frac{1}{3} (\frac{1}{2}\cdot L+\frac{1}{2}\cdot 2) + \frac{1}{3}\cdot L + \frac{1}{3}\cdot 2=\frac{1}{2}L+1$. If country~$1$ misreports its segment number as $\sv_1'=1$, then the expected utility for country~1 is $\frac{1}{3}\cdot L + \frac{1}{3}\cdot L = \frac{2}{3} L$. Hence, in this case, country~1 {\it can} benefit significantly from misreporting (if $L$ is large), so our modified mechanism would not be IC. 
}\dia
\end{example}

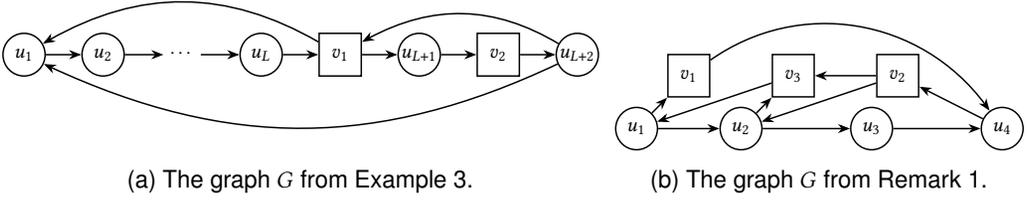
\begin{figure}
\centering
\begin{subfigure}[t]{0.58\columnwidth}
\centering
\resizebox{\columnwidth}{!}{
\begin{tikzpicture}[triangle/.style = {regular polygon, regular polygon sides=3 }]
    \begin{scope}[every node/.style={circle,thick,draw,minimum size=0.8cm,inner sep=0pt,font=\large}]  
        \node (u1) at (-1.5,0) {$u_1$};
        \node (u2) at (0,0) {$u_2$};
        \node (uL) at (3,0) {$u_L$};
        \node (uL1) at (6,0) {$u_{L+1}$};
        \node (uL2) at (9,0) {$u_{L+2}$};
    \end{scope}

    \begin{scope}[>={Stealth[black]}, every node/.style={rectangle,thick,draw,minimum size=0.8cm,inner sep=0pt,font=\large}]
        \node (v1) at (4.5,0) {$v_1$};
        \node (v2) at (7.5,0) {$v_2$};
    \end{scope}

    \node (dot) at (1.5,0) {$\cdots$};

    \begin{scope}[>={Stealth[black]}, every edge/.style={thick,draw}]
        % cycle D
        \path [->] (u1) edge node {} (u2);
        \path [->] (u2) edge node {} (dot);
        \path [->] (dot) edge node {} (uL);
        \path [->] (uL) edge node {} (v1);
        \path [->] (v1) edge node {} (uL1);
        \path [->] (uL1) edge node {} (v2);
        \path [->] (v2) edge node {} (uL2);
        \path [->] (uL2) edge[bend left=25] node {} (u1);

        % two arcs
        \path [->] (v1) edge[bend right=30] node {} (u1);
        \path [->] (uL2) edge[bend right=30] node {} (v1);
    \end{scope}
\end{tikzpicture} 
}
\Description{}
\caption{The graph $G$ from Example~3.}\label{f-ex3}
\end{subfigure}
\begin{subfigure}[t]{0.4\columnwidth}
\centering
\resizebox{\columnwidth}{!}{
\begin{tikzpicture}[triangle/.style = {regular polygon, regular polygon sides=3 }]
    \begin{scope}[every node/.style={circle,thick,draw,minimum size=0.8cm,inner sep=0pt,font=\large}]  
        \node (u1) at (0,0) {$u_1$};
        \node (u2) at (2,0) {$u_2$};
        \node (u3) at (4.5,0) {$u_3$};
        \node (u4) at (7,0) {$u_{4}$};
    \end{scope}
    
    \begin{scope}[every node/.style={rectangle,thick,draw,minimum size=0.8cm,inner sep=0pt,font=\large}]
        \node (v1) at (1,1) {$v_1$};
        \node (v2) at (5,1) {$v_2$};
        \node (v3) at (3,1) {$v_3$};
    \end{scope}

    \begin{scope}[>={Stealth[black]}, every edge/.style={thick,draw}]
        % cycle C
        \path [->] (u1) edge node {} (u2);
        \path [->] (u2) edge node {} (u3);
        \path [->] (u3) edge node {} (u4);
        \path [->] (u4) edge node {} (v2);
        \path [->] (v2) edge node {} (v3);
        \path [->] (v3) edge node {} (u1);

        % cycle D
        \path [->] (v2) edge node {} (u2);
        \path [->] (u2) edge node {} (v3);
        \path [->] (u1) edge node {} (v1);
        \path [->] (v1)[bend left=45] edge node {} (u4);
    \end{scope}
\end{tikzpicture}
}
\Description{}
\caption{The graph $G$ from Remark~1.}\label{f-rm1}
\end{subfigure}
\Description{}
\caption{The graphs from Example~3 and Remark~1.}
\end{figure}

We now analyse the running time of $\mech_{{\tt order}}$.

\begin{proposition}
$\mech_{{\tt order}}$ can be executed in polynomial time if Step~1 takes polynomial time and $\emph{\intv}$ is a constant.
\end{proposition}

\begin{proof}
In Step~1, we must compute a maximum national $\Gamma$-cycle packing. Assuming $\cP\neq \NP$, this problem is only polynomial-time solvable if $\natv_i\in \{0,2,\infty\}$ for every $i\in \{1,\ldots,n\}$ by Theorem~\ref{t-dicho}. The other steps take polynomial time. This is because the set $\X(\hG)$ has size at most $O(n^{\intv})$, which is polynomial since $\intv$ is a constant.
It also takes polynomial time to check if a cycle $C$ is a $\Gamma$-cycle and to determine its sets of substitutes, which all belong to $\X(\hG)$. 
\end{proof}

\begin{remark}
\textup{
 If computational time is not an issue, we could enhance any mechanism that is IR and/or IC with one additional property, namely \emph{perfectness}. 
A mechanism $\mech$ is {\it perfect} if on any input $I=(G,\V,\Gamma)$, it always returns a perfect $\Gamma$-cycle packing of $G$ if there exists one. A check on this could take place at the start. It is not possible to relax this property by requiring only {\it near-perfectness}, so returning a near-perfect $\Gamma$-cycle packing if there exists one. For not being IR, this is illustrated by the example in the proof of Proposition~\ref{p-alpha1}. For not being IC, let $I=(G,\V,\Gamma)$ be an IKEP instance with $V_1=\{u_1,u_2,u_3.u_4\}$ and $V_2=\{v_1,v_2,v_3\}$. Let $G$ be constructed from a 6-vertex cycle $C = \langle u_1,u_2,u_3,u_4,v_2,v_3\rangle$ that we extend by adding arcs $(v_2,u_2)$ and $(u_2,v_3)$ and by adding $v_1$ with arcs $(u_1,v_1)$ and $(v_1,u_4)$, as shown in Figure \ref{f-rm1}. So, we have exactly one other $6$-vertex cycle, namely $D=\langle u_1,v_1,u_4,v_2,u_2,v_3 \rangle$. We can now use the same arguments as in the proof of Proposition~\ref{p-ic}.
}\dia
\end{remark}

We are now ready to formally prove that our mechanism is both IR and IC.

\begin{theorem}\label{t-ms4}
 $\mech_{{\tt order}}$ is a mechanism that satisfies IR and IC.
\end{theorem}

\begin{proof}
As mentioned, it is readily seen from its description that $\mech_{{\tt order}}$ always returns a $\Gamma$-cycle packing for every IKEP instance $I$; in particular, if in Step~5 we replace a turnable cycle $C$ with a substitute $D$ and another turnable cycle $C'$ by a substitute~$D'$, then $D$ and $D'$ will be vertex-disjoint (as $C$ and $C'$ are vertex-disjoint, and $D$ is a subgraph of $C$, and $D'$ is a subgraph of $C'$).
Hence, $\mech_{{\tt order}}$ is indeed a mechanism.
Step~1 ensures that $\mech_{{\tt order}}$ is IR. We will now prove that $\mech_{{\tt order}}$ is IC. 

Let $I= \langle G, \V, \Gamma\rangle $ be an IKEP instance.  Assume country~$1$ misreports either its segment size $\segv_1$ or its segment number $\sv_1$ by choosing a smaller value. Let $\Gamma'$ be the new set of country-specific parameters and let $I'= \langle G, \V, \Gamma'\rangle $ be the new IKEP instance. We will prove the theorem by showing that $U_1(\mech_{{\tt order}}(I')) \leq U_1(\mech_{{\tt order}}(I))$.

Let ${\mathcal C}_{{\tt nat}}$ be the maximum national $\Gamma$-cycle packing that is pre-selected by $\mech_{{\tt order}}$ in Step~1. So, ${\mathcal C}_{{\tt nat}}$ will always be selected in Step~1 also after the misreporting of country~$1$.
Hence, the graph $\hG$ constructed in Step~2 by deleting all vertices that belong to a cycle in ${\mathcal C}_{{\tt nat}}$ from $G$ will be the same graph irrespective of the misreporting of country~$1$. Let ${\mathcal Y}$ consist of all national $\Gamma$-cycles of $\hG$. Let ${\mathcal Z}$ consist of all international $\Gamma$-cycles of $\hG$. Similarly, let ${\mathcal Y}'$ denote the set of all national $\Gamma'$-cycles of $\hG$, and let ${\mathcal Z}'$ denote the set of all international $\Gamma'$-cycles of $\hG$. 

For a cycle $C\in \X(\hG)$, let $p(C)$ be the probability that $C$ will be selected for the $\Gamma$-cycle packing returned by $\mech_{{\tt order}}$ on input $I$. Similarly, we let
$q(C)$ be the probability that $C\in \X(\hG)$ will be selected for the $\Gamma'$-cycle packing returned by $\mech_{{\tt order}}$ on input $I'$. Note that $p(C)=0$ if $C$ is not a $\Gamma$-cycle (even if $C$ is selected in Step 4, it would be discarded in Step~5). Similarly, $q(C)=0$ if $C$ is not a $\Gamma'$-cycle. We observe that $U_1(\mech_{{\tt order}}(I')) \leq U_1(\mech_{{\tt order}}(I))$ if and only if
\begin{equation}
\sum_{C\in {\mathcal Y}'}q(C)\cdot u_1(C) + \sum_{C\in {\mathcal Z}'}q(C)\cdot u_1(C) \leq  \sum_{C\in {\mathcal Y}}p(C)\cdot u_1(C) + \sum_{C\in {\mathcal Z}}p(C)\cdot u_1(C).
\end{equation}
Recall that ${\mathcal C}_{{\tt nat}}$ will always be selected in Step~1, also after the misreporting of country~$1$. 
This means that ${\mathcal Y}={\mathcal Y}'$, and for every $C\in {\mathcal Y}$, we have $p(C)=q(C)$ (which is equal to $0$ or $1$). Hence,  $\sum_{C\in {\mathcal Y}'}q(C)\cdot u_1(C)=\sum_{C\in {\mathcal Y}}p(C)\cdot u_1(C)$, and we are left to prove:
\begin{equation}
\sum_{C\in {\mathcal Z}'}q(C)\cdot u_1(C) \leq  \sum_{C\in {\mathcal Z}}p(C)\cdot u_1(C).
\end{equation}
Note that  ${\mathcal Z}'\subseteq {\mathcal Z}$, as country~$1$ misreports by claiming a smaller value for its segment size $\segv_1$ or its segment number $\sv_1$. 
We need some more notation. With a cycle $C\in \X(\hG)$, we associate a set $\X_C$ that consists of all the cycles in $\X(\hG)$ sharing at least one vertex with $C$, so $C$ also belongs to $\X_C$. Note that in Step~4, every cycle $C\in \X(\hG)$ is selected with probability $\frac{1}{|\X_C|}$.

\begin{figure}[ht]
\centering
\includegraphics[width=0.4\linewidth]{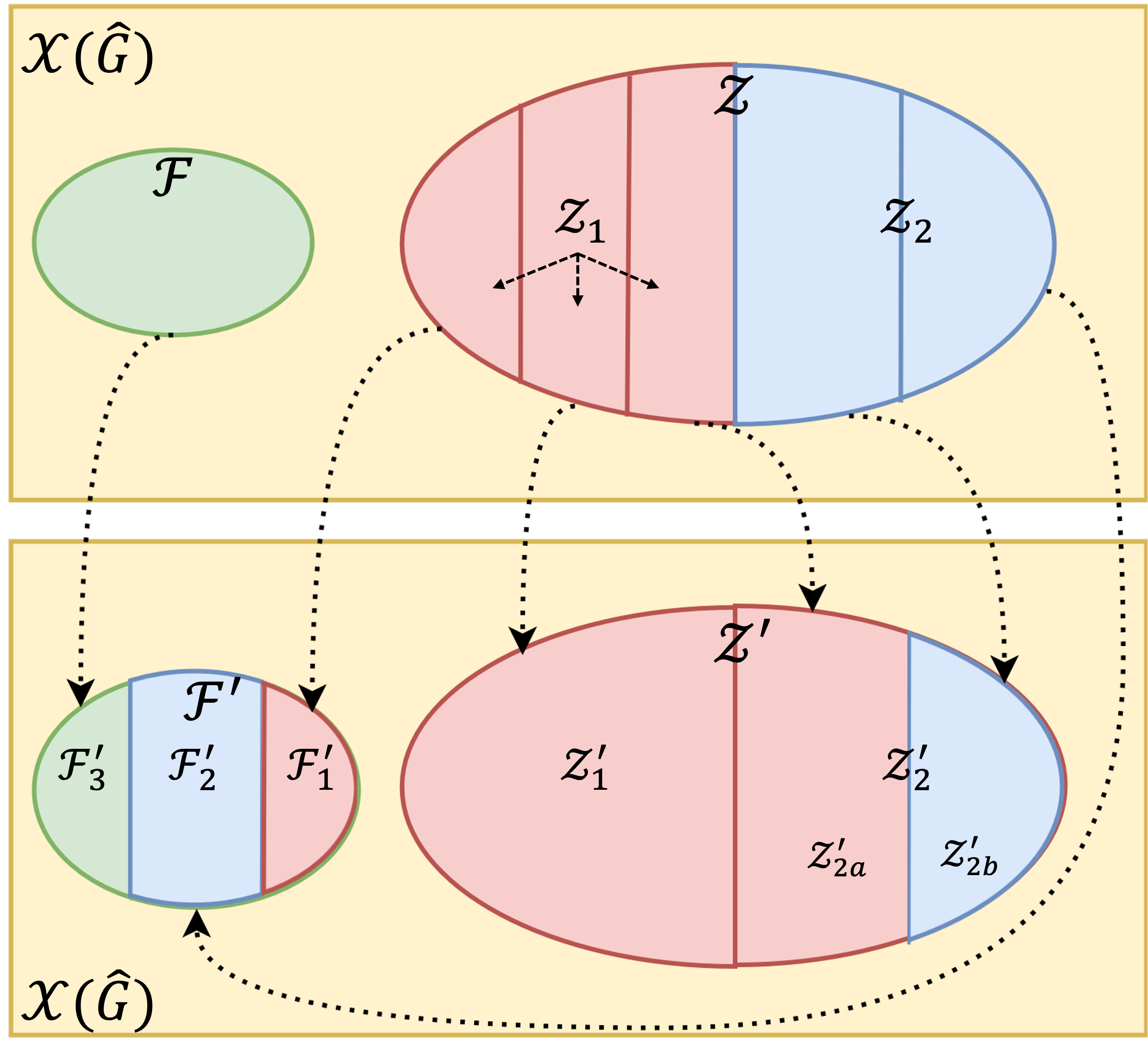}
\Description{}
\caption{An illustration of the cycles in $\X(\hG)$ for instances $I$ (top figure) and $I'$ (bottom figure).}\label{fig:sets}
\end{figure}

We first consider instance $I$. We partition ${\mathcal Z}$ into sets ${\mathcal Z}_1$ and ${\mathcal Z}_2$ as follows. Let ${\mathcal Z}_1$ consist of all $\Gamma$-cycles in ${\mathcal Z}$ that are not substitutes for turnable cycles in $\X(\hG)$ for~$I$. Note that cycles in ${\mathcal Z}_1$ may still be contained as subgraphs of larger $\Gamma$-cycles in $\X(\hG)$. Let ${\mathcal Z}_2={\mathcal Z}\setminus {\mathcal Z}_1$, so ${\mathcal Z}_2$ consists of $\Gamma$-cycles that are a substitute for at least one turnable cycle in $\X(\hG)$ for $I$.  

Let ${\mathcal F}=\{F_1,\ldots,F_r\}$ consist of all the cycles in $\X(\hG)$ that are not $\Gamma$-cycles. We note that $\X(\hG)={\mathcal Z}_1\cup {\mathcal Z}_2\cup {\mathcal F}$; see also Figure~\ref{fig:sets}. For $i=1,\ldots,r$, we let ${\mathcal A}_i$ consist of all $\Gamma$-cycles that are substitutes for $F_i$. Note that every~${\mathcal A}_i$ consists of cycles from ${\mathcal Z}_2$ only, but ${\mathcal A}_i$ might be empty. Moreover, two sets ${\mathcal A}_i$ and ${\mathcal A}_j$ for some $i\neq j$ may have a nonempty intersection.  If $\mech_{{\tt order}}$ selects $F_i$ in Step~4, then in Step~5 each cycle from ${\mathcal A}_i$ has equal probability $\frac{1}{|{\mathcal A}_i|}$ to be chosen. We now find that
$$\sum_{C\in {\mathcal Z}}p(C)\cdot u_1(C) = \sum_{C\in {\mathcal Z}_1} \frac{1}{|\X_C|}u_1(C) + \sum_{C\in {\mathcal Z}_2} \frac{1}{|\X_C|}u_1(C) + \sum_{i=1}^r \frac{1}{|\X_{F_i}|}\sum_{C\in {\mathcal A}_i}\frac{1}{|{\mathcal A}_i|}u_1(C).$$

We now consider instance $I'$. 
Let ${\mathcal Z}_1'$ consist of all $\Gamma$-cycles in ${\mathcal Z}'$ that are not substitutes for turnable cycles in $\X(\hG)$ for~$I'$. Note that ${\mathcal Z}_1'\subseteq {\mathcal Z}_1$. This is because country~$1$ only misreports by giving a lower value for $\segv_1$ or $\sv_1$. Hence, any turnable cycles for instance~$I$ are also turnable for~$I'$, and thus, any cycle $C\in {\mathcal Z}_2$ is still a substitute for a turnable cycle in $I'$, or $C$ is not a $\Gamma'$-cycle.

Let ${\mathcal Z}_2' ={\mathcal Z}'\setminus {\mathcal Z}_1'$, so ${\mathcal Z}_2'$ consists of $\Gamma'$-cycles that are a substitute for at least one turnable cycle in $\X(\hG)$ for instance $I'$. We partition ${\mathcal Z}_2'$ into sets ${\mathcal Z}_{2{\tt a}}'$ and ${\mathcal Z}_{2{\tt b}}'$ as follows. We let ${\mathcal Z}_{2{\tt a}}'$ consist of all $\Gamma'$-cycles of ${\mathcal Z}_1$ that became a substitute for at least one turnable cycle in $\X(\hG)$ for $I'$. We let ${\mathcal Z}_{2{\tt b}}'$ consist of all $\Gamma'$-cycles of ${\mathcal Z}_2$ that are a substitute for at least one turnable cycle in $\X(\hG)$ for $I'$.

Let ${\mathcal F}'$  consist of all cycles in $\X(\hG)$ that are not $\Gamma'$-cycles. So, $\X(\hG)={\mathcal Z}_1'\cup {\mathcal Z}_2'\cup {\mathcal F}' (= {\mathcal Z}_1\cup {\mathcal Z}_2\ \cup {\mathcal F})$. We partition ${\mathcal F}'$ into three sets ${\mathcal F}_1'$, ${\mathcal F}_2'$ and ${\mathcal F}_3'$ as follows (see also Figure~\ref{fig:sets}):

\begin{enumerate}
\item 
There may exist cycles in ${\mathcal Z}_1$ that are not $\Gamma'$-cycles. If so, we let ${\mathcal F}_1'=\{F_1',\ldots F_s'\}$, for some $s\geq 1$, be the set of such cycles. For $i=1,\ldots,s$, we let ${\mathcal A}_i'$ consist of all $\Gamma'$-cycles that are substitutes for $F_i'$. 
\item There may exist cycles in ${\mathcal Z}_2$ that are not $\Gamma'$-cycles. If so, we let ${\mathcal F}_2'=\{F_1'',\ldots,F_t''\}$, for some $t\geq 1$, be the set of such cycles. For $i=1,\ldots,t$, we let ${\mathcal A}_i''$ consist of all $\Gamma'$-cycles that are substitutes for $F_i''$.  
\item We note that each cycle in $\X(\hG)$ that is not a $\Gamma$-cycle is also not a $\Gamma'$-cycle (as country~$1$ only misreports by giving a lower value of $\segv_1$ or $\sv_1$). Hence, ${\mathcal F}\subseteq {\mathcal F}'$, and we set ${\mathcal F}_3'={\mathcal F}$. For $i=1,\ldots,r$, we let ${\mathcal A}_i^*$ consist of all $\Gamma'$-cycles that are substitutes for $F_i$.
\end{enumerate}

We note that indeed $({\mathcal F}_1',{\mathcal F}_2',{\mathcal F}_3')$ forms a partition of ${\mathcal F}'$. We also note that the sets ${\mathcal A}_i'$, ${\mathcal A}_i''$ and ${\mathcal A}_i^*$ might overlap, and some of them might be empty. 
Finally, a set ${\mathcal A}_i'$ may consist of cycles from both ${\mathcal Z}_{2{\tt a}}'$ and ${\mathcal Z}_{2{\tt b}}'$, whereas a set ${\mathcal A}_i''$ and a set ${\mathcal A}_i^*$  may only consist of cycles from  ${\mathcal Z}_{2{\tt b}}'$ (if these sets contains a cycle from ${\mathcal Z}_{2{\tt a}}'\subseteq {\mathcal Z}_1$, then such a cycle would have been placed in ${\mathcal Z}_2$ instead of in ${\mathcal Z}_1$). Hence, in particular, we have  $A_i^*\subseteq A_i$ for every $i\in \{1,\ldots,r\}$.
We now find that
\[\begin{array}{lcl}
\displaystyle\sum_{C\in {\mathcal Z}'}q(C)\cdot u_1(C) &= &\displaystyle\sum_{C\in {\mathcal Z}_1'} \frac{1}{|\X_C|}u_1(C) + \displaystyle\sum_{C\in {\mathcal Z}_{2{\tt a}}'} \frac{1}{|\X_C|}u_1(C) +\displaystyle\sum_{C\in {\mathcal Z}_{2{\tt b}}'} \frac{1}{|\X_C|}u_1(C)\\[20pt]
&+ &\displaystyle\sum_{i=1}^s \frac{1}{|\X_{F_i'}|}\sum_{C\in {\mathcal A}_i'}\frac{1}{|{\mathcal A}_i'|}u_1(C) +  \displaystyle\sum_{i=1}^t \frac{1}{|\X_{F_i''}|}\sum_{C\in {\mathcal A}_i''}\frac{1}{|{\mathcal A}_i''|}u_1(C)\\[20pt] &+  &\displaystyle\sum_{i=1}^r \frac{1}{|\X_{F_i}|}\sum_{C\in {\mathcal A}_i^*}\frac{1}{|{\mathcal A}_i^*|}u_1(C).
\end{array}\]
We have that ${\mathcal Z}_1={\mathcal Z}_1'\cup {\mathcal Z}_{2{\tt a}}' \cup {\mathcal F}_1'$. Moreover, for every $i\in \{1,\ldots,s\}$ the number of vertices of country~1 in every cycle $C$ in the set of substitutes ${\mathcal A}_i'$ for $F_i'$ is at most the number of vertices of country~1 in $F_i$, that is, $u_1(C)\leq u_1(F_i')$ for every $C\in {\mathcal A}_i'$. Hence, we also have
$\sum_{C\in {\mathcal A}_i'}\frac{1}{|{\mathcal A}_i'|}u_1(C)\leq \sum_{C\in {\mathcal A}_i'}\frac{1}{|{\mathcal A}_i'|}u_1(F_i')=u_1(F_i')$ for every $i\in \{1,\ldots,s\}$. This means that:
$$\displaystyle\sum_{C\in {\mathcal Z}_1'} \frac{1}{|\X_C|}u_1(C) + \displaystyle\sum_{C\in {\mathcal Z}_{2{\tt a}}'} \frac{1}{|\X_C|}u_1(C)+\displaystyle\sum_{i=1}^s \frac{1}{|\X_{F_i'}|}\sum_{C\in {\mathcal A}_i'}\frac{1}{|{\mathcal A}_i'|}u_1(C)\leq
 \sum_{C\in {\mathcal Z}_1} \frac{1}{|\X_C|}u_1(C).$$
Similarly, as ${\mathcal Z}_2={\mathcal Z}_{2{\tt b}}' \cup {\mathcal F}_2'$ and $\sum_{C\in {\mathcal A}_i''}\frac{1}{|{\mathcal A}_i''|}u_1(C)\leq \sum_{C\in {\mathcal A}_i''}\frac{1}{|{\mathcal A}_i''|}u_1(F_i'')=u_1(F_i'')$ for $i\in \{1,\ldots,t\}$, we also have:
$$\displaystyle\sum_{C\in {\mathcal Z}_{2{\tt a}}'} \frac{1}{|\X_C|}u_1(C) + \displaystyle\sum_{i=1}^t \frac{1}{|\X_{F_i''}|}\sum_{C\in {\mathcal A}_i''}\frac{1}{|{\mathcal A}_i''|}u_1(C) \leq  \sum_{C\in {\mathcal Z}_2} \frac{1}{|\X_C|}u_1(C).$$
Hence, it remains to prove that 
\begin{equation}\displaystyle\sum_{i=1}^r \frac{1}{|\X_{F_i}|}\sum_{C\in {\mathcal A}_i^*}\frac{1}{|{\mathcal A}_i^*|}u_1(C)
\leq \sum_{i=1}^r \frac{1}{|\X_{F_i}|}\sum_{C\in {\mathcal A}_i}\frac{1}{|{\mathcal A}_i|}u_1(C).\end{equation}
In order to prove this, we focus on some $F_i$ for some $i\in \{1,\ldots,r\}$. 
We recall that country~$1$ misreports by claiming a smaller value for $\segv_1$ or $\sv_1$. We also recall that ${\mathcal A}_i^*\subseteq {\mathcal A}_i$ and that
by definition, every substitute has at most one $V_1$-segment. This means that the reason that a substitute in ${\mathcal A}_i\setminus {\mathcal A}_i^*$ is not in ${\mathcal A}_i^*$ (i.e., is not a $\Gamma'$-cycle) is that its only $V_1$-segment was larger than the value of $\sv_1$ reported by country~1.
Hence, if we compare the average utility for country~1 of the cycles in ${\mathcal A}_i$ that are not $\Gamma'$-cycles with the average utility of the $\Gamma'$-cycles in ${\mathcal A}_i$, then we find that the value of the latter is at most the value of the former:
\begin{equation}\frac{\sum_{C\in {\mathcal A}_i^*}u_1(C)}{|{\mathcal A}_i^*|} \leq \frac{\sum_{C\in {\mathcal A}_i\setminus {\mathcal A}_i^*}u_1(C)}{|{\mathcal A}_i\setminus {\mathcal A}_i^*|}.\end{equation}
We use inequality~(4) in the following deduction:
\[\begin{array}{lcl}
\displaystyle\frac{\sum_{C\in {\mathcal A}_i^*}u_1(C)}{|{\mathcal A}_i^*|} &= &\displaystyle\frac{\sum_{C\in {\mathcal A}_i^*}u_1(C)}{|{\mathcal A}_i^*|}\bigg(\frac{|{\mathcal A}_i^*|}{|{\mathcal A}_i|}+\frac{|{\mathcal A}_i|-|{\mathcal A}_i^*|}{|{\mathcal A}_i|}\bigg)\\[20pt]
&= &\displaystyle\frac{|{\mathcal A}_i^*|}{|{\mathcal A}_i^*|} \displaystyle\frac{\sum_{C\in {\mathcal A}_i^*}u_1(C)}{|{\mathcal A}_i|} +  \frac{(|{\mathcal A}_i|-|{\mathcal A}_i^*|)}{|{\mathcal A}_i|}\displaystyle\frac{\sum_{C\in {\mathcal A}_i^*}u_1(C)}{|{\mathcal A}_i^*|}\\[20pt]
&\leq &\displaystyle\frac{|{\mathcal A}_i^*|}{|{\mathcal A}_i^*|} \displaystyle\frac{\sum_{C\in {\mathcal A}_i^*}u_1(C)}{|{\mathcal A}_i|} +  \frac{(|{\mathcal A}_i|-|{\mathcal A}_i^*|)}{|{\mathcal A}_i|} \frac{\sum_{C\in {\mathcal A}_i\setminus {\mathcal A}_i^*}u_1(C)}{|{\mathcal A}_i\setminus {\mathcal A}_i^*|}\\[20pt]
&= & \displaystyle\frac{\sum_{C\in {\mathcal A}_i}u_1(C)}{|{\mathcal A}_i|}.
\end{array}\]
From the above, inequality (3) directly follows. Thus, $\mech_{{\tt order}}$ is IC, completing the proof.
\end{proof}

\begin{remark}
\textup{
We observe that we do not necessarily need to pre-select a maximum national $\Gamma$-cycle packing in Step~1 of $\mech_{{\tt order}}$. For example, let the set of maximum national $\Gamma$-cycle packings of an IKEP-instance $I=(G,\V,\Gamma)$ be $\{{\mathcal C}^{{\tt nat}}_1,\ldots,{\mathcal C}^{{\tt nat}}_r\}$.
We can modify Step~1 by picking a maximum national $\Gamma$-cycle packing ${\mathcal C}^{{\tt nat}}_i$ of $I$
with probability $p^*({\mathcal C}^{{\tt nat}}_i)$ for some probability distribution $p^*$ defined on $\{{\mathcal C}^{{\tt nat}}_1,\ldots,{\mathcal C}^{{\tt nat}}_r\}$.
See Section~\ref{a-remark2} for a proof. }\dia
\end{remark}

Recall that for an IKEP instance $I= \langle G, \V, \Gamma\rangle$, we denote the length of a longest international $\Gamma$-cycle in $G$ by $c_{{\tt int}}$ and the maximum international cycle degree of $G$ by $d^*$. 

\begin{theorem}\label{t-bounds}
$\mech_{{\tt order}}$ has an $\alpha$-approximation ratio with $\alpha \leq \max\{c_{{\tt int}},d^*\}$, which is tight and asymptotically optimal.
\end{theorem}

\begin{proof}
Let $I= \langle G, \V, \Gamma\rangle$ be an IKEP instance.
Let ${\mathcal C}^*$ be a maximum $\Gamma$-cycle packing for $I$, so $||{\mathcal C^*}|| =${\tt opt}$(I)$. Let ${\mathcal C}^*_{{\tt nat}}$ be the national $\Gamma$-cycle packing of $G$ that consists of all national $\Gamma$-cycles in ${\mathcal C}^*$. Let ${\mathcal C}^*_{{\tt int}}={\mathcal C}^*\setminus  {\mathcal C}^*_{{\tt nat}}$ be the international $\Gamma$-cycle packing of $G$ that consists of all international $\Gamma$-cycles in ${\mathcal C}^*$. 
In Step~2, $\mech_{{\tt order}}$ selects a maximum national $\Gamma$-cycle packing ${\mathcal C}_{{\tt nat}}$. We construct a $\Gamma$-cycle packing ${\mathcal C}'$ by adding every international cycle of ${\mathcal C^*}$ to ${\mathcal C}_{{\tt nat}}$ that does not share a vertex with any (national) $\Gamma$-cycle in ${\mathcal C}_{{\tt nat}}$. We let ${\mathcal C}^{**}_{{\tt int}} = {\mathcal C}'\setminus {\mathcal C}_{{\tt nat}}$.
In the worst case, each (national) $\Gamma$-cycle~$C$ in ${\mathcal C}_{{\tt nat}}$ blocks $c_{{\tt int}}|V(C)|$ cycles from ${\mathcal C}$ 
from joining ${\mathcal C'}$. Hence, we find that
\begin{equation}
||{\mathcal C}_{{\tt nat}}|| - ||{\mathcal C}^*_{{\tt nat}}|| \geq \frac{1}{c_{{\tt int}}}\big(||{\mathcal C}^*_{{\tt int}}||-||{\mathcal C}^{**}_{{\tt int}}||\big) \Leftrightarrow ||{\mathcal C}_{{\tt nat}}|| + \frac{1}{c_{{\tt int}}}||{\mathcal C}^{**}_{{\tt int}}|| \geq 
 ||{\mathcal C}^*_{{\tt nat}}||+\frac{1}{c_{{\tt int}}}||{\mathcal C}^*_{{\tt int}}||
\end{equation}
Recall that $\mech_{{\tt order}}$ returns ${\mathcal C}_{{\tt nat}}\cup {\mathcal C}$, where ${\mathcal C}$ is an international $\Gamma$-cycle packing whose construction is completed at the end of Step~5.
As $\mech_{{\tt order}}$ selects each international $\Gamma$-cycle with probability at least $\frac{1}{d^*}$, we also find that
\begin{equation}
||{\mathcal C}|| \geq \frac{1}{d^*}||{\mathcal C}^{**}_{{\tt int}}||.
\end{equation}
We first assume that $d^*\leq c_{{\tt int}}$. Using (5) and (6) we get:
\begin{equation}
||{\mathcal C}^*_{{\tt nat}}\cup {\mathcal C}|| \geq  ||{\mathcal C}^*_{{\tt nat}}|| + \frac{1}{d^*}||{\mathcal C}^{**}_{{\tt int}}||\geq ||{\mathcal C}^*_{{\tt nat}}|| + \frac{1}{c_{{\tt int}}}||{\mathcal C}^{**}_{{\tt int}}||\geq  ||{\mathcal C}_{{\tt nat}}^*|| + \frac{1}{c_{{\tt int}}}||{\mathcal C}^{*}_{{\tt int}}|| \geq \frac{1}{c_{{\tt int}}}||{\mathcal C}^*||. 
\end{equation}
Now we assume that $d^*>c_{{\tt int}}$. Again, we use (5) and (6), but this time we get:
\begin{equation}
\begin{array}{lclclcl}
||{\mathcal C}^*_{{\tt nat}}\cup {\mathcal C}|| &\geq  &||{\mathcal C}^*_{{\tt nat}}||+\displaystyle\frac{1}{c_{{\tt int}}}||{\mathcal C}^*_{{\tt int}}|| &\geq  &||{\mathcal C}^*_{{\tt nat}}||+\displaystyle\frac{1}{d^*}||{\mathcal C}^*_{{\tt int}}||\\[8pt]
&\geq  
&\displaystyle\frac{c_{{\tt int}}}{d^*}\bigg(||{\mathcal C}^*_{{\tt nat}}||+\displaystyle\frac{1}{c_{{\tt int}}}||{\mathcal C}^*_{{\tt int}}||\bigg)&\geq 
&\displaystyle\frac{c_{{\tt int}}}{d^*}\displaystyle\frac{1}{c_{{\tt int}}}||{\mathcal C}^*|| &=
&\displaystyle\frac{1}{d^*}||{\mathcal C}^*||.
\end{array}
\end{equation}
Due to (7) and (8),  $\mech_{{\tt order}}$ has an $\alpha$-approximation ratio with $\alpha \leq \max\{c_{{\tt int}},d^*\}$.
We claim that this bound is tight. First, we can take a national cycle $C=\langle u,v\rangle$ of length~$2$ and identify $u$ with a vertex of an international cycle $D_u$, and $v$ with a vertex of another international cycle $D_v$. Let $D_u$ and $D_v$ both have length $c_{{\tt int}}\geq 2$. Second, we can take an international cycle with three vertices $u_1,u_2,u_3$ from country~1 and one vertex $v$ from country~2. We then identify $v$ with a unique vertex of $d^*$ international cycles $D_1,\ldots, D_{d^*}$, in all of which countries~1 and~2 alternate twice, so $c_{{\tt int}}=4$. We now choose $d^*\geq 4$, and we assume country~$1$ misreports its segment number $\sv_1=2$ as $\sv_1'=1$.

Finally, to show that $\mech_{{\tt order}}$ is asymptotically optimal, we can set $c_{{\tt int}}\geq d^*$ in the construction of Proposition~\ref{p-alpha1} (displayed in Figure~\ref{fig:IRalpha}).
\end{proof}

\subsection{Alternative IR and IC Mechanisms}\label{a-remark2}

In Remark~2, we claimed
that we do not necessarily need to pre-select a maximum national $\Gamma$-cycle packing in Step~1 of $\mech_{{\tt order}}$. For example, we could let the set of maximum national $\Gamma$-cycle packings of an IKEP-instance $I=(G,\V,\Gamma)$ be $\{{\mathcal C}^{{\tt nat}}_1,\ldots,{\mathcal C}^{{\tt nat}}_r\}$.
We said we can modify Step~1 by picking a maximum national $\Gamma$-cycle packing ${\mathcal C}^{{\tt nat}}_i$ of $I$
with probability $p^*({\mathcal C}^{{\tt nat}}_i)$ for some probability distribution $p^*$ defined on $\{{\mathcal C}^{{\tt nat}}_1,\ldots,{\mathcal C}^{{\tt nat}}_r\}$.
We prove this statement below.

Let for every $i\in \{1,\ldots,r\}$, the set $\{{\mathcal C}^i_1,\ldots, {\mathcal C}^i_{s_i}\}$ be the set of $\Gamma$-cycle packings in the graph $\hat{G}_i$ that is obtained from $G$ after deleting all the vertices of the cycles in ${\mathcal C}^{{\tt nat}}_i$ from $G$. Let $p^i_h({\mathcal C}^i_h)$ be the probability that $\mech_{{\tt order}}$ selects international $\Gamma$-cycle packing ${\mathcal C}^i_h$ if it first selected ${\mathcal C}^{{\tt nat}}_i$.
The expected utility for country~$j\in \N$ is
$$U_j(\mech_{{\tt order}}(I)) = 
\displaystyle\sum_{i=1}^1\bigg(
p^*({\mathcal C}^{{\tt nat}}_i)u_j({\mathcal C}^{{\tt nat}}_i) + p^*({\mathcal C}^{{\tt nat}}_i)\cdot \sum_{h=1}^{s_i}p^i_h({\mathcal C}^i_h)u_j({\mathcal C}^i_h)\bigg).$$ 
Assume country~$1$ misreports either its segment size $\segv_1$ or its segment number $\sv_1$ by choosing a smaller value. Let $\Gamma'$ be the new set of country parameters, and let $I'= \langle G, \V, \Gamma'\rangle $ be the new IKEP instance.
 Let $q^i_h({\mathcal C}^i_h)$ be the probability that $\mech_{{\tt order}}$ selects international $\Gamma$-cycle packing ${\mathcal C}^i_h$ if it first selected ${\mathcal C}^{{\tt nat}}_i$. So, now the expected utility for country~$j\in \N$ is
$$U_j(\mech_{{\tt order}}(I')) = 
\displaystyle\sum_{i=1}^1\bigg(
p^*({\mathcal C}^{{\tt nat}}_i)u_j({\mathcal C}^{{\tt nat}}_i) + p^*({\mathcal C}^{{\tt nat}}_i)\cdot \sum_{h=1}^{s_i}q^i_h({\mathcal C}^i_h)u_j({\mathcal C}^i_h)\bigg).$$
In the proof of Theorem~\ref{t-ms4}, we showed
for every fixed $1\leq i \leq r$, that
$$\sum_{h=1}^{s_i}q^i_h({\mathcal C}^i_h)u_1({\mathcal C}^i_h)
\leq
\sum_{h=1}^{s_i}p^i_h({\mathcal C}^i_h)u_1({\mathcal C}^i_h).$$
Hence, we obtain $U_1(\mech_{{\tt order}}(I'))\leq U_1(\mech_{{\tt order}}(I))$, implying that the modified version of $\mech_{{\tt order}}$ is still IC.
We also note that this modification still satisfies the bounds in Theorem~\ref{t-bounds}.

\subsection{Manipulating the National Cycle Limit}\label{sec:nclmanip}

\begin{figure}[t]
\centering
\resizebox{0.55\columnwidth}{!}{
\begin{tikzpicture}
\begin{scope}[ every node/.style={draw,rectangle,thick,minimum size=7mm}]
    \node  (x1) at (8,1) {$h_1$};
    \node  (x2) at (4,-0.5) {$h_2$};
    \node  (x3) at (4,-1.5) {$h_3$};
    \node  (x4) at (6,-1) {$h_4$};
    
\end{scope}

  \begin{scope}[every node/.style={draw,circle,thick,minimum size=8mm}]  
    \node (y2) at (6,1) {$j_2$};
    \node (y3) at (4,0.5) {$j_3$};
    \node (y1) at (4,1.5) {$j_1$};
    \node (y4) at (2,1) {$j_4$};
    \node (y5) at (0,0) {$j_5$};
    \node (y6) at (2,-1) {$j_6$};
\end{scope}

\begin{scope}[>={Stealth[black]},
              every edge/.style={draw, very thick}]
    \path [->] (y2) edge[bend right=15] node {} (y1);
    \path [->] (y1) edge[bend right=15] node {} (y4);
    \path [->] (y4) edge[bend right=15] node {} (y3);
    \path [->] (y3) edge[bend right=15] node {} (y2);
    \path [->] (y5) edge[bend left=15] node {} (y4);
    \path [->] (y6) edge[bend left=15] node {} (y5);

    \path [->] (y4) edge node {} (y6);

    \path [->] (x3) edge[bend right=15] node {} (x4);
    \path [->] (x4) edge[bend right=15] node {} (x2);

    \path [->] (y6) edge[dashed, bend right=15] node {} (x3);
    \path [->] (x2) edge[dashed, bend right=15] node {} (y6);
    \path [->] (y1) edge[dashed, bend left=30] node {} (x1);
    \path [->] (x1) edge[dashed, bend left=30] node {} (y3);
\end{scope}
\end{tikzpicture}}
\Description{}
\caption{The graph $\pG$ showing that $\Morder$ is not IC with respect to $\natv$. Note that $\N=\{H,J\}$ and the vertices in $V_H$ and $V_J$ are represented by square and circular vertices, respectively.  } 
\label{fig:MordernotICnatv}
\end{figure}

Theorem \ref{t-ms4} showed that $\Morder$ is IR and is also IC with respect to $\segv$ and $\sv$. However, the following example shows that $\Morder$ is not IC with respect to $\natv$. We assume that the true value of $\natv_i$ for country $i$ is constrained for logistical reasons, and hence a country can only misreport its national cycle limit by giving a lower value.

\begin{example}
 \textup{
Let $\N=\{H,J\}$ with vertices $V_H=\{h_1, \ldots, h_4\}$ and $V_J=\{j_1, \ldots, j_6\}$ and arcs given in Figure~\ref{fig:MordernotICnatv}. The vertices in $V_H$ and $V_J$ are represented in Figure~\ref{fig:MordernotICnatv} using squares and circles, respectively. Solid lines represent national arcs, whereas dashed lines represent international arcs.
Suppose that $\intv=4$, $\natv_H\in \{0,1,\ldots\} \cup \{\infty\}$ (there are no national cycles for $H$ and thus $\natv_H$ can have any value), $\natv_J=4$, $\segv_H,\segv_J\geq 2$ and $\sv_H,\sv_J\geq 1$. 
Observe that $\Morder$ and $\Mnat$ will both first select the cycle $\langle j_1, j_4, j_3, j_2\rangle$. 
In the second step of $\Morder$, the cycle $\langle h_2, j_6, h_3, h_4\rangle$ is selected. The total number of transplants given by $\Morder$ is $8$. However, country $J$ has one pair, namely $j_5$, that is left uncovered by $\Morder$. 
}

\textup{
Suppose now that $\natv_J'=3$, i.e., country $J$ misreports its national cycle limit to be $3$ instead of~$4$. In this case, observe that $\Morder$ and $\Mnat$ will instead select the national cycle $\langle j_4, j_6, j_5\rangle$. In the second step of $\Morder$, the cycle $\langle j_1, h_1, j_3, j_2 \rangle$ will now be selected. In this case, the total number of transplants has decreased to $7$; however, all pairs in $V_J$ are covered by $\Morder$.
}

\textup{
Hence, $\Morder$ is not IC with respect to $\natv$, as country $J$ has misreported $\natv_J$ to strictly increase its number of pairs covered by $\Morder$. 
}\dia
\end{example}

We note that on the other hand, $\Mnat$ is trivially IR and IC with respect to $\natv$, $\segv$ and $\sv$.

\section{Simulations}\label{sec:simulations}

\begin{figure}[b!]
\centering
\begin{subfigure}{.45\textwidth}
    \centering
    \includegraphics[width=.8\linewidth]{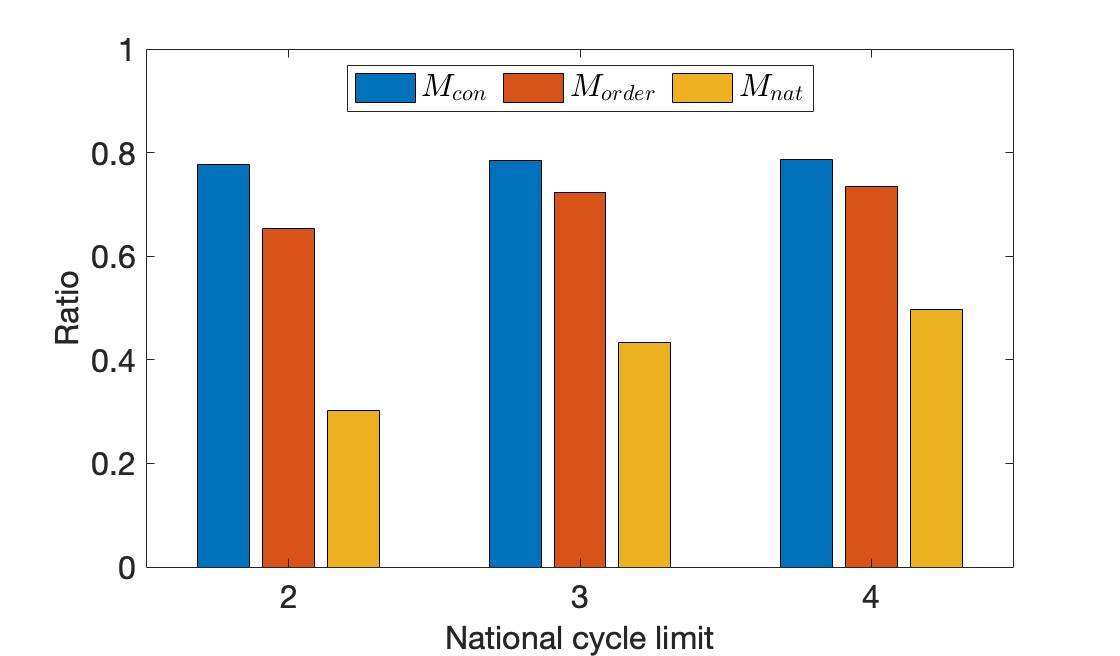}  
    \Description{}
    \caption{Performance when varying  $\natv$ when fixing other parameters $\intv=4,\segv=\{2\}^3,\sv=\{1\}^3$.}
    \label{sub:k}
\end{subfigure}
\hspace{5mm}
\begin{subfigure}{.45\textwidth}
    \centering
    \includegraphics[width=.8\linewidth]{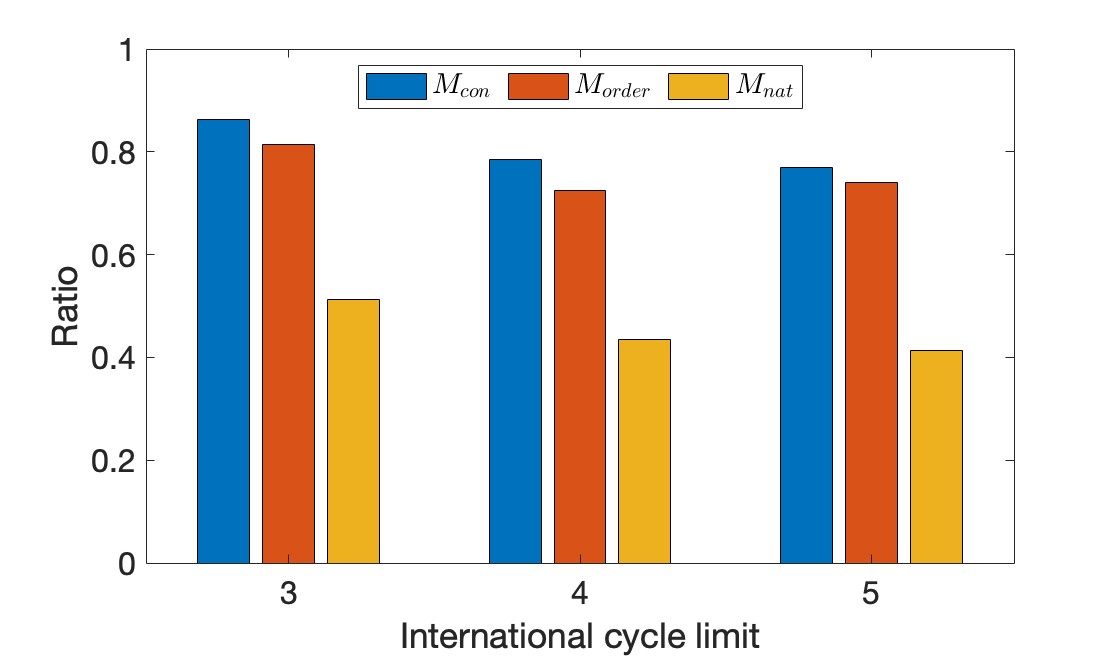}  
    \Description{}
    \caption{Performance when varying $\intv$ when fixing other parameters $\natv=\{3\}^3,\segv=\{2\}^3,\sv=\{1\}^3$.}
    \label{sub:l}
\end{subfigure}
\begin{subfigure}{.47\textwidth}
    \centering
    \includegraphics[width=\linewidth]{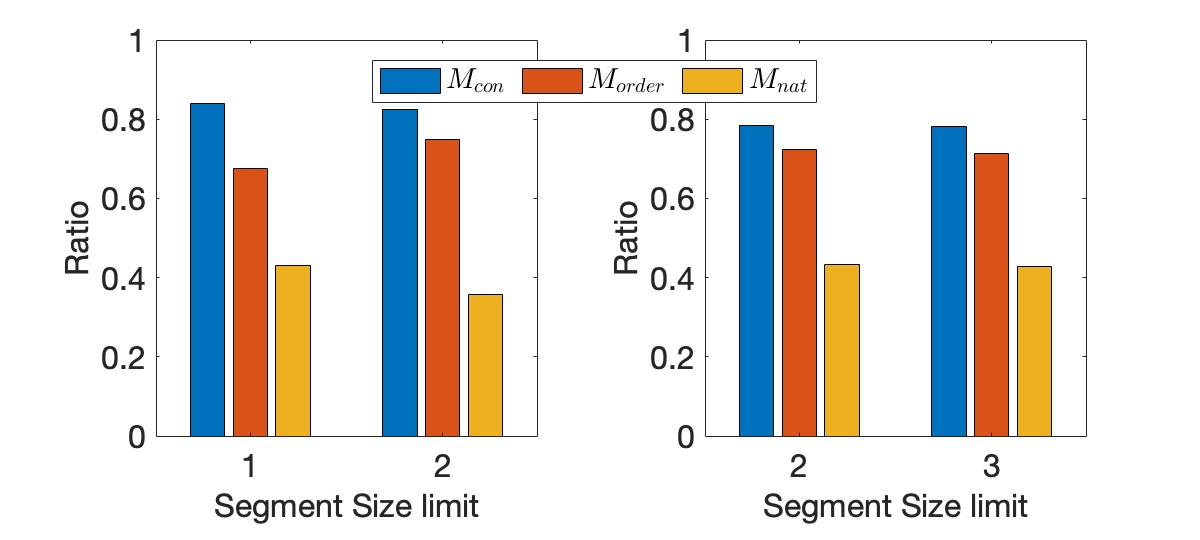} 
    \Description{}
    \caption{Performance when varying $\segv$. 
     In the left figure, we change $\segv$ from $\{1\}^3$ to $\{2\}^3$, fixing $\natv=\{2\}^3,\intv=3,\sv=\{1\}^3$. In the right figure, we change $\segv$ from $\{2\}^3$ to $\{3\}^3$, fixing $\natv=\{3\}^3,\intv=4,\sv=\{1\}^3$. }
    \label{sub:c}
\end{subfigure}
\hspace{5mm}
\begin{subfigure}{.47\textwidth}
    \centering
    \includegraphics[width=\linewidth]{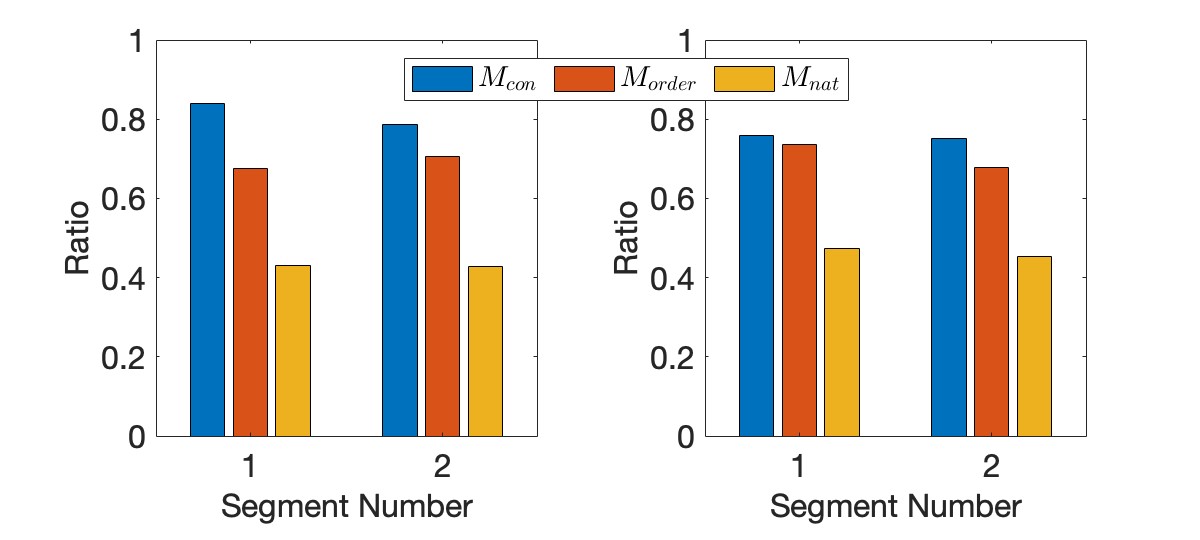} 
    \Description{}
    \caption{Performance when varying $\sv$. In the left figure, we change $\sv$ from $\{1\}^3$ to $\{2\}^3$, fixing $\natv=\{3\}^3,\intv=4,\segv=\{2\}^3$. In the right figure, we change $\segv$ from $\{2\}^3$ to $\{3\}^3$, fixing $\natv=\{4\}^3,\intv=5,\sv=\{2\}^3$.}
    \label{sub:s}
\end{subfigure}
\Description{}
\caption{Performance of $\Mcon, \Morder,\Mnat$ compared with $\Mint$ when varying a single parameter.}
\label{fig:bigSims}
\end{figure}
In this section, we present the results of simulations that were carried out to investigate the performance in practice of $\Morder$ against the three other mechanisms $\Mnat$, $\Mcon$ and $\Mint$ we discussed in Section~\ref{s-intro}.
We recall their definitions. First, $\Mnat$ is the \emph{national} mechanism, which finds a maximum $\Gamma$-cycle packing for each national country pool.
Second, $\Mcon$ is the \emph{consecutive} mechanism, which finds a maximum $\Gamma$-cycle packing for each national country pool, and then finds a maximum $\Gamma$-cycle packing among the remaining uncovered recipient-donor pairs in the international pool.  Note that $\Mcon$ is IR, and it also initially finds a maximum $\Gamma$-cycle packing for each national country pool. Third, $\Mint$ is the \emph{international} mechanism, which finds a maximum $\Gamma$-cycle packing for the entire pool (including the national and international arcs together). 
We measure the performance of the mechanisms as the proportion of transplants found by each mechanism against the maximum number found by $\Mint$ and investigate the effects of country-specific parameters on the performance of $\Morder$.

Our simulations were conducted in Python using the \texttt{kep\_solver} package~\cite{pettersson2022kep}. 
We created instances with the generator described by \citet{delorme2022improved}, using their blood group distributions for recipients and donors and their \emph{TweakXMatch-PRA0} compatibility distribution to obtain consistency with key characteristics of real data from the UK KEP.  We generated $15$ instances, each with $300$ recipient-donor pairs, to be consistent with the typical current size of UK KEP datasets.
We let $n=3$, as many current IKEP collaborations involve a small number of countries, e.g., the \emph{KEPSAT} collaboration between Italy, Portugal, and Spain \citep{valentin2019international}. Indeed, the simulations conducted by \citet{mincu2021ip} and \citet{druzsin2024performance} also assume $n=2$ and $n=3$ respectively.  
The vertices of $\V$ are therefore partitioned into three sets $V_1, V_2, V_3$. To test the effect of the country size on the performance of our mechanism, we considered two cases: (1) countries all have the same size, and (2) countries have different sizes with a ratio of $3:2:1$.
We ensured that $\Mnat$, $\Mcon$ and $\Morder$ all chose the same national cycle packing (in their first steps in the last two cases), as constructed using the ILP solver.  We note that the number of transplants found by $\Mcon$ will in general be greater than that found by $\Morder$, yet  $\Mcon$ is not IC. 
To counteract the randomness of $\Morder$, we ran $\Morder$ five times per instance and report the average number of transplants selected.

We focused on values of country-specific parameters for our countries $\N=\{1,2,3\}$ that satisfied $\segv_i \leq \natv_i \leq \intv$ for all $i\in \mathcal{N}$. We motivate this rule through practical considerations, i.e., that the maximum size of an international segment for a given country is not expected to exceed its national cycle limit, which in turn is not expected to exceed the international cycle limit.
When varying country-specific parameters, we focused on the following four analyses: {\bf (A)} to test the effect of varying the national cycle limit, we considered $\natv=\{2\}^3,$ $\{3\}^3, \{4\}^3$, while fixing $\intv, \segv, \sv$; {\bf (B)} to test the effect of varying the international cycle limit, we considered $\intv = 3,4,5$ while fixing $\natv, \segv, \sv$; {\bf (C)} to test the effect of varying the segment size limit, we considered $\segv = \{1\}^3, \{2\}^3$, and  $\{3\}^3$ while fixing  $\natv, \intv, \sv$; and {\bf (D)} to test the effect of varying the segment number, we considered $\sv= \{1\}^3$, $\{2\}^3$ whilst fixing  $\natv, \intv, \segv$.   When fixing country-specific parameters, the following baseline values are used unless otherwise stated: $\natv=\{3\}^3,\intv=4,\segv=\{2\}^3,\sv=\{1\}^3$.

\paragraph{Simulation Results for Equal Country Sizes.}
We begin by presenting results for case (1) where all countries have the same number of vertices.  Figures~\ref{sub:k} and~\ref{sub:l} show the proportion of transplants (relative to $\Mint$) achieved by $\Morder$, $\Mcon$ and $\Mnat$ when (A) we vary \natv, and (B) we vary \intv, with the other variables fixed at their baseline values.  On average, $\Morder$ selects $72\%$ of pairs covered by $\Mint$, with $\Mcon$ covering $79\%$ of pairs covered by $\Mint$. Moreover, on average, $\Morder$ covers  $91\%$ of the number of transplants found by $\Mcon$.
These percentages give an indication of the price of ensuring IR (which $\Mint$ does not satisfy in general) and IC (which both $\Mint$ and $\Mcon$ do not satisfy in general), both of which are satisfied by $\Morder$.  These results also suggest that in practice, $\Morder$'s performance is substantially superior to its theoretical worst-case performance guarantee.
We also note that countries benefit significantly from joining an IKEP, since $\Morder$ covers $84\%$ more pairs than $\Mnat$ on average.

We next report on the performance of $\Morder$ and other mechanisms when varying a single country-specific parameter, whilst fixing the remaining parameters.  
(A) Figure~\ref{sub:k} shows the effect of varying the national cycle limit.  
As $\natv$ increases, the performance of $\Mcon$ and $\Morder$ marginally improves, while $\Mnat$ performs significantly better.
(B) Figure~\ref{sub:l}, shows the effect of varying the international cycle limit. 
As $\intv$ increases, the three mechanisms perform slightly worse. One explanation is that $\Mint$ can utilise the additional international cycle length more readily whilst the other mechanisms are more constrained by the need to satisfy IR.  
(C)  The effect of varying segment size limits $\segv$ is shown in Figure~\ref{sub:c}. In the left of Figure~\ref{sub:c}, we change $\segv$ from $\{1\}^3$ to $\{2\}^3$, whilst fixing $\natv=\{2\}^3,\intv=3,\sv=\{1\}^3$. In the right of Figure~\ref{sub:c}, we change $\segv$ from $\{2\}^3$ to $\{3\}^3$, whilst fixing $\natv=\{3\}^3,\intv=4,\sv=\{1\}^3$. There is no clear pattern when increasing  $\segv$, perhaps showing that the number of transplants selected increases proportionally with $\Mint$ or that more instances should be checked. 
(D)  In Figure~\ref{sub:s}, we examine the effect of increasing the segment number from $\sv{=}\{1\}^3$ to $\sv{=}\{2\}^3$. Due to the dependencies between the parameters, the other parameters have two sets of values, namely $\natv{=}\{3\}^3, \intv=4, \segv{=}\{2\}^3$ and $\natv=\{4\}^3, \intv=5, \segv=\{2\}^3$. We do not consider the case where $\sv=\{3\}^3$ when $\intv=4,5$, as the international cycle limit is not large enough to allow 
for this. 
As $\sv$ increases from $\{1\}^3$ to $\{2\}^3$, there is no significant change in the number of transplants, however when $\sv$ increases from $\{2\}^3$ to $\{3\}^3$, the proportion of transplants selected relative to $\Mint$ decreases. 

\paragraph{Simulation Results for Unequal Country Sizes}\label{app:unequal}

\begin{figure}[t!]
    \centering
    \includegraphics[width=1\linewidth]{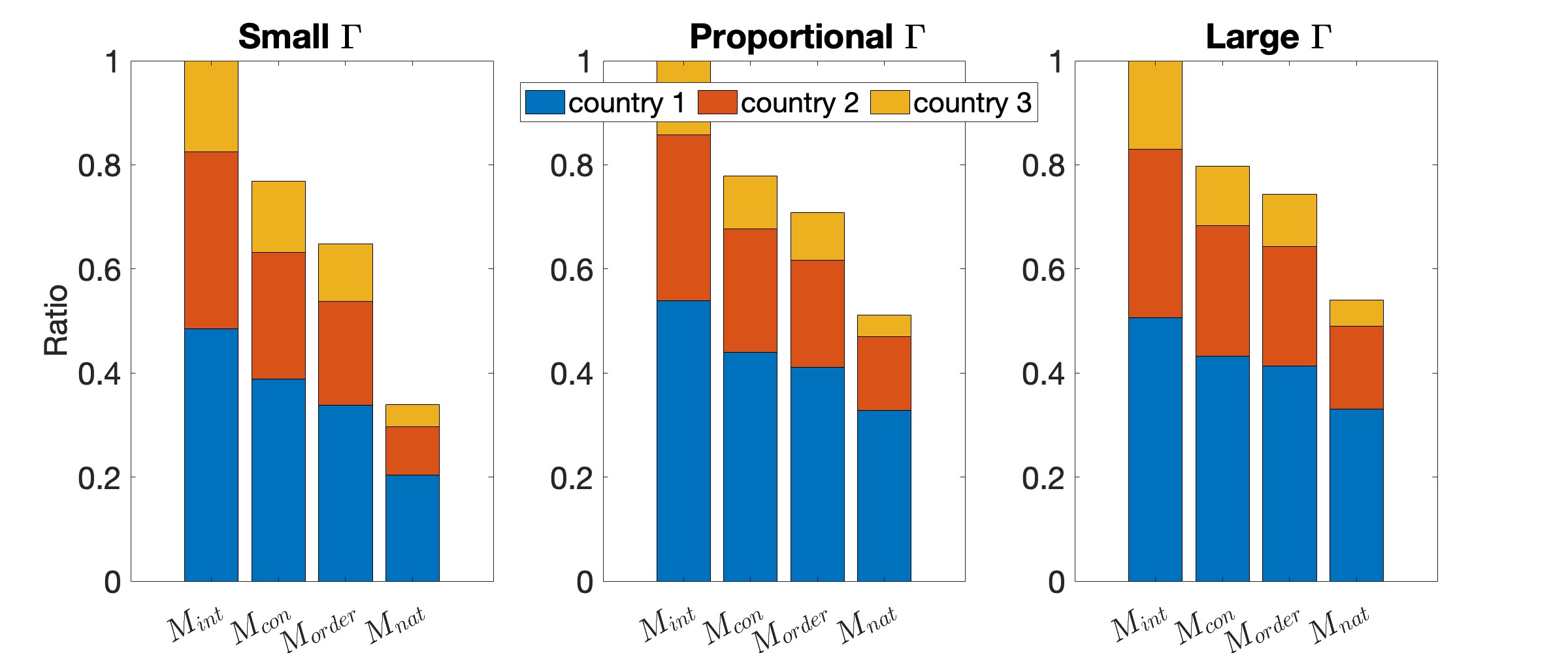}
    \Description{}
    \caption{Performance of $\Mint, \Mcon, \Morder,\Mnat$. The results on \textbf{Small $\Gamma$}, \textbf{Proportional}, and \textbf{Large $\Gamma$} are shown in the left, middle and right plots, respectively.}
    \Description{}
    \label{fig:uneven}
\end{figure} 
A second aim of our simulations was to investigate the impact of countries of different sizes. We study the case where $|\N|=3$ and the countries' sizes are determined by the ratio $3:2:1$. That is, country~$1$ has half of the recipient-donor pairs, country~$2$ has a third of the recipient-donor pairs, and country~$3$ has a sixth of the recipient-donor pairs. 
We studied three sets of parameters. The first of parameters, denoted by \textbf{Small $\Gamma$}, lets all countries have the same smaller parameters: $\natv=\{2\}^3$, $\sv=\{1\}^3$, $\segv=\{2\}^3$, and $\intv=4$. The second set of parameters, \textbf{Proportional $\Gamma$} sets the parameter's values more proportionally with respect to the countries' sizes: $\natv=(4,3,2)$, $\sv=\{1\}^3$, $\segv=(3,2,1)$, and $\intv=4$. This can be thought of as larger countries having more medical facilities allowing more transplants to proceed.  The final set of parameters, denoted by \textbf{Large $\Gamma$}, lets all countries have a larger set of parameters: $\natv=\{4\}^3$, $\sv=\{1\}^3$, $\segv=\{2\}^3$, and $\intv=4$. 

Figure~\ref{fig:uneven} depicts the average proportion of transplants selected by $\Mint$, $\Mcon$, $\Morder$, and $\Mnat$, broken down by country, for our three sets of parameters. 

It is clear in all cases that participating in the IKEP is beneficial for all countries and all sets of parameters as more transplants are selected in each case. 
Proportionally, we see that the smallest country, country~$3$, benefits the most when the parameters are smaller, i.e., \textbf{Small $\Gamma$}. The proportional parameters \textbf{Proportional $\Gamma$} significantly benefit the largest country the most for each method of selecting $\Gamma$-cycle packings.

In Table~\ref{tab:unequalratio}, we display the average ratio between the transplants selected by one mechanism to another. We observe that $\Morder$ and $\Mcon$ with the \textbf{Large $\Gamma$} consistently yield proportionally more transplants, with respect to $\Mint$, than the other sets of country-specific parameters. A similar pattern is observed for $\Morder$ with \textbf{Large $\Gamma$}, with respect to $\Mcon$.

\begin{table}
    \centering
    \begin{tabular}{|c|ccc|} \hline
    & \textbf{Small $\Gamma$} & \textbf{Proportional $\Gamma$} & \textbf{Large $\Gamma$}\\
         \hline
$\Morder:\Mint$& $64\%$ & $71\%$ & $74\%$ \\
 $\Mcon:\Mint$& $77\%$ & $78\%$  &$80\%$\\
 $\Morder:\Mcon$& $84\%$ & $91\%$ &$93\%$ \\ \hline
    \end{tabular}
    \caption{The ratio between the number of transplants selected between the different mechanisms for each of our different values of $\Gamma$. For example, the first row shows the average number of transplants selected by $\Morder$ over the average number selected by $\Mint$.  }
    \label{tab:unequalratio}
\end{table}

\section{Conclusion}\label{sec:conc}
We have addressed our research questions concerning computational complexity and manipulation for IKEPs with a set $\Gamma$ of country-specific parameters.
We completely classified the complexity of {\sc Max $\Gamma$-Cycle Packing}, namely for all values of $\Gamma$. We demonstrated that manipulation through misreporting country-specific parameters, such as segment size or segment number, is possible, which led us to propose a mechanism $\Morder$ that is IR and IC with respect to segment size and segment number.  By giving upper and lower bounds, we showed that the approximation ratio of $\Morder$ is asymptotically best possible.  Our simulations suggested that $\Morder$ performs well in relation to other mechanisms for IKEPs and does better in practice than its worst-case approximation ratio would suggest. 

Key future directions are to incorporate non-directed donors (chains) and compatible pairs (self-loops) and to study the weighted case, where we seek a maximum weight (rather than maximum size) solution. Another possibility is to investigate whether a similar lower bound to that established by Theorem~\ref{t-alpha2} holds if $\intv$ has a fixed small value (as long international cycles currently feature in the construction).
Additional objectives could also be considered, even optimised in a hierarchical way, as is the case in many European KEPs~\cite{biro2021european}.

%DP: added a new paragraph.
Further work could also focus on strategy-proofness. As illustrated in Section~\ref{sec:nclmanip}, it is possible for countries to manipulate $\Morder$ via their national cycle limit. Hence, a natural next step would be to find a mechanism that is IC with respect to $\natv$, $\segv$, and $\sv$.
Another direction for future work is to study the possibility of countries manipulating via the country-specific parameters and by hiding pairs from the IKEP pool. This could involve extending the results of \citet{ashlagi2014free} to these more complex forms of manipulation. 
Finally, the definition of $\Gamma$ could be extended to include the number of countries involved in an international cycle, and the total number of recipients from a given country that can be selected in any international cycle, as discussed in Section~\ref{s-intro}.

\begin{acks}
This work was supported by the Engineering and Physical Sciences Research Council [grant numbers EP/X01357X/1, EP/X013618/1].
\end{acks}
\section*{Data Availability Statement}
The data supporting the findings reported in Section~\ref{sec:simulations} of this paper are available under a CC-BY licence from the Enlighten Research Data repository at
\url{https://doi.org/10.5525/gla.researchdata.1989}. In particular, the repository contains the simulation code and the datasets generated for the simulations.

\bibliographystyle{ACM-Reference-Format}
\bibliography{bib.bib}

\clearpage

\appendix

\section{Proof of the Dichotomy Result: Theorem~\ref{thm:dico}}\label{app:complexity}

To prove Theorem~\ref{thm:dico} we provide 14 general lemmas that each prove a hardness case for various sets of country-specific parameters~$\Gamma$s. These are summarized in Table~\ref{tab:dicho}.

\begin{table}[t]
    \centering
    \resizebox{\columnwidth}{!}{
    \begin{tabular}{|c|c||c|c|c|c|c|}
    \hline
 Row
 No. & Lemma no.     & $n$  & $\intv$ &  $\natv$ & $\segv$ & $\sv$ \\

   &  {\tiny Min values} & {\tiny 1} & { \tiny 0} & {\tiny 0} & {\tiny 1} &{\tiny 1}  \\
     \hline
     \hline

      \refstepcounter{rownumber}\label{row:n=2l=2k=inftyAnycs} \rownumber  &  \ref{lem:n=2l=2k=inftyAnycs}   &2 &$2$  & $(N^\infty(3), N^\infty(0))$ & $(N^\infty(1),N^\infty(1))$ & $(N^\infty(1),N^\infty(1))$\\

    \refstepcounter{rownumber}\label{row:Gl=3c=3s=1n=3}  \rownumber&     \ref{lem:Gl=3c=3s=1n=3}  & 3&   $3$&  $(N^\infty(0), N^\infty(0), N^\infty(0))$&  $(N^\infty(1), N^\infty(1), N^\infty(1))$&$(N^\infty(1), N^\infty(1), N^\infty(1))$ \\

    \refstepcounter{rownumber}\label{row:n=2l=2k=inftyAnycsSEGV} \rownumber  &  \ref{lem:n=2l=2k=inftyAnycsSEGV}   &2 &$N^\infty(3)$  & $(N^\infty(3), N^\infty(0))$ & $(1,N^\infty(1))$ & $(N^\infty(1),N^\infty(1))$\\

    \refstepcounter{rownumber}\label{row:n=2l=2k=2Anycs} \rownumber  &     \ref{lem:n=2l=2k=2Anycs}  & 2 & $3$  & $(N^\infty(0), N^\infty(0))$ & $(N^\infty(2), N^\infty(1))$ & $(N^\infty(1), N^\infty(1))$\\

    \refstepcounter{rownumber}\label{row:Gl=3c=3s=1n=4}  \rownumber  &\ref{lem:Gl=3c=3s=1n=4}  & 3&   $N^\infty(4)$&  $(N^\infty(0), N^\infty(0), N^\infty(0))$&  $(N^\infty(1), N^\infty(1), N^\infty(1))$&$(1, N^\infty(1), N^\infty(1))$ \\

    \refstepcounter{rownumber}\label{row:n=2l=2k=2Anycs4}\rownumber  &  \ref{lem:n=2l=2k=2Anycs4}  & 2 & $N^\infty(4)$  & $(N^\infty(0), N^\infty(0))$ & $(N^\infty(2), N^\infty(1))$ & $(1, N^\infty(1))$\\

    \refstepcounter{rownumber}\label{row:n=2l=2k=inftyAnycsSV}\rownumber  &  \ref{lem:n=2l=2k=inftyAnycsSV}   &2 &$N^\infty(3)$  & $(N^\infty(3), N^\infty(0))$ & $(N^\infty(1),N^\infty(1))$ & $(1,N^\infty(1))$\\

    \refstepcounter{rownumber}\label{row:l=4k>=2c=1s>=2n=2}  \rownumber  &   \ref{lem:l=4k>=2c=1s>=2n=2}    & 2 &  $\infty$ & $(N^\infty(0), N^\infty(0))$ & $(N^\infty(1), N^\infty(1))$ & $(N(2), N^\infty(2))$\\
  \refstepcounter{rownumber}\label{row:l=4k>=2c=1s=2n=2}\rownumber  &  \ref{lem:l=4k>=2c=1s=2n=2}    & 2 &  $N^\infty(4)$ & $(N^\infty(0), N^\infty(0))$ & $(N^\infty(1), N^\infty(1))$ & $(N^\infty(2), N^\infty(2))$\\
     \refstepcounter{rownumber}\label{row:l=inftyk=0c=1s=inftyn=2} \rownumber  &    \ref{lem:l=inftyk=0c=1s=inftyn=2}     & 2 &  $N^\infty(4)$ & $(N^\infty(2), N^\infty(0))$ & $(1,N^\infty(1))$  & $(N^\infty(2), N^\infty(2))$ \\

     \refstepcounter{rownumber}\label{row:l=inftyk>=0c=1s>=2n=2} \rownumber  &    \ref{lem:l=inftyk>=0c=1s>=2n=2}      & 2 & $\infty$ &  $(N^\infty(0), N^\infty(0))$ & $(N(2), N^\infty(1))$  &  $(N^\infty(2), N^\infty(2))$\\

     \refstepcounter{rownumber}\label{row:l=inftyc=inftys=1,infty} \rownumber  &  \ref{lem:l=inftyc=inftys=1,infty}    & 2 & $\infty$ & $(N^\infty(0), N^\infty(0))$  &$(\infty, \infty)$  & $(N(1), N^\infty(1))$\\

     \refstepcounter{rownumber}\label{row:lk=intyc=mathbbs=infty}\rownumber  &  \ref{lem:lk=intyc=mathbbs=infty}        & 2 & $\infty$ &  $(N^\infty(2), N^\infty(2))$ &  $(N(1), N^\infty(1))$ & $(N^\infty(2), N^\infty(2))$\\

     \refstepcounter{rownumber}\label{row:l=inftyk=2c=inftys=inftyn=2}\rownumber  &   \ref{lem:l=inftyk=2c=inftys=inftyn=2}   & 2& $\infty$ & $(N(0), N^\infty(0))$ & $(\infty,\infty)$ &$(\infty,\infty)$ \\
     
    \hline
    \end{tabular}}\\[5pt]
    \caption{Summary of the hardness lemmas that are proven in Appendix~\ref{app:Lems}, where 
    $N(0)=\mathbb{N}_{\geq 0}\setminus \{1\}=\{0,2,3,\ldots\}$ and $N(0)^\infty= N(0) \cup \{\infty\}$ (as we exclude self-loops) and for $x\geq 1$,
    $N(x)= \mathbb{N}_{\geq x}=\{x,x+1,\ldots\}$ and $N^\infty(x)=N(x)\cup \{\infty\}$. 
   Note that some rows are overlapping, and that in any row entry we may swap the order of the two elements in a pair. As an example, Row~\ref{row:lk=intyc=mathbbs=infty} indicates that Lemma~\ref{lem:lk=intyc=mathbbs=infty} shows the hardness of {\sc Perfect $\Gamma$-Cycle Packing} for every set of country-specific parameters $\Gamma=(n,\intv,\natv,\segv,\sv)$ such that  $n=2$, $\intv=\infty$,  every $i\in \N$ satisfies $\natv_i, \sv_i\in\mathbb{N}_{\geq 2}\cup\{\infty\}$, and some $j \in \N$ has $\segv_j\in \mathbb{N}_{\geq 1}$ (i.e., a natural number greater than or equal to one, but cannot be unbounded) while the remaining $j' \in\N\backslash\{j\}$ have  $\segv_{j'}\in\mathbb{N}_{\geq 1}\cup\{\infty\}$.   }\label{tab:dicho}
\end{table}

The remainder of this section is organized as follows: we next prove the hardness lemmas, as presented in Table~\ref{tab:dicho} in Appendix~\ref{app:Lems}, after we give the case distinction in Appendix~\ref{app:Cases} to prove that there is, in fact, a dichotomy, together entails that Theorem~\ref{thm:dico} has been shown.

\subsection{Proofs of the Lemmas in Table~\ref{tab:dicho}}\label{app:Lems}

We first recall the problem that we are going to show is \NP-complete for various sets of country-specific parameters $\Gamma$.

\decisionprob{{\sc Perfect $\Gamma$-Cycle Packing}}{an $n$-partitioned graph $(G,\V)$.}{does $G$ have a perfect $\Gamma$-cycle packing?}

\noindent
We next establish that {\sc Perfect $\Gamma$-Cycle Packing} is a member of \NP.

\begin{lemma}\label{lem:member}
     {\sc Perfect $\Gamma$-Cycle Packing} is in \NP. 
\end{lemma}

\begin{proof}
    Consider an arbitrary instance of {\sc Perfect $\Gamma$-Cycle Packing} and a certificate  $\mathcal{C}$. We can verify if $\C$ is indeed a perfect $\Gamma$-cycle packing by checking in polynomial time if each $v \in V$ appears exactly once in $\C$ and if each $C\in \C$ is a $\Gamma$-cycle.
\end{proof}

\noindent
 We now recall the two \NP-complete problems from which we will reduce. The first is {\sc Perfect $3$-Dimensional Matching} ({\sc P3DM})~\cite{karp1972}.

\decisionprob{{\sc Perfect $3$-Dimensional Matching} ({\sc P3DM})}{three sets of elements $X, Y,$ $Z$ and a set of triples $T\subseteq X \times Y \times Z$.  }{is there a perfect matching $M \subseteq T$ such that for each $u\in X \cup Y \cup Z$, there exists exactly one $t \in M$ with $u \in t$? }

\noindent
The second is {\sc Perfect $4$-Dimensional Matching} ({\sc P4DM}). Recall that a simple reduction from {\sc P3DM} shows that {\sc P4DM} is also \NP-complete.

\noindent
\decisionprob{{\sc Perfect $4$-Dimensional Matching} ({\sc P4DM})}{four sets of elements $W,X, Y, Z$ and a set of quadruples $T\subseteq W\times X \times Y \times Z$.  }{is there a perfect matching $M \subseteq T$ such that for each $u\in W \cup X \cup Y \cup Z$, there exists exactly one $t \in M$ with
$u \in t$? }

Given an $n$-partitioned graph $(G,\V)$ and a tuple $\Gamma$ of country-specific parameters, we let $\gcyc$ denote the set of $\Gamma$-cycles in $(G,\V)$.  We are now ready to prove the lemmas from Table~\ref{tab:dicho}.
 
\begin{lemma}\label{lem:n=2l=2k=inftyAnycs}
For a set of country-specific parameters $\Gamma= (n, \emph{\intv}, \emph{\natv},\emph{\segv}, \emph{\sv})$ with 
$n \geq 2$,  
$\emph{\intv}=2$, 
$\emph{\sv}_i \in \mathbb{N}_{\geq 1}\cup\{\infty\}$ for all $i\in \N$,
$\emph{\segv}_i\in \mathbb{N}_{\geq 1}\cup\{\infty\}$  for all $i\in \N$,
$\emph{\natv}_j \in \mathbb{N}_{\geq 3}\cup\{\infty\}$ some $j\in\N$, and
$\emph{\segv}_i\in \mathbb{N}_{\geq 0}\cup\{\infty\}$ for all $i\in \N\setminus \{j\}$,
{\sc Perfect $\Gamma$-Cycle Packing} is \NP-complete.
\end{lemma}
\begin{proof}
\begin{figure*}[t]
\centering
\resizebox{0.35\columnwidth}{!}{
\begin{tikzpicture}[triangle/.style = {regular polygon, regular polygon sides=3 }]

\begin{scope}[every node/.style={circle,thick,draw,minimum size=0.9cm}]
    \node[regular polygon, thick, regular polygon sides=3, draw,
     inner sep=0.005cm] (a1) at (2.5,3) {$a_i^1$};
    \node (x) at (2.5,0) {$x$};
    \node (a3) at (2.5,7) {$a_i^2$};

    \node[regular polygon, thick, regular polygon sides=3, draw,
     inner sep=0.005cm] (a4) at (5,3) {$a_i^3$};
    \node (y) at (5,0) {$y$};
    \node (a6) at (5,6) {$a_i^4$};

    \node[regular polygon, thick, regular polygon sides=3, draw,
     inner sep=0.005cm] (a7) at (7.5,3) {$a_i^5$};
    \node (z) at (7.5,0) {$z$};
    
    \node (a9) at (7.5,7) {$a_i^6$};

\end{scope}

 \begin{scope}[>={Stealth[black]},
               every edge/.style={draw,very thick}]
   
    \path [->] (x) edge[bend left] node {} (a1);
    \path [->] (a1) edge[bend left] node {} (x);
    \path [->] (a1) edge[bend left] node {} (a3);
    \path [->] (a3) edge[bend left] node {} (a1);

    \path [->] (y) edge[bend left] node {} (a4);
    \path [->] (a4) edge[bend left] node {} (y);
    \path [->] (a4) edge[bend left] node {} (a6);
    \path [->] (a6) edge[bend left] node {} (a4);

    \path [->] (z) edge[bend left] node {} (a7);
    \path [->] (a7) edge[bend left] node {} (z);
    \path [->] (a7) edge[bend left] node {} (a9);
    \path [->] (a9) edge[bend left] node {} (a7);

    \path [->] (a3) edge[] node {} (a6);
    \path [->] (a6) edge[] node {} (a9);
    \path [->] (a9) edge[] node {} (a3);

\end{scope}
\end{tikzpicture}}
\Description{}
    \caption{Gadget used in the proof of Lemma~\ref{lem:n=2l=2k=inftyAnycs} showing \NP-hardness when  $n=2$, $\intv=2$, $\natv_j\in \mathbb{N}_{\geq 3}\cup\{\infty\}$ for some $j \in \N$ (represented by the circular nodes) while $j'\in\N\backslash\{j\}$ is such that $\natv_{j'}\in \mathbb{N}_{\geq 0}\cup\{\infty\}$, and for all $i\in\N$ we have $\segv_i, \sv_i \in \mathbb{N}_{\geq 1}\cup\{\infty\}$.}
    \label{fig:l2kinty}
\end{figure*}
Membership of {\sc Perfect $\Gamma$-Cycle Packing} in \NP\ was shown in Lemma~\ref{lem:member}. To show \NP-hardness, we transform an arbitrary instance $\I$ of \textsc{P3DM} to a  constructed 
instance $\J$  of {\sc Perfect $\Gamma$-Cycle Packing} when $n =2$, $\intv\in \mathbb{N}_{\geq 2}\cup\{\infty\}$, $\natv_j\in \mathbb{N}_{\geq 3}\cup\{\infty\}$ for some $j \in \N$, and $\segv=\sv=(1,1)$. 

In our reduction, we create one vertex for every $w \in  X \cup Y \cup Z$. The remaining vertices and all edges are added via gadgets. One gadget is created for each $t_i=(x,y,z) \in T$, as depicted in Figure~\ref{fig:l2kinty}. We add the following vertices $\{a_i^1, \cdots, a_i^6\}$ for the gadget corresponding to $t_i$, giving the following vertices and edges:
\[ \begin{array}{c}
V= X \cup Y\cup Z\cup \{a_i^1, \cdots, a_i^6 \mid  t_i \in T\} \text{  and}\\
    A=\{(x, a_i^1), (a_i^1, x),(y, a_i^3), (a_i^3, y,(z, a_i^5),(a_i^5, z),(a_i^1, a_i^2), (a_i^2, a_i^1),(a_i^3, a_i^4),\\
    (a_i^4, a_i^3),(a_i^5, a_i^6),(a_i^6, a_i^5),(a_i^2, a_i^4),(a_i^4, a_i^6), (a_i^6, a_i^2)\mid t_i \in T\}.
\end{array}
\]
In our reduction, we have $\N=\{1,2\}$, with the vertices belonging to countries~$1$ and $2$ depicted by the circular and triangular vertices in Figure~\ref{fig:l2kinty}, respectively. Thus, $V_1=\{x,y,z,a_i^2,a_i^4, a_i^6\mid  t_i=(x,y,z)\in T\}$ and $V_2 =V\backslash V_1$.
We let $\natv_1\in \mathbb{N}_{\geq 3}\cup\{\infty\}$, $\natv_2\in \mathbb{N}_{\geq 0}\cup\{\infty\}$, $\segv=(1,1)$ and $\sv=(1,1)$ and, thus, our restrictions for the country-specific parameters are as given in the lemma statement.
Hence, six international cycles of length $2$ and one national $V_1$-cycle of length $3$ have been created. 
Moreover, note that as $\segv=(1,1)$ and $\sv=(1,1)$, there are no international $\Gamma$-cycles with more than two vertices.
Hence, the  only $\Gamma$-cycles in the compatibility graph are:
\[\gcyc=\{\langle x, a_i^1\rangle, \langle y, a_i^3\rangle, \langle z, a_i^5\rangle, \langle a_i^1, a_i^2\rangle, \langle a_i^3, a_i^4\rangle,\langle a_i^5, a_i^6\rangle,\langle a_i^2, a_i^4, a_i^6\rangle \mid  t_i \in T \}.\]

Next, we give a proof of correctness and first assume that there is a perfect matching $M$ of instance $\I$ of \textsc{P3DM} and construct a perfect $\Gamma$-cycle packing $\C$ of instance $\J$.
 If a triple $t_i=(x,y,z)$ is chosen in $M$, then the following cycles are added to $\C$: 
\[\langle x, a_i^1\rangle, \langle y, a_i^3\rangle, \langle z, a_i^5\rangle,\langle a_i^2, a_i^4, a_i^6\rangle\in \C.\]  
Thus, all the nodes within the gadget are covered by some $\Gamma$-cycle.  If a triple $t_i=(x,y,z)$ is not chosen in $M$, then the following cycles are added to $\C$: 
\[\langle a_i^1, a_i^2\rangle, \langle a_i^3, a_i^4\rangle, \langle a_i^5, a_i^6\rangle\in \C.\] 
Observe that all of the gadget's vertices have been selected in some cycle, apart from $x$, $y$, and $z$.  As $M$ is a perfect matching, every vertex $w \in X \cup Y\cup Z$ will be covered by some gadget exactly once. Therefore, all vertices are covered by a cycle in $\gcyc$ and instance $\J$ has a perfect $\Gamma$-cycle packing.  

Conversely, suppose there is a perfect $\Gamma$-cycle packing $\C$ of instance $\J$. We construct a perfect matching $M$ in {\sc P3DM} instance $\I$ as follows. First, recall that no $\Gamma$-cycle contains edges from different gadgets. Hence, each cycle in $\C$ belongs to a single gadget, and we argue about the cycles within the gadget $t_i$ that are selected in a perfect $\Gamma$-cycle packing $\C$ via the following two cases:

\noindent
(1) Suppose $\langle x, a_i^1\rangle\in \C$. Then $\langle  a_i^2,  a_i^4,  a_i^6\rangle\in\C$ to ensure that $a_i^2$ is covered. Thus, any perfect $\Gamma$-cycle packing $\C$ with $\langle x, a_i^1\rangle\in \C$ must also contain $\langle y, a_i^3\rangle, \langle z, a_i^5\rangle\in \C$ to ensure that $a_i^3$ and $a_i^5$ are covered. Thus, if $\langle x, a_i^1\rangle\in \C$ then $\langle y, a_i^3\rangle, \langle z, a_i^5\rangle\in \C$ and when $x$ is selected in the gadget corresponding to $t_i$, then so are vertices $y$ and $z$.\\

\noindent
(2) Suppose $\langle x, a_i^1\rangle\notin \C$. Thus, $a_i^1$ must be covered by $\langle a_i^1, a_i^2\rangle\in \C$. In turn, $\langle a_i^3, a_i^4\rangle, \langle a_i^5, a_i^6\rangle\in \C$ must be selected in order to to cover $a_i^4$ and $a_i^6$. Hence, when $\langle x, a_i^1\rangle\notin \C$ and $x$ is not covered by gadget $t_i$, then $y$ and $z$ cannot be covered by $t_i$ without leaving some vertex uncovered that only exists in gadget $t_i$. \\

 Therefore, it must be the case that either $\{\langle x, a_i^1\rangle$, $\langle y, a_i^3\rangle$, $\langle z, a_i^5\rangle$, $\langle a_i^2, a_i^4, a_i^6\rangle\}\subseteq \C$ or  $\{\langle a_i^1, a_i^2\rangle$, $\langle a_i^3, a_i^4\rangle$, $\langle a_i^5, a_i^6\rangle\} \subseteq \C$, for each $t_i \in T$. In the former case, we add $t_i$ to $M$.
 Finally, $M$ is a perfect matching in $\I$ because $\C$ is a perfect $\Gamma$-cycle packing of $\J$, and in view of points~$(1)$ and ~$(2)$ above. 

Therefore, {\sc Perfect $\Gamma$-Cycle Packing} is \NP-complete  when $n=2$, $\intv=2$, $\natv_j\in \mathbb{N}_{\geq 3}\cup\{\infty\}$ for some $j \in \N$ while $j'\in\N\backslash\{j\}$ is such that $\natv_{j'}\in \mathbb{N}_{\geq 0}\cup\{\infty\}$, and for all $i\in\N$ we have $\segv_i, \sv_i \in \mathbb{N}_{\geq 1}\cup\{\infty\}$. 
\end{proof}

\begin{lemma}\label{lem:Gl=3c=3s=1n=3}
    For a set of country-specific parameters $\Gamma= (n, \emph{\intv}, \emph{\natv},\emph{\segv}, \emph{\sv})$ with 
$n \geq 3$,  
$\emph{\intv}=3$, 
$\emph{\natv}_i \in \mathbb{N}_{\geq 0}\cup\{\infty\}$ for all $i\in\N$,
$\emph{\sv}_i, \emph{\segv}_i \in \mathbb{N}_{\geq 1}\cup\{\infty\}$ for all $i\in \N$,
{\sc Perfect $\Gamma$-Cycle Packing} is \NP-complete.
\end{lemma}
\begin{proof}

\begin{figure*}[t]
    \centering
\resizebox{0.85\columnwidth}{!}{
\begin{tikzpicture}[triangle/.style = {regular polygon, regular polygon sides=3 }]

\begin{scope}[every node/.style={circle,thick,draw,minimum size=0.9cm}]
    \node[style={diamond,thick,draw}]  (a1) at (1,2) {$a_i^1$};
    \node (x) at (2.5,0) {$x$};
    \node[regular polygon, thick, regular polygon sides=3, draw,
     inner sep=0.005cm]  (a2) at (4, 2) {$a_i^2$};

    \node (a4) at (7.5,2) {$a_i^4$};
    \node[regular polygon, thick, regular polygon sides=3, draw,
     inner sep=0.005cm] (y) at (9.5,0) {$y$};
    \node[style={diamond,thick,draw}]   (a5) at (11.5,2) {$a_i^5$};

    \node (a7) at (14.5,2) {$a_i^7$};
    \node[style={diamond,thick,draw}] (z) at (16,0) {$z$};
    \node[regular polygon, thick, regular polygon sides=3, draw,
     inner sep=0.005cm]  (a8) at (17.5,2) {$a_i^8$};

    \node[regular polygon, thick, regular polygon sides=3, draw,
     inner sep=0.005cm] (a6) at (9.5,6) {$a_i^6$};

    \node (a3) at (2.5,7) {$a_i^3$};
    \node[style={diamond,thick,draw}] (a9) at (16,7) {$a_i^9$};
   
\end{scope}

 \begin{scope}[>={Stealth[black]},
               every edge/.style={draw,very thick}]
   
    \path [->] (a2) edge node {} (x);
    \path [->] (x) edge node {} (a1);
    \path [->] (a1) edge node {} (a2);
    \path [->] (a3) edge[] node {} (a1);
    \path [->] (a2) edge[] node {} (a3);

    \path [->] (z) edge node {} (a7);
    \path [->] (a7) edge node {} (a8);
    \path [->] (a8) edge node {} (z);
    \path [->] (a9) edge node {} (a7);
    \path [->] (a8) edge node {} (a9);

%center bottom

    \path [->] (a4) edge node {} (a5);
    \path [->] (a5) edge node {} (y);
    \path [->] (y) edge node {} (a4);

%center middle
    \path [->] (a5) edge node {} (a6);

    \path [->] (a6) edge node {} (a4);

%top cycle
    \path [->] (a3) edge node {} (a6);
    \path [->] (a6) edge node {} (a9);

    \path [->] (a9) edge node {} (a3);
\end{scope}
\end{tikzpicture}}
\Description{}
    \caption{Gadget used in the proof of Lemma~\ref{lem:Gl=3c=3s=1n=3} when $n=3$,  $\intv=3$, and for all $i\in \N$ we have that  $\segv_i, \sv_i\in \mathbb{N}_{\geq 1}\cup\{\infty\}$ and $\natv_i \in\mathbb{N}_{\geq 0}\cup\{\infty\}$.
    }
    \label{fig:Gl=3c=3s=1n=3}
\end{figure*}
Membership of {\sc Perfect $\Gamma$-Cycle Packing} in \NP\ was shown in Lemma~\ref{lem:member}. To show \NP-hardness, we  reduce from an instance $\I$ of \textsc{P3DM}~\cite{karp1972} to a constructed instance $\J$  of {\sc Perfect $\Gamma$-Cycle Packing}. In particular, we show hardness when $\N=\{1,2,3\}$, $\intv\in \mathbb{N}_{\geq 3}\cup\{\infty\}$, $\natv=(0,0,0)$, $\segv,\sv=(1,1,1)$. Our following proof in fact extends to the different $\Gamma$ listed in the Lemma statement. 

In our reduction, we create one vertex for every $w \in  X \cup Y \cup Z$. The remaining vertices and all edges are added via gadgets. One gadget is created for each $t_i=(x,y,z) \in T$, as depicted in Figure~\ref{fig:Gl=3c=3s=1n=3}. We add the following vertices  $\{a_i^1, \dots, a_i^9\}$ for the gadget corresponding to $t_i$. We have the following vertices and edges: 
\[
\begin{array}{c}
V = X \cup Y \cup Z \cup \{a_i^1, \dots, a_i^9 \mid  t_i \in T\} \quad \text{   and}\\
A=\{(x, a_i^1), (a_i^1, a_i^2), (a_i^2,x), (y, a_i^4), (a_i^4,a_i^5), (a_i^5, y), (z, a_i^7), (a_i^7,a_i^8),(a_i^7, z),(a_i^3, a_i^1), \\
  (a_i^2, a_i^3),(a_i^6, a_i^4), (a_i^5, a_i^6),(a_i^9, a_i^7),(a_i^8, a_i^9),(a_i^3, a_i^6),(a_i^6, a_i^9), (a_i^9, a_i^3)\mid  t_i \in T\}. 
\end{array}
\]
The vertices are distributed among the three countries as follows: $\mathcal{V} = (V_1, V_2, V_3)$, where  
\[
V_1 = \{x, a_i^3, a_i^4, a_i^7 \mid t_i \in T\}, \quad
V_2 = \{y, a_i^2, a_i^6, a_i^8 \mid t_i \in T\}, \quad
V_3 = \{z, a_i^1, a_i^5, a_i^9 \mid t_i \in T\}.
\]

These sets are represented in Figure~\ref{fig:Gl=3c=3s=1n=3} using circular, triangular, and diamond-shaped vertices, respectively. 
All international $\Gamma$-cycles can have at most $3$ vertices when $\intv\in\mathbb{N}_{\geq 4}\cup\{ \infty\}$. Hence, the  only $\Gamma$-cycles in the compatibility graph are:
\[
\X(I) = \{\langle x, a_i^1, a_i^2\rangle, \langle y, a_i^4, a_i^5\rangle, \langle z, a_i^7, a_i^8\rangle, \langle a_i^1, a_i^2, a_i^3\rangle, \langle a_i^4, a_i^5, a_i^6\rangle, \langle a_i^7, a_i^8, a_i^9\rangle, \langle a_i^3, a_i^6, a_i^9\rangle \mid t_i \in T \}.
\]
Any other cycle must exit and enter a gadget via the vertices corresponding to the elements of $t_i = (x, y, z)$. However, this would require a cycle with at least two segments from the same country.  
Assume that we exit via vertex $x$. In this case, the cycle must contain the path $[a_i^3, a_i^1, a_i^2, x]$, which requires two $V_1$-segments, namely $a_i^3$ and $x$. This would not be part of a $\Gamma$-cycle. Similar arguments apply for exiting via $y$ and $z$.

We next show that the instance $\I$ of \textsc{P3DM} has a perfect matching $M$ exactly when the constructed {\sc Perfect $\Gamma$-Cycle Packing} instance $\J$ has a perfect $\Gamma$-cycle packing. 
First, assume that $M$ is a perfect matching of \textsc{P3DM} instance $\I$. We construct a $\Gamma$-cycle packing in $\J$ given a perfect matching $M$. If a triple $t_i=(x,y,z)$ is chosen in $M$, then we add the $\Gamma$-cycles to $\C$: 
$\langle x, a_i^1, a_i^2\rangle,   \langle y, a_i^4, a_i^5\rangle,   \langle z, a_i^7, a_i^8\rangle,$  and   $\langle a_i^3, a_i^6, a_i^9\rangle$. 
 Thus, all the nodes within the gadget are covered by some $\Gamma$-cycle in this case. Alternatively, if a triple $t_i=(x,y,z)$ is not chosen in $M$, then the following cycles are added to $\C$: 
$\langle a_i^1, a_i^2, a_i^3\rangle,   \langle a_i^4, a_i^5, a_i^6\rangle,   \text{and}   \langle a_i^7, a_i^8, a_i^9\rangle$. 
In this case, all vertices are selected except for $x$, $y$, and $z$. By assumption, $x$, $y$, and $z$ will be selected in a different gadget, given that $M$ is a perfect matching.  
Thus, when there is a perfect matching $M$, there exists a perfect $\Gamma$-cycle packing $\C$.

Conversely, suppose there is a perfect $\Gamma$-cycle packing $\C$ of instance $\J$. We construct a perfect matching $M$ in {\sc P3DM} instance $\I$ as follows. First, recall that no $\Gamma$-cycle contains edges from different gadgets. Hence, each cycle in $\C$ belongs to a single gadget, and we argue about the cycles of gadget $t_i$ that are selected in perfect $\Gamma$-cycle packing $\C$ via the following two cases:
\begin{enumerate}
    \item Suppose $\langle x, a_i^1, a_i^2
    \rangle\in \C$. Then, $a_i^3$ can only be selected by the cycle $\langle a_i^3, a_i^6, a_i^9\rangle$, and it is in $\C$. Thus, $\langle y, a_i^4, a_i^5\rangle$, $\langle z, a_i^7, a_i^8\rangle\in \C$, in order for $a_i^4, a_i^5,a_i^7$, and $a_i^8$ to be selected. Hence, we see that if in a perfect $\Gamma$-cycle packing $\langle x, a_i^1, a_i^2
    \rangle\in \C$, then vertices $y$ and $z$ must also be selected in the same gadget.
    \item Suppose $\langle x, a_i^1, a_i^2\rangle\notin \C$. Then $\langle a_i^1, a_i^2, a_i^3\rangle\in \C$ to ensure that $ a_i^1$ and $a_i^2 $ are covered by some $\Gamma$-cycle. In turn, $\langle a_i^4, a_i^5, a_i^6\rangle, \langle a_i^7, a_i^8, a_i^9\rangle\in \C$ to ensure that $a_i^6$ an $a_i^7$ are covered by a $\Gamma$-cycle. Therefore, when $\langle x, a_i^1, a_i^2\rangle\notin \C$, the vertices $x,y$, and $z$ cannot be selected by a cycle within the gadget corresponding to $t_i$.  
\end{enumerate}

Therefore, it must be the case that either $$\{\langle x, a_i^1, a_i^2\rangle, \langle y, a_i^4, a_i^5\rangle, \langle z, a_i^7, a_i^8\rangle, \langle a_i^3, a_i^6, a_i^9\rangle\} \subseteq \C \text{  or  }\{ \langle a_i^1, a_i^2, a_i^3\rangle, \langle a_i^4, a_i^5, a_i^6\rangle, \langle a_i^7, a_i^8, a_i^9\rangle\} \subseteq \C$$, for each $t_i \in T$. In the former case, we add $t_i$ to $M$. 
Finally, $M$ is a perfect matching in $\I$ because $\C$ is a perfect $\Gamma$-cycle packing in $\J$, and in view of points~$(1)$ and ~$(2)$ above.

Hence, we have shown that \textsc{Perfect $\Gamma$-Cycle Packing} is \NP-hard when $\intv \in \mathbb{N}_{\geq 3} \cup \{\infty\}$, $\segv = \{1\}^n$, $\sv = \{1\}^n$, and $n=3$. Therefore, \textsc{Perfect $\Gamma$-Cycle Packing} is \NP-complete under the conditions on $\Gamma$ stated in the lemma statement.
\end{proof}

\begin{lemma}\label{lem:n=2l=2k=inftyAnycsSEGV}
    For a set of country-specific parameters $\Gamma= (n, \emph{\intv}, \emph{\natv},\emph{\segv}, \emph{\sv})$ with 
$n \geq 2$,  
$\emph{\intv}\in \mathbb{N}_{\geq 3}\cup\{\infty\}$, 
$\emph{\natv}_j \in \mathbb{N}_{\geq 3}\cup\{\infty\}$ for some $j\in\N$,
and $\emph{\natv}_i \in \mathbb{N}_{\geq 0}\cup\{\infty\}$ for all $i\in\N\backslash \{j\}$,
$\emph{\sv}_i \in \mathbb{N}_{\geq 1}\cup\{\infty\}$ for all $i\in \N$,
$\emph{\segv}_j=1$ for some $j\in \N$, and
$\emph{\segv}_i\in \mathbb{N}_{\geq 1}\cup\{\infty\}$ for all $i\in \N\setminus \{j\}$,
{\sc Perfect $\Gamma$-Cycle Packing} is \NP-complete.
\end{lemma}
\begin{proof}
    The proof follows from Lemma~\ref{lem:n=2l=2k=inftyAnycs}, observe that the same reduction induces the same set $\gcyc$ given these parameters. 
\end{proof}

\begin{lemma}\label{lem:n=2l=2k=2Anycs}
    For a set of country-specific parameters $\Gamma= (n, \emph{\intv}, \emph{\natv},\emph{\segv}, \emph{\sv})$ with 
$n \geq 2$,  
$\emph{\intv}=3$, 
$\emph{\natv}_i \in \mathbb{N}_{\geq 0}\cup\{\infty\}$ for all $i\in\N$,
$\emph{\sv}_i \in  \mathbb{N}_{\geq 1}\cup\{\infty\}$ for all $i\in \N$,
$\emph{\segv}_j\in \mathbb{N}_{\geq 2}$ for some $j\in \N$, and
$\emph{\segv}_i\in \mathbb{N}_{\geq 1}\cup\{\infty\}$ for all $i\in \N\setminus \{j\}$,
{\sc Perfect $\Gamma$-Cycle Packing} is \NP-complete.
\end{lemma}
\begin{proof}
Membership of {\sc Perfect $\Gamma$-Cycle Packing} in \NP\ was shown in Lemma~\ref{lem:member}. To show \NP-hardness, we reduce from an instance $\I$ of \textsc{P3DM}~\cite{karp1972} to a constructed instance $\J$  of {\sc Perfect $\Gamma$-Cycle Packing}. In particular, we show hardness when 
$n=2$, $\intv=3$,  there is some $j \in\N$ such that $\segv_j\in \mathbb{N}_{\geq 2}$ while the remaining country $j' \in \N\backslash\{j\}$ has $\segv_{j'}\in \mathbb{N}_{\geq 1}\cup\{\infty\}$ and for all $i \in \N$  have  $\sv_i\in \mathbb{N}_{\geq 1}\cup\{\infty\}$ and $\natv_i \in \mathbb{N}_{\geq 0}\cup\{\infty\}$. 
\begin{figure*}[t]
    \centering
\resizebox{0.34\columnwidth}{!}{
\begin{tikzpicture}[triangle/.style = {regular polygon, regular polygon sides=3 }]

\begin{scope}[every node/.style={circle,thick,draw,minimum size=0.9cm}]
    \node[regular polygon, thick, regular polygon sides=3, draw,
     inner sep=0.005cm] (a1) at (2.5,3) {$a_i^1$};
    \node (x) at (2.5,0) {$x$};
    \node (a3) at (2.5,7) {$a_i^2$};

    \node (a4) at (5,3) {$a_i^3$};
    \node[regular polygon, thick, regular polygon sides=3, draw,
     inner sep=0.005cm]  (y) at (5,0) {$y$};
    \node[regular polygon, thick, regular polygon sides=3, draw,
     inner sep=0.005cm]  (a6) at (5,6) {$a_i^4$};

    \node(a7) at (7.5,3) {$a_i^5$};
    \node[regular polygon, thick, regular polygon sides=3, draw,
     inner sep=0.005cm]  (z) at (7.5,0) {$z$};
    
    \node[regular polygon, thick, regular polygon sides=3, draw,
     inner sep=0.005cm] (a9) at (7.5,7) {$a_i^6$};
\end{scope}

 \begin{scope}[>={Stealth[black]},
               every edge/.style={draw,very thick}]
 
    \path [->] (x) edge[bend left] node {} (a1);
    \path [->] (a1) edge[bend left] node {} (x);
    \path [->] (a1) edge[bend left] node {} (a3);
    \path [->] (a3) edge[bend left] node {} (a1);

    \path [->] (y) edge[bend left] node {} (a4);
    \path [->] (a4) edge[bend left] node {} (y);
    \path [->] (a4) edge[bend left] node {} (a6);
    \path [->] (a6) edge[bend left] node {} (a4);

    \path [->] (z) edge[bend left] node {} (a7);
    \path [->] (a7) edge[bend left] node {} (z);
    \path [->] (a7) edge[bend left] node {} (a9);
    \path [->] (a9) edge[bend left] node {} (a7);

    \path [->] (a3) edge[] node {} (a6);
    \path [->] (a6) edge[] node {} (a9);
    \path [->] (a9) edge[] node {} (a3);

\end{scope}
\end{tikzpicture}}
\Description{}
    \caption{Gadget used in the proof of Lemma~\ref{lem:n=2l=2k=2Anycs} showing \NP-hardness when $n=2$, $\intv=3$,  there is some $j \in\N$ such that $\segv_j\in \mathbb{N}_{\geq 2}$ (represented by the triangular nodes in the figure) while the remaining country $j' \in \N\backslash\{j\}$ has $\segv_{j'}\in \mathbb{N}_{\geq 1}\cup\{\infty\}$ and for all $i \in \N$  have  $\sv_i\in \mathbb{N}_{\geq 1}\cup\{\infty\}$ and $\natv_i \in \mathbb{N}_{\geq 0}\cup\{\infty\}$. 
    }
    \label{fig:lintyk0c12s11n2}
\end{figure*}

In our reduction, we create one vertex for every $w \in  X \cup Y \cup Z$. The remaining vertices and all edges are added via gadgets. One gadget is created for each $t_i=(x,y,z) \in T$, as depicted in Figure~\ref{fig:lintyk0c12s11n2}. Hence, we have the following vertices and edges:
\[
\begin{array}{c}
     V = X \cup Y \cup Z \cup \{a_i^1, \ldots, a_i^6 \mid t_i \in T\} \quad \text{  and}\\
      A = \{(x, a_i^1), (a_i^1, x), (y, a_i^3), (a_i^3, y), (z, a_i^5), (a_i^5, z), (a_i^1, a_i^2), (a_i^2, a_i^1),(a_i^3, a_i^4), \\
    (a_i^4, a_i^3), (a_i^5, a_i^6), (a_i^6, a_i^5), (a_i^2, a_i^4), (a_i^4, a_i^6), (a_i^6, a_i^2) \mid t_i = (x,y,z) \in T\}.
\end{array}
\]

In our reduction, we define the countries as $\N = \{1,2\}$ where
\[V_1=\{x, a_i^2, a_i^3, a_i^5 \mid t_i = (x,y,z) \in T\} \quad \text{ and } \quad V_2= V\backslash V_1.\] 
The vertices from countries $1$ and $2$ are represented by circular and triangular vertices in Figure~\ref{fig:lintyk0c12s11n2}, respectively. We set $\segv_2 = 2$, and let $\segv_1\in\mathbb{N}_{\geq 1}\cup\{\infty\}$, and $\sv_i\in \mathbb{N}_{\geq 1}\cup\{\infty\}$ for all $i\in \N$. 
Observe that the only $\Gamma$-cycles in the $\pG$ are those present within each gadget. Specifically, we define:
\[
\gcyc = \{\langle a_i^2, a_i^4, a_i^6\rangle, \langle a_i^1, a_i^2\rangle, \langle a_i^3, a_i^4\rangle, \langle a_i^5, a_i^6\rangle, \langle a_i^1, x\rangle, \langle a_i^3, y\rangle, \langle a_i^5, z\rangle \mid t_i = (x,y,z) \in T\}.
\]
Moreover, any cycle spanning multiple gadgets would violate $\intv=3$. Specifically, entering a gadget via one vertex $w \in t_i$ and exiting via another $w' \in t_i$ with $w \neq w'$ would require at least three segments from both country $1$ and country $2$, which contradicts the given constraints.

The proof of correctness follows the same reasoning as in Lemma~\ref{lem:n=2l=2k=inftyAnycs}, given that the reduction contains the same method of the creation of gadgets and the same set of $\Gamma$-cycles. Hence, we have shown that {\sc Perfect $\Gamma$-Cycle Packing} is \NP-complete when $n=2$, $\intv=3$,  there is some $j \in\N$ such that $\segv_j\in \mathbb{N}_{\geq 2}$ while the remaining countries $j' \in \N\backslash\{j\}$ have $\segv_{j'}\in \mathbb{N}_{\geq 1}\cup\{\infty\}$ and for all $i \in \N$  have  $\sv_i\in \mathbb{N}_{\geq 1}\cup\{\infty\}$ and $\natv_i \in \mathbb{N}_{\geq 0}\cup\{\infty\}$. 
\end{proof}

\begin{lemma}\label{lem:Gl=3c=3s=1n=4}
    For a set of country-specific parameters $\Gamma= (n, \emph{\intv}, \emph{\natv},\emph{\segv}, \emph{\sv})$ with 
$n \geq 3$,  
$\emph{\intv}\in \mathbb{N}_{\geq 4}\cup\{\infty\}$, 
$\emph{\natv}_i \in  \mathbb{N}_{\geq 0}\cup\{\infty\}$ for all $i\in\N$,
$\emph{\sv}_j =1$ for some $j\in \N$, and $\emph{\sv}_i \in  \mathbb{N}_{\geq 1}\cup\{\infty\}$ for all $i\in \N\backslash\{i\}$,
$\emph{\segv}_i\in  \mathbb{N}_{\geq 1}\cup\{\infty\}$ for all $i\in \N$,
{\sc Perfect $\Gamma$-Cycle Packing} is \NP-complete.
\end{lemma}
\begin{proof}
    The proof follows from Lemma~\ref{lem:Gl=3c=3s=1n=3}, observe that the same reduction induces the same set $\gcyc$ given these parameters. 
\end{proof}

\begin{lemma}\label{lem:n=2l=2k=2Anycs4}
    {\sc Perfect $\Gamma$-Cycle Packing} is \NP-complete when $n=2$, $\intv\in \{\mathbb N_{\geq 4}, \infty\}$,  there is some $j \in\N$ such that $\segv_j\in \mathbb{N}_{\geq 2}$ while the remaining countries $j' \in \N\backslash\{j\}$ have $\segv_{j'}\in \mathbb{N}_{\geq 1}\cup\{\infty\}$, for some $k\in\N$ has  $\sv_k=1$ while the remaining $k'\in \N\backslash\{k\}$ have $\sv_{k'}\in \mathbb{N}_{\geq 1}\cup\{\infty\}$,  and for all $i \in \N$ we have $\natv_i \in \mathbb{N}_{\geq 0}\cup\{\infty\}$.
\end{lemma}
\begin{proof}
    The proof follows from Lemma~\ref{lem:n=2l=2k=2Anycs}, observe that the same reduction induces the same set $\gcyc$ given these parameters. 
\end{proof}

\begin{lemma}\label{lem:n=2l=2k=inftyAnycsSV}
    For a set of country-specific parameters $\Gamma= (n, \emph{\intv}, \emph{\natv},\emph{\segv}, \emph{\sv})$ with 
    $n \geq 2$,  
    $\emph{\intv}\in \mathbb{N}_{\geq 3}\cup\{\infty\}$, 
    $\emph{\natv}_j \in  \mathbb{N}_{\geq 3}\cup\{\infty\}$ for some $j\in\N$, and $\emph{\natv}_i \in \mathbb{N}_{\geq 0}\cup\{\infty\}$ for all $i\in\N\backslash\{j\}$,
    $\emph{\sv}_j=1$ for all $j\in \N$, and
    $\emph{\sv}_i \in\mathbb{N}_{\geq 1}\cup\{\infty\}$ for all $i\in \N\backslash\{j\}$,
    $\emph{\segv}_i\in\mathbb{N}_{\geq 1}\cup\{\infty\}$ for all $i\in \N$,
    {\sc Perfect $\Gamma$-Cycle Packing} is \NP-complete.
\end{lemma}
\begin{proof}
    The proof follows from Lemma~\ref{lem:n=2l=2k=inftyAnycs}, observe that the same reduction induces the same set $\gcyc$ given these parameters. 
\end{proof}

\begin{lemma}\label{lem:l=4k>=2c=1s>=2n=2}
    For a set of country-specific parameters $\Gamma= (n, \emph{\intv}, \emph{\natv},\emph{\segv}, \emph{\sv})$ with 
    $n \geq 2$,  
    $\emph{\intv}=\infty$, 
    $\emph{\natv}_i \in \mathbb{N}_{\geq 0}\cup\{\infty\}$ for all $i\in\N$,
    $\emph{\sv}_j \in \mathbb{N}_{\geq 2}$ for all $j\in \N$,
    and $\emph{\sv}_i \in \mathbb{N}_{\geq 2}\cup\{\infty\}$ for all $i\in \N\backslash\{j\}$,
    $\emph{\segv}_i\in  \mathbb{N}_{\geq 1}\cup\{\infty\}$ for all $i\in \N$,
    {\sc Perfect $\Gamma$-Cycle Packing} is \NP-complete.
\end{lemma}
\begin{proof}
    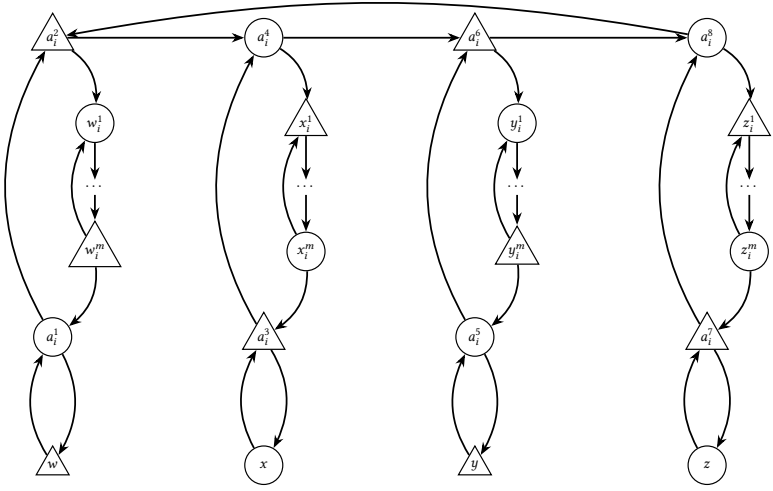
\begin{figure*}[t]
    \centering
\resizebox{0.75\columnwidth}{!}{
\begin{tikzpicture}[triangle/.style = {regular polygon, regular polygon sides=3 }]

\begin{scope}[every node/.style={circle,thick,draw,minimum size=0.9cm}]
    \node (a1) at (-5,3) {$a_i^1$};
    \node[regular polygon, thick, regular polygon sides=3, draw,
     inner sep=0.005cm] (w) at (-5,0) {$w$};
    \node[regular polygon, thick, regular polygon sides=3, draw,
     inner sep=0.005cm] (a2) at (-5,10) {$a_i^2$};

    \node[regular polygon, thick, regular polygon sides=3, draw,
     inner sep=0.005cm] (a3) at (0,3) {$a_i^3$};
    \node (x) at (0,0) {$x$};
    \node (a4) at (0,10) {$a_i^4$};

    \node (a5) at (5,3) {$a_i^5$};
    \node[regular polygon, thick, regular polygon sides=3, draw,
     inner sep=0.005cm] (y) at (5,0) {$y$};
    \node[regular polygon, thick, regular polygon sides=3, draw,
     inner sep=0.005cm] (a6) at (5,10) {$a_i^6$};

    \node[regular polygon, thick, regular polygon sides=3, draw,
     inner sep=0.005cm] (a7) at (10.5,3) {$a_i^7$};
    \node (z) at (10.5,0) {$z$};
    
    \node (a8) at (10.5,10) {$a_i^8$};

    \node[regular polygon, thick, regular polygon sides=3, draw,
     inner sep=0.005cm] (z1) at (11.5,8) {$z_i^1$};
    \node (zs) at (11.5,5) {$z_i^{m}$};

    \node (y1) at (6,8) {$y_i^1$};
    \node[regular polygon, thick, regular polygon sides=3, draw,
     inner sep=0.005cm] (ys) at (6,5) {$y_i^{m}$};
    \node[regular polygon, thick, regular polygon sides=3, draw,
     inner sep=0.005cm] (x1) at (1,8) {$x_i^1$};
    \node (xs) at (1,5) {$x_i^{m}$};
    \node (w1) at (-4,8) {$w_i^1$};
    \node[regular polygon, thick, regular polygon sides=3, draw,
     inner sep=0.005cm] (ws) at (-4,5) {$w_i^{m}$};
   
\end{scope}
 \node (zdot) at (11.5,6.5) {$\cdots$};
  \node (ydot) at (6,6.5) {$\cdots$};
   \node (xdot) at (1,6.5) {$\cdots$};
    \node (wdot) at (-4,6.5) {$\cdots$};

 \begin{scope}[>={Stealth[black]},
               every edge/.style={draw,very thick}]

    \path [->] (w) edge[bend left] node {} (a1);
    \path [->] (a1) edge[bend left] node {} (w);
    \path [->] (a1) edge[bend left] node {} (a2);
    \path [->] (a2) edge[bend left] node {} (w1);
    \path [->] (w1) edge node {} (wdot);
    \path [->] (wdot) edge node {} (ws);
    \path [->] (ws) edge[bend left] node {} (a1);

    \path [->] (x) edge[bend left] node {} (a3);
    \path [->] (a3) edge[bend left] node {} (x);
    \path [->] (a3) edge[bend left] node {} (a4);
    \path [->] (a4) edge[bend left] node {} (x1);

    \path [->] (x1) edge node {} (xdot);
    \path [->] (xdot) edge node {} (xs);
    \path [->] (xs) edge[bend left] node {} (a3);

    \path [->] (y) edge[bend left] node {} (a5);
    \path [->] (a5) edge[bend left] node {} (y);
    \path [->] (a5) edge[bend left] node {} (a6);
    \path [->] (a6) edge[bend left] node {} (y1);

    \path [->] (y1) edge node {} (ydot);
    \path [->] (ydot) edge node {} (ys);
    \path [->] (ys) edge[bend left] node {} (a5);

    \path [->] (z) edge[bend left] node {} (a7);
    \path [->] (a7) edge[bend left] node {} (z);
    \path [->] (a7) edge[bend left] node {} (a8);
    \path [->] (a8) edge[bend left] node {} (z1);

    \path [->] (z1) edge node {} (zdot);
    \path [->] (zdot) edge node {} (zs);
    \path [->] (zs) edge[bend left] node {} (a7);

    \path [->] (a2) edge[] node {} (a4);
    \path [->] (a4) edge[] node {} (a6);
    \path [->] (a6) edge[] node {} (a8);
    \path [->] (a8) edge[bend right=10] node {} (a2);

    \path [->] (ws) edge[bend left] node {} (w1);
    \path [->] (xs) edge[bend left] node {} (x1);
    \path [->] (ys) edge[bend left] node {} (y1);
    \path [->] (zs) edge[bend left] node {} (z1);
\end{scope}
\end{tikzpicture}}
\Description{}
    \caption{Gadget used in the proof of Lemmas~\ref{lem:l=4k>=2c=1s>=2n=2} 
    and~\ref{lem:l=4k>=2c=1s=2n=2}. 
    Note that we let $m =2\min_{i\in\N}(\sv_i)-2$ in Lemma~\ref{lem:l=4k>=2c=1s>=2n=2} and $m = \min(2\min_{i\in\N}(\sv_i), \intv)-2$ in Lemma~\ref{lem:l=4k>=2c=1s=2n=2}.}
    \label{fig:l=4k>=2c=1s>=2n=2}
\end{figure*}

Membership of {\sc Perfect $\Gamma$-Cycle Packing} in \NP\ was shown in Lemma~\ref{lem:member}. To show \NP-hardness, we transform an instance $\I$ of {\sc P4DM} to a constructed instance $\J$ of {\sc Perfect $\Gamma$-Cycle Packing}. We provide a reduction for the case when $\N=\{1,2\}$, $\intv=\infty$, $\segv=(1,1)$, $\natv=(0,0)$, and 
$\sv_1\in  \mathbb{N}_{\geq 2}\cup\{\infty\}$ and $\sv_2=\sigma\geq 2$ for some fixed constant $\sigma\in\mathbb{N}$.
We assume without loss of generality that $\sv_2\leq \sv_1$ (else we reverse the roles of countries $1$ and $2$).
Note that given these country-specific parameters, the longest possible $\Gamma$-cycle will have at most~$2\sigma$ vertices (both countries having $\sigma$ segments of size~$1$). 

In our reduction, we first create a vertex for each $u \in W \cup X \cup Y \cup Z$. We create a gadget for every $t_i = (w, x, y, z) \in T$, adding the remaining vertices and edges as depicted in Figure~\ref{fig:l=4k>=2c=1s>=2n=2}. 
Hence, the vertex and edge sets are:
\[
\begin{array}{c}
V = W \cup X \cup Y \cup Z \cup \{a_i^1, \dots, a_i^8, w_i^1, \dots, w_i^m, x_i^1, \dots, x_i^m, y_i^1, \dots, y_i^m, z_i^1, \dots, z_i^m \mid  t_i \in T\}, \text{  and}\\
A = \{(w, a_i^1), (a_i^1, w), (a_i^1, a_i^2), (a_i^2, w_i^1), (w_i^{q}, w_i^{q+1}), (w_i^m, a_i^1),
(x, a_i^3), (a_i^3, x), \\ (a_i^3, a_i^4), (a_i^4, x_i^1),  (x_i^m, a_i^3), (x_i^{q}, x_i^{q+1}),
(y, a_i^5), (a_i^5, y), (a_i^5, a_i^6), (a_i^6, y_i^1), (y_i^{q}, y_i^{q+1}), \\(y_i^m, a_i^5),
(z, a_i^7), (a_i^7, z), (a_i^7, a_i^8),(a_i^8, z_i^1),(z_i^{q}, z_i^{q+1}),(z_i^m, a_i^7),
(a_i^2, a_i^4), (a_i^4, a_i^6), (a_i^6, a_i^8),\\ (a_i^8, a_i^2),
(w_i^m, w_i^1), (x_i^m, x_i^1), (y_i^m, y_i^1), (z_i^m, z_i^1)
 \mid q \in [1, m-1], t_i \in T\}.
\end{array}
\]
We divide the vertices into $V_1$ and $V_2$ such that:
\[
V_1 = \{a_i^1, w_i^q, a_i^4, x_i^{q+1}, x, y_i^q, a_i^5, a_i^8, z_i^{q+1}, z \mid  t_i \in T, \text{ and odd } q \in [1, 2\sigma - 1]\} \text{ and } V_2=V\backslash V_1,
\]
where vertices in $V_1$ and $V_2$ are represented by circular and triangular nodes in Figure~\ref{fig:l=4k>=2c=1s>=2n=2}, respectively.
Note that the vertices are distributed across the countries in an alternating manner, ensuring that there are no national arcs. Finally, we set $m=2\min_{i\in\N}(\sv_i)-2=2\sigma-2$.

Next, we consider which cycles in the compatibility graph are $\Gamma$-cycles. We observe that every $\Gamma$-cycle is contained within a single gadget. 
Any cycle $C$ that spans at least two gadgets must enter and exit one of the gadgets $t_i$ through the nodes $w, x, y$, and $z$. Without loss of generality, assume that the cycle $C$ contains a path from $w$ to $x$. Hence, $C$ includes the path:
$[w, a_i^1, a_i^2, a_i^4, x_i^1, \dots, x_i^m, a_i^3, x]$
where, given the assumption that $m = 2\sigma - 2$, there are $2\sigma$ $V_1$-segments and $2\sigma$ $V_2$-segments in this path. Therefore, it cannot be part of any $\Gamma$-cycle.
Thus, the set of $\Gamma$-cycles is:
\[
\begin{array}{c}
\gcyc = \{\langle w, a_i^1 \rangle, 
\langle x, a_i^3 \rangle, 
\langle y, a_i^5 \rangle, 
\langle z, a_i^7 \rangle, 
\langle a_i^2, a_i^4, a_i^6, a_i^8 \rangle, \\
\langle w_i^1, \dots, w_i^m, a_i^1, a_i^2 \rangle, 
\langle w_i^1, \dots, w_i^m \rangle, 
\langle x_i^1, \dots, x_i^m, a_i^3, a_i^4 \rangle, 
\langle x_i^1, \dots, x_i^m \rangle, \\
\langle y_i^1, \dots, y_i^m, a_i^5, a_i^6 \rangle, 
\langle y_i^1, \dots, y_i^m \rangle, 
\langle z_i^1, \dots, z_i^m, a_i^7, a_i^8 \rangle, 
\langle z_i^1, \dots, z_i^m \rangle \mid \ t_i \in T
\}.
\end{array}
\]

We now provide a proof of correctness by first assuming that we have an instance of {\sc P4DM} for which there is a perfect matching $ M $. We can then build a perfect $\Gamma$-cycle packing depending on whether some $ t_i \in T $ is included in $ M $. If $ t_i \in M $, we select the following cycles and add them to $\C$: 
\[\begin{array}{c}
\langle w, a_i^1 \rangle,   \langle x, a_i^3 \rangle,   \langle y, a_i^5 \rangle,   \langle z, a_i^7 \rangle,   \langle a_i^2, a_i^4, a_i^6, a_i^8 \rangle,  \\ \langle w_i^1, \cdots, w_i^m \rangle,  
\langle x_i^1, \cdots, x_i^m \rangle,   \langle y_i^1, \cdots, y_i^m \rangle,   \langle z_i^1, \cdots, z_i^m \rangle \in C.
\end{array}\]
By selecting these cycles, every vertex within the gadget corresponding to $ t_i $ has been selected in some $\Gamma$-cycle. Next, when $ t_i \notin M $, we select the following cycles:
$$\langle w_i^1, \cdots, w_i^m, a_i^1, a_i^2 \rangle,   \langle x_i^1, \cdots, x_i^m, a_i^3, a_i^4 \rangle,  \langle y_i^1, \cdots, y_i^m, a_i^5, a_i^6 \rangle,\langle z_i^1, \cdots, z_i^m, a_i^7, a_i^8 \rangle\in \C.$$
Here, all vertices selected in the gadget, except for $ w, x, y, $ and $ z $, have been selected.
Therefore, given that $ M $ is a perfect matching, every vertex in our reduction will be covered by the cycle packing described.

Conversely, suppose there is a perfect $\Gamma$-cycle packing $\C$ of instance $\J$. We construct a perfect matching $M$ in {\sc P4DM} instance $\I$ as follows. First, recall that no $\Gamma$-cycle contains edges from different gadgets. 
We argue about the cycles of gadget $t_i$ that are selected in perfect $\Gamma$-cycle packing $\C$ via the following two cases:
\begin{enumerate}
    \item Suppose $\langle w, a_i^1\rangle\in \C$. Thus, the only way vertices $w_i^1, \ldots, w_i^m$ can be selected is via $\C$ containing the cycle $\langle w_i^1, \ldots, w_i^m\rangle$. Then, $a_i^2$ must be covered by $\langle a_i^2, a_i^4,a_i^6, a_i^8 \rangle\in \C$. Therefore,  to ensure that these vertices are covered, we have that $\langle x_i^1, \ldots, x_i^m\rangle, \langle y_i^1, \ldots, y_i^m\rangle, \langle z_i^1, \ldots, z_i^m\rangle \in\C$. Finally, the remaining vertices labelled with $a$ must be selected as such: $\langle x, a_i^3\rangle, \langle y, a_i^5\rangle, \langle z, a_i^7\rangle\in \C$. Thus, whenever a perfect $\Gamma$-cycle packing selects $\langle w, a_i^1\rangle$, it must be the case that $w,x, y,$ and $z$ are all selected by the same gadget corresponding to $t_i$. 
    
    \item Suppose $\langle w, a_i^1\rangle\notin \C$. For $a_i^1$ to be covered in $\C$, we see that the cycle $\langle a_i^1, a_i^2, w_i^1, \ldots, w_i^m\rangle\in \C$. Thus, for vertices $a_i^4, a_i^6$, and $a_i^8$ to be covered, then $\langle a_i^3, a_i^4, x_i^1, \ldots, x_i^m\rangle$, $\langle a_i^5, a_i^6, y_i^1, \ldots, y_i^m\rangle$, $\langle a_i^7, a_i^8, z_i^1, \ldots, z_i^m\rangle \in\C$. We see that when $\langle w, a_i^1\rangle\notin \C$ and vertex $w$ is not covered by the gadget corresponding to $t_i$, then none of the vertices $w,x,y$, and $z$ are covered by $t_i$.  
\end{enumerate}
Therefore, for each $t_i \in T$, it must be the case that either: \[
\begin{array}{c}
    \{ \langle w, a_i^1 \rangle,   \langle x, a_i^3 \rangle,   \langle y, a_i^5 \rangle,   \langle z, a_i^7 \rangle,   \langle a_i^2, a_i^4, a_i^6, a_i^8 \rangle,  \langle w_i^1, \cdots, w_i^m \rangle,\\ \langle x_i^1, \cdots, x_i^m \rangle,   \langle y_i^1, \cdots, y_i^m \rangle,   \langle z_i^1, \cdots, z_i^m \rangle\}\subseteq \C, \quad\text{  or}\\
\{\langle w_i^1, \cdots, w_i^m, a_i^1, a_i^2 \rangle,   \langle x_i^1, \cdots, x_i^m, a_i^3, a_i^4 \rangle,   \langle y_i^1, \cdots, y_i^m, a_i^5, a_i^6 \rangle,\langle z_i^1, \cdots, z_i^m, a_i^7, a_i^8 \rangle. \} \subseteq \C.\end{array}\]
In the former case, we add $t_i$ to $M$.
Finally, $M$ is a perfect matching in $\I$ because $\C$ is a perfect $\Gamma$-cycle packing in $\J$, and in view of points~$(1)$ and ~$(2)$ above.

The proof of correctness follows the same reasoning as in Lemma~\ref{lem:n=2l=2k=inftyAnycs}. Hence, we have shown that {\sc Perfect $\Gamma$-Cycle Packing} is \NP-complete  when $n=2$, $\intv=\infty$, for all $i\in \N$ we have that $\natv_i\in \mathbb{N}_{\geq 0}\cup\{\infty\}$ and $\segv_i\in \mathbb{N}_{\geq 1}\cup\{\infty\}$, and for some $j\in \N$ ,  $\sv_j\in\mathbb{N}_{\geq 2}$  while the remaining countries  $j' \in \N\backslash\{j\}$ have $\sv_{j'}\in  \mathbb{N}_{\geq 2}\cup\{\infty\}$.
\end{proof}

\begin{lemma}\label{lem:l=4k>=2c=1s=2n=2} 
      For a set of country-specific parameters $\Gamma= (n, \emph{\intv}, \emph{\natv},\emph{\segv}, \emph{\sv})$ with 
    $n \geq 2$,  
    $\emph{\intv}\in \mathbb{N}_{\geq 4}$, 
    $\emph{\natv}_i \in  \mathbb{N}_{\geq 0}\cup\{\infty\}$ for all $i\in\N$,
    $\emph{\sv}_i \in\mathbb{N}_{\geq 2}\cup\{\infty\}$ for all $i\in \N$,
    $\emph{\segv}_i\in \mathbb{N}_{\geq 1}$ for all $\in \N$, 
    {\sc Perfect $\Gamma$-Cycle Packing} is \NP-complete.
\end{lemma}
\begin{proof}
 Membership of {\sc Perfect $\Gamma$-Cycle Packing} in \NP\ was shown in Lemma~\ref{lem:member}. The hardness result follows from Lemma~\ref{lem:l=4k>=2c=1s>=2n=2}, which establishes that only cycles within a single gadget are permitted, given the same country-specific parameters and the absence of national cycles within the reduction.
\end{proof}

\begin{lemma}\label{lem:l=inftyk=0c=1s=inftyn=2}
For a set of country-specific parameters $\Gamma= (n, \emph{\intv}, \emph{\natv},\emph{\segv}, \emph{\sv})$ with 
$n \geq 2$,  
$\emph{\intv}\in \mathbb{N}_{\geq 4}\cup\{ \infty\}$, 
$\emph{\natv}_j \in \mathbb{N}_{\geq 2}\cup\{\infty\}$ for some $j\in\N$,
and  $\emph{\natv}_i \in \mathbb{N}_{\geq 0}\cup\{\infty\}$ for all $i\in\N\backslash\{j\}$,
$\emph{\sv}_i \in\mathbb{N}_{\geq 2}\cup\{\infty\}$ for all $i\in \N$,
$\emph{\segv}_j=1$ for some $j\in \N$, and
$\emph{\segv}_i\in \mathbb{N}_{\geq 1}\cup\{\infty\}$ for all $i\in \N\setminus \{j\}$,
{\sc Perfect $\Gamma$-Cycle Packing} is \NP-complete.
\end{lemma}

\begin{proof}

    The membership of {\sc Perfect $\Gamma$-Cycle Packing} in \NP\ was shown in Lemma~\ref{lem:member}. To prove \NP-hardness, we reduce from the \NP-complete problem {\sc P4DM}. Given an arbitrary instance $\I$ of {\sc P4DM}, we construct an instance $\J$ of {\sc Perfect $\Gamma$-Cycle Packing} with the following parameters listed in the lemma statement.

\begin{figure*}[t]
    \centering
\resizebox{0.75\columnwidth}{!}{
\begin{tikzpicture}[triangle/.style = {regular polygon, regular polygon sides=3 }]

\begin{scope}[every node/.style={circle,thick,draw,minimum size=0.9cm}]
    \node[regular polygon, thick, regular polygon sides=3, draw,
     inner sep=0.005cm] (a1) at (-5,3) {$a_i^1$};
    \node (w) at (-5,0) {$w$};
    \node[regular polygon, thick, regular polygon sides=3, draw,
     inner sep=0.005cm] (a2) at (-5,7) {$a_i^2$};

    \node[regular polygon, thick, regular polygon sides=3, draw,
     inner sep=0.005cm] (a3) at (0,3) {$a_i^3$};
    \node[regular polygon, thick, regular polygon sides=3, draw,
     inner sep=0.005cm] (x) at (0,0) {$x$};
    \node (a4) at (0,7) {$a_i^4$};

    \node[regular polygon, thick, regular polygon sides=3, draw,
     inner sep=0.005cm] (a5) at (5,3) {$a_i^5$};
    \node (y) at (5,0) {$y$};
    \node[regular polygon, thick, regular polygon sides=3, draw,
     inner sep=0.005cm] (a6) at (5,7) {$a_i^6$};

    \node[regular polygon, thick, regular polygon sides=3, draw,
     inner sep=0.005cm] (a7) at (10.5,3) {$a_i^7$};
    \node[regular polygon, thick, regular polygon sides=3, draw,
     inner sep=0.005cm] (z) at (10.5,0) {$z$};
    
    \node (a8) at (10.5,7) {$a_i^8$};

\end{scope}

 \begin{scope}[>={Stealth[black]},
               every edge/.style={draw,very thick}]

    \path [->] (w) edge[bend left] node {} (a1);
    \path [->] (a1) edge[bend left] node {} (w);
    \path [->] (a1) edge[bend left] node {} (a2);
    \path [->] (a2) edge[bend left] node {} (a1);

    \path [->] (x) edge[bend left] node {} (a3);
    \path [->] (a3) edge[bend left] node {} (x);
    \path [->] (a3) edge[bend left] node {} (a4);
    \path [->] (a4) edge[bend left] node {} (a3);

    \path [->] (y) edge[bend left] node {} (a5);
    \path [->] (a5) edge[bend left] node {} (y);
    \path [->] (a5) edge[bend left] node {} (a6);
    \path [->] (a6) edge[bend left] node {} (a5);

    \path [->] (z) edge[bend left] node {} (a7);
    \path [->] (a7) edge[bend left] node {} (z);
    \path [->] (a7) edge[bend left] node {} (a8);
    \path [->] (a8) edge[bend left] node {} (a7);

    \path [->] (a2) edge[] node {} (a4);
    \path [->] (a4) edge[] node {} (a6);

    \path [->] (a6) edge[] node {} (a8);

    \path [->] (a8) edge[bend right=10] node {} (a2);
\end{scope}
\end{tikzpicture}}
\Description{}
    \caption{Gadget used in the proof of Lemma~\ref{lem:l=inftyk=0c=1s=inftyn=2} showing \NP-hardness when $n=2$, $\intv\in \mathbb{N}_{\geq 4}\cup\{ \infty\}$  there exists a $j \in \N$ such that $\natv_j\in \mathbb{N}_{\geq 2}\cup\{\infty\}$ (represented by the triangular vertices)  while the remaining countries $j'\in \N\backslash\{j\}$ have $\natv_{j'}\in \mathbb{N}_{\geq 0}\cup\{\infty\}$,  for every $i\in\N$ has $\sv_i\in \mathbb{N}_{\geq 2}\cup\{\infty\}$, and there is some $k \in \N$ such that $\segv_k=1$ while the remaining $k\in \N\backslash\{k\}$  have $\segv_i\in \mathbb{N}_{\geq 1}\cup\{\infty\}$. 
    }
    \label{fig:l=4k>=2c=1s=2n=2}
\end{figure*}

Our instance $\J$ consists of four sets of elements: $W$, $X$, $Y$, and $Z$, and a set of quadruples $T \subseteq W \times X \times Y \times Z$. We first create a vertex for each $u \in W \cup X \cup Y \cup Z$. For every quadruple $t_i = (w, x, y, z) \in T$, we create a corresponding gadget, as depicted in
Thus, the translated instance $\J$ has the set of vertices $V = W \cup X \cup Y \cup Z \cup \{a_i^1, \dots, a_i^8 \mid  t_i \in T\}$ and the set of edges
\[\begin{array}{c}A = \{(w, a_i^1), (x, a_i^3), (y, a_i^5), (z, a_i^7), (a_i^1, a_i^2), (a_i^3, a_i^4), (a_i^5, a_i^6), 
(a_i^7, a_i^8), (a_i^2, a_i^4, a_i^6, a_i^8) \mid  t_i \in T\}.\end{array}\]
We define the set of countries as $\N = \{1, 2\}$, where $V_1 = \{w, a_i^4, y, a_i^8 \mid  t_i = (w, x, y, z) \in T\}$ and $V_2 = V \setminus V_1$. The vertices of countries~$1$ and $2$ are represented by circular and triangular nodes in Figure~\ref{fig:l=4k>=2c=1s=2n=2}, respectively. The country-specific parameters are as stated in the lemma, with $\natv_2 \in\{\mathbb{N}_{\geq  2}, \infty\}$ and $\segv_2=1$.    
Next, we list the $\Gamma$-cycles in our instance $\J$: 
\[
\begin{array}{c}\gcyc = \{\langle w, a_i^1 \rangle, \langle x, a_i^3 \rangle, \langle y, a_i^5 \rangle, \langle z, a_i^7 \rangle, \langle a_i^1, a_i^2 \rangle, \langle a_i^3, a_i^4 \rangle,\\
\langle a_i^5, a_i^6 \rangle, \langle a_i^7, a_i^8 \rangle, \langle a_i^2, a_i^4, a_i^6, a_i^8 \rangle \mid  t_i = (w, x, y, z) \in T \}.
\end{array}\]
This set includes the five international cycles and the four $V_2$-cycles in each gadget.
To show that no other cycle in the instance is a $\Gamma$-cycle, consider an arbitrary cycle $C$ that spans over more than one gadget. Without loss of generality, assume that $C$ enters a gadget corresponding to $t_i = (w, x, y, z)$ via the vertex $w$ and exits via the vertex $x$. Thus, $C$ must contain the path
$[w, a_i^1, a_i^2, a_i^4, a_i^3, x]$. 
Observe that this path contains two $V_2$-segments of size $2$, namely $[a_i^1, a_i^2]$ and $[a_i^3, x]$. Since $\segv_2 = 1$, this path violates the condition for being a $\Gamma$-cycle, which requires each $V_2$-segment to have size at most $1$. Thus, $C$ cannot be a $\Gamma$-cycle.
This proves that all $\Gamma$-cycles must be contained within a single gadget.

Next, to show the correctness of our reduction, we show that when an instance $\I$ of {\sc P4DM} has a perfect matching exactly when the translated instance $\J$ has a perfect $\Gamma$-cycle packing. 

We now provide a proof of correctness by first assuming that we have an instance of {\sc P4DM} for which there is a perfect matching $ M $. We can then build a perfect $\Gamma$-cycle packing depending on whether each $ t_i \in T $ is included in $ M $. If $ t_i \in M $, we select the following cycles and add them to $\C$: 

\[ \{\langle w, a_i^1 \rangle, \langle x, a_i^3 \rangle, \langle y, a_i^5 \rangle, \langle z, a_i^7 \rangle, \langle a_i^2, a_i^4, a_i^6, a_i^8 \rangle \} \subseteq \C. \]

Thus, all the nodes within the gadget are covered by some $\Gamma$-cycle.  If a quadruple $t_i=(w,x,y,z)$ is not chosen in $M$, then the following cycles are added to $\C$: 
\[ \{\langle a_i^1, a_i^2 \rangle, \langle a_i^3, a_i^4 \rangle, \langle a_i^5, a_i^6 \rangle, \langle a_i^7, a_i^8 \rangle \} \subseteq \C. \]
Thus, all vertices in the instance $\J$ are covered by a $\Gamma$-cycle in the cycle packing $\C$. Therefore, $\J$ has a perfect $\Gamma$-cycle packing.

Conversely, suppose there is a perfect $\Gamma$-cycle packing $\C$  of instance $\J$. We construct a perfect matching $M$ in {\sc P4DM} instance $\I$ as follows. First, recall that no $\Gamma$-cycle contains edges from different gadgets. Hence, each cycle in $\C$ belongs to a single gadget, and we argue about the cycles of gadget $t_i$ that are selected in perfect $\Gamma$-cycle packing $\C$ via the following two cases:
\begin{enumerate}
    \item Suppose $\langle w, a_i^1\rangle\in \C$. Thus, the only way vertex $a_i^2$ can be selected is via $\C$ containing the cycle  $\langle a_i^2, a_i^4,a_i^6, a_i^8 \rangle\in \C$.  Then, the remaining vertices labelled with $a$ must be selected as such: $\langle x, a_i^3\rangle, \langle y, a_i^5\rangle, \langle z, a_i^7\rangle\in \C$. Thus, whenever a perfect $\Gamma$-cycle packing contains $\langle w, a_i^1\rangle$, it must be the case that $w,x, y,$ and $z$ are all selected by the same gadget corresponding to $t_i$. 
    
    \item Suppose $\langle w, a_i^1\rangle\notin \C$. For $a_i^1$ to be covered in $\C$, we see that the cycle $\langle a_i^1, a_i^2\rangle\in \C$. Thus, for vertices $a_i^4, a_i^6$, and $a_i^8$ to be covered, then $\langle a_i^3, a_i^4\rangle$, $\langle a_i^5, a_i^6\rangle$, $\langle a_i^7, a_i^8\rangle \in\C$. We see that when $\langle w, a_i^1\rangle\notin \C$ and vertex $w$ is not covered by the gadget corresponding to $t_i$, then none of the vertices $w,x,y$, and $z$ are covered by cycles in the gadget corresponding to $t_i$.  
\end{enumerate}
Therefore, it must be the case that either $$\{\langle w, a_i^1 \rangle, \langle x, a_i^3 \rangle, \langle y, a_i^5 \rangle, \langle z, a_i^7 \rangle, \langle a_i^2, a_i^4, a_i^6, a_i^8 \rangle  \}\subseteq \C \text{  or  }\{ \langle a_i^1, a_i^2 \rangle, \langle a_i^3, a_i^4 \rangle, \langle a_i^5, a_i^6 \rangle, \langle a_i^7, a_i^8 \rangle \} \subseteq \C$$, for each $t_i \in T$. In the former case, we add $t_i$ to $M$.
Finally, $M$ is a perfect matching in $\I$ because $\C$ is a perfect $\Gamma$-cycle packing in $\J$, and in view of points~$(1)$ and ~$(2)$ above.

Therefore, {\sc Perfect $\Gamma$-Cycle Packing} is \NP-complete when $n=2$, $\intv\in \mathbb{N}_{\geq 4}\cup\{ \infty\}$  there exists a $j \in \N$ such that $\natv_j\in \mathbb{N}_{\geq 2}\cup\{\infty\}$  while the remaining countries $j'\in \N\backslash\{j\}$ have $\natv_{j'}\in \mathbb{N}_{\geq 0}\cup\{\infty\}$,  for every $i\in\N$ has $\sv_i\in \mathbb{N}_{\geq 2}\cup\{\infty\}$, and there is some $k \in \N$ such that $\segv_k=1$ while the remaining $k\in \N\backslash\{k\}$  have $\segv_i\in \mathbb{N}_{\geq 1}\cup\{\infty\}$. 
\end{proof}

The next lemma has been sketched in Section~\ref{sec:complexityMax}. Due to this, the lemma is stated and proven for all $n\geq 2$ instead of $n=2$ as indicated in Table~\ref{t-dicho}.

\sampleLem*

\begin{proof}
Membership of \textsc{Perfect $\Gamma$-Cycle Packing} in \NP\ was shown in Lemma~\ref{lem:member}. To show \NP-hardness, we reduce from the \NP-complete problem \textsc{P4DM}.  Take an arbitrary instance $\I$ of \textsc{P4DM} and construct an instance $\J$ of \textsc{Perfect $\Gamma$-Cycle Packing} when $n \geq 2$,  $\intv=\infty$,  for all $i\in\N$, $\natv_i\in \mathbb{N}_{\geq 0}\cup\{\infty\},    \sv_i\in \mathbb{N}_{\geq 2}\cup\{\infty\}$, $\sv=\{\infty\}^n$, and for some $j \in \N$ $\segv_j\in \mathbb{N}_{\geq 2}\cup\{\infty\}$ with the remaining $j'\in\N\backslash\{j\}$ having $\segv_j\in \mathbb{N}_{\geq 1}\cup\{\infty\}$.

\begin{figure*}[t]
\centering
\resizebox{0.75\columnwidth}{!}{
\begin{tikzpicture}[triangle/.style = {regular polygon, regular polygon sides=3 }]

\begin{scope}[every node/.style={circle,thick,draw,minimum size=0.9cm}]
    \node[regular polygon, thick, regular polygon sides=3, draw, inner sep=-0.1cm] (a1) at (-5,3.5) {$w_{i,1}^{m-1}$};
    \node (w) at (-5,0) {$w$};
    \node (a2) at (-5,10) {$w_{i,1}^0$};

    \node (a3) at (0,3) {$x_{i,1}^{m+1}$};
    \node[regular polygon, thick, regular polygon sides=3, draw, inner sep=0.005cm] (x) at (0,0) {$x$};
    \node[regular polygon, thick, regular polygon sides=3, draw, inner sep=0.005cm] (a4) at (0,10) {$x_{i,1}^2$};

    \node[regular polygon, thick, regular polygon sides=3, draw,
     inner sep=-0.1cm] (a5) at (5,3.5) {$y_{i,1}^{m-1}$};
    \node (y) at (5,0) {$y$};
    \node (a6) at (5,10) {$y_{i,1}^0$};

    \node (a7) at (10.5,3) {$z_{i,1}^{m+1}$};
    \node[regular polygon, thick, regular polygon sides=3, draw, inner sep=0.005cm] (z) at (10.5,0) {$z$};
    
    \node[regular polygon, thick, regular polygon sides=3, draw, inner sep=0.005cm] (a8) at (10.5,10) {$z_{i,1}^2$};

    \node[regular polygon, thick, regular polygon sides=3, draw, inner sep=0.005cm] (z1) at (7.25,10) {$z_{i,1}^1$};
    \node[regular polygon, thick, regular polygon sides=3, draw, inner sep=0.005cm] (zs) at (11.5,5) {$z_{i,1}^{m}$};

    \node[regular polygon, thick, regular polygon sides=3, draw,
     inner sep=0.005cm] (y1) at (6,8) {$y_{i,1}^1$};
    \node[regular polygon, thick, regular polygon sides=3, draw, inner sep=0.005cm] (ys) at (6,1.5) {$y_{i,1}^{m}$};
    \node[regular polygon, thick, regular polygon sides=3, draw,
     inner sep=0.005cm] (x1) at (1,5) {$x_{i,1}^{m}$};
    \node[regular polygon, thick, regular polygon sides=3, draw, inner sep=0.005cm] (xs) at (-1.75,10) {$x_{i,1}^{1}$};

    \node[regular polygon,  regular polygon sides=3, draw,
     inner sep=-0.09cm] (wm-2) at (-4,5) {$w_{i,1}^{m-2}$};
     \node[regular polygon, thick, regular polygon sides=3, draw,
     inner sep=0.005cm] (x3) at (1,8) {$x_{i,1}^{3}$};
     \node[regular polygon, thick, regular polygon sides=3, draw,
     inner sep=-0.1cm] (ym-2) at (6,5) {$y_{i,1}^{m-2}$};
     \node[regular polygon, thick, regular polygon sides=3, draw,
     inner sep=0.005cm] (z3) at (11.5,8) {$z_{i,1}^{3}$};
    
    \node[regular polygon, thick, regular polygon sides=3, draw,
     inner sep=0.005cm] (w1) at (-4,8) {$w_{i,1}^1$};
  
    \node[regular polygon, thick, regular polygon sides=3, draw, inner sep=0.005cm] (ws) at (-4,1.5) {$w_{i,1}^{m}$};

    \node (k1) at (-3,9) {$w_{i,3}^{1}$};
    \node[regular polygon, thick, regular polygon sides=3, draw, inner sep=-0.05cm] (k2) at (-3,7
    ) {$w_{i,2}^1$};

    \node (k3) at (-3,2.5) {$w_{i,3}^{m}$};
    \node[regular polygon, thick, regular polygon sides=3, draw, inner sep=-0.005cm] (k4) at (-3,0.5
    ) {$w_{i,2}^m$};

    \node (k5) at (-2.75,11.1) {$x_{i,3}^{1}$};
    \node[regular polygon, thick, regular polygon sides=3, draw, inner sep=-0.05cm] (k6) at (-0.5,11.1) {$x_{i,2}^{1}$};

    \node (k7) at (2,6) {$x_{i,3}^{m}$};
    \node[regular polygon, thick, regular polygon sides=3, draw, inner sep=0.005cm] (k8) at (2,4
    ) {$x_{i,2}^{m}$};

    \node (k9) at (7,9) {$x_{i,3}^{1}$};
    \node[regular polygon, thick, regular polygon sides=3, draw, inner sep=-0.05cm] (k10) at (7,7
    ) {$y_{i,2}^{1}$};

    \node (k11) at (7,2.5) {$y_{i,3}^{m}$};
    \node[regular polygon, thick, regular polygon sides=3, draw, inner sep=0.05cm] (k12) at (7,0.5
    ) {$y_{i,2}^{m}$};

    \node (k13) at (6,11) {$z_{i,3}^{1}$};
    \node[regular polygon, thick, regular polygon sides=3, draw, inner sep=-0.05cm] (k14) at (8.25,11) {$z_{i,2}^{1}$};

    \node (k15) at (12.5,6) {$z_{i,3}^{m}$};
    \node[regular polygon, thick, regular polygon sides=3, draw, inner sep=0.005cm] (k16) at (12.5,4) {$z_{i,2}^{m}$};
    
    \node (k17) at (12.5,9) {$z_{i,3}^{3}$};
    \node[regular polygon, thick, regular polygon sides=3, draw, inner sep=-0.05cm] (k18) at (12.5,7) {$z_{i,2}^{3}$};

    \node[inner sep=-0.05cm] (k19) at (7,6) {$y_{i,3}^{m-2}$};
    \node[regular polygon, thick, regular polygon sides=3, draw, inner sep=-0.05cm] (k20) at (7,4) {$y_{i,2}^{m-2}$};

    \node (k21) at (2,9) {$x_{i,3}^{3}$};
    \node[regular polygon, thick, regular polygon sides=3, draw, inner sep=-0.05cm] (k22) at (2,7) {$x_{i,2}^{3}$};

    \node[inner sep=-0.03cm] (k23) at (-2.5,6) {$w_{i,3}^{m-2}$};
    \node[regular polygon, thick, regular polygon sides=3, draw, inner sep=-0.09cm] (k24) at (-2.5,3.75) {$w_{i,2}^{m-2}$};

\end{scope}
 \node (zdot) at (11.5,6.5) {$\cdots$};
  \node (ydot) at (6,6.5) {$\cdots$};
   \node (xdot) at (1,6.5) {$\cdots$};
    \node (wdot) at (-4,6.5) {$\cdots$};

 \begin{scope}[>={Stealth[black]},
               every edge/.style={draw,very thick}]

    \path [->] (w) edge[bend left] node {} (a1);
    \path [->] (ws) edge[bend left] node {} (w);
    \path [->] (a1) edge[bend left] node {} (a2);
    \path [->] (a2) edge[bend left=10] node {} (w1);
    \path [-] (w1) edge node {} (wdot);
    \path [->] (wdot) edge node {} (wm-2);
    \path [->] (wm-2) edge[bend left=10] node {} (a1);
    \path [->] (a1) edge[bend left=10] node {} (ws);

    \path [->] (x) edge[bend left] node {} (a3);
    \path [->] (a3) edge[bend left] node {} (x);
    \path [->] (a3) edge[bend left] node {} (a4);
    \path [->] (a4) edge[bend left=10] node {} (x3);
    \path [-] (x3) edge node {} (xdot);

    \path [->] (x1) edge[bend left=10] node {} (a3);
    \path [->] (xdot) edge node {} (x1);
    \path [->] (a2) edge node {} (xs);
    \path [->] (xdot) edge node {} (x1);
    \path [->] (xs) edge node {} (a4);

    \path [->] (y) edge[bend left] node {} (a5);
    \path [->] (ys) edge[bend left] node {} (y);
    \path [->] (a5) edge[bend left] node {} (a6);
    \path [->] (a6) edge[bend left=10] node {} (y1);

    \path [-] (y1) edge node {} (ydot);
    \path [->] (ydot) edge node {} (ym-2);
    \path [->] (ym-2) edge[bend left=10] node {} (a5);
    
    \path [->] (a5) edge[bend left=10] node {} (ys);

    \path [->] (z) edge[bend left] node {} (a7);
    \path [->] (a7) edge[bend left] node {} (z);
    \path [->] (a7) edge[bend left] node {} (a8);

    \path [->] (a8) edge[bend left=10] node {} (z3);
    \path [-] (z3) edge node {} (zdot);
    \path [->] (zdot) edge node {} (zs);
    \path [->] (zs) edge[bend left=10] node {} (a7);

    \path [->] (a4) edge[] node {} (a6);
    \path [->] (a6) edge[] node {} (z1);
    \path [->] (z1) edge[] node {} (a8);
    \path [->] (a8) edge[bend right=45] node {} (a2);

    \path [->] (k1) edge[] node {} (w1);
    \path [->] (w1) edge[] node {} (k2);
    \path [->] (k1) edge[bend left=15] node {} (k2);
    \path [->] (k2) edge[bend right=15] node {} (k1);

    \path [->] (k3) edge[] node {} (ws);
    \path [->] (ws) edge[] node {} (k4);
    \path [->] (k3) edge[bend left=15] node {} (k4);
    \path [->] (k4) edge[bend right=15] node {} (k3);

    \path [->] (k5) edge[] node {} (xs);
    \path [->] (xs) edge[] node {} (k6);
    \path [->] (k5) edge[bend left=15] node {} (k6);
    \path [->] (k6) edge[bend right=15] node {} (k5);

    \path [->] (k7) edge[] node {} (x1);
    \path [->] (x1) edge[] node {} (k8);
    \path [->] (k7) edge[bend left=15] node {} (k8);
    \path [->] (k8) edge[bend right=15] node {} (k7);

    \path [->] (k9) edge[] node {} (y1);
    \path [->] (y1) edge[] node {} (k10);
    \path [->] (k9) edge[bend left=15] node {} (k10);
    \path [->] (k10) edge[bend right=15] node {} (k9);

    \path [->] (k11) edge[] node {} (ys);
    \path [->] (ys) edge[] node {} (k12);
    \path [->] (k11) edge[bend left=15] node {} (k12);
    \path [->] (k12) edge[bend right=15] node {} (k11);

    \path [->] (k13) edge[] node {} (z1);
    \path [->] (z1) edge[] node {} (k14);
    \path [->] (k13) edge[bend left=15] node {} (k14);
    \path [->] (k14) edge[bend right=15] node {} (k13);

    \path [->] (k15) edge[] node {} (zs);
    \path [->] (zs) edge[] node {} (k16);
    \path [->] (k15) edge[bend left=15] node {} (k16);
    \path [->] (k16) edge[bend right=15] node {} (k15);

    \path [->] (k17) edge[bend left=15] node {} (k18);
    \path [->] (k18) edge[bend right=15] node {} (k17);

    \path [->] (k19) edge[bend left=15] node {} (k20);
    \path [->] (k20) edge[bend right=15] node {} (k19);

    \path [->] (k21) edge[bend left=15] node {} (k22);
    \path [->] (k22) edge[bend right=15] node {} (k21);

    \path [->] (k23) edge[bend left=15] node {} (k24);
    \path [->] (k24) edge[bend right=15] node {} (k23);
    \path [->] (k23) edge node {} (wm-2);
    \path [->] (wm-2) edge node {} (k24);

    \path [->] (k19) edge node {} (ym-2);
    \path [->] (ym-2) edge node {} (k20);

    \path [->] (k21) edge node {} (x3);
    \path [->] (x3) edge node {} (k22);

    \path [->] (k17) edge node {} (z3);
    \path [->] (z3) edge node {} (k18);

\end{scope}
\end{tikzpicture}}
\Description{}
    \caption{Gadget used in the proof of Lemma~\ref{lem:l=inftyk>=0c=1s>=2n=2} showing \NP-hardness when $n\geq 2$, $\intv=\infty$, any $\natv$ and $\sv$,  and for some $j \in \N$ has $\segv_j \in \mathbb{N}_{\geq 2}$ where index $m$ of the nodes is defined as $m = \segv_j +1$.}
    \label{fig:l=inftyk>=0c=1s>=2n=2}
\end{figure*}

    We first create a vertex for each $u \in W \cup X \cup Y \cup Z$. We create a gadget for every $t_i = (w, x, y, z) \in T$, and then all remaining vertices and edges are given per gadget, as depicted in Figure~\ref{fig:l=inftyk>=0c=1s>=2n=2}.  
Thus, the translated instance $\J$ has vertices  and edges: 
\[\begin{array}{c}
    V = W \cup X \cup Y \cup Z \cup    
          \{w_{i,1}^0, x_{i,1}^2, y_{i,1}^0, z_{i,1}^2, w_{i,1}^{m-1}, x_{i,1}^{m+1}, y_{i,1}^{m-1}, z_{i,1}^{m+1},w_{i,1}^q, w_{i,2}^q, w_{i,3}^{q}, x_{i,1}^{p}, x_{i,2}^{p}, x_{i,3}^{p},\\ 
          y_{i,1}^q, y_{i,2}^q, y_{i,3}^{q}, z_{i,1}^{p}, z_{i,2}^{p}, z_{i,3}^{p}\mid  q\in \{1, m-2\}\cup\{m\}, p\in \{1\}\cup\{3, m\} t_i = (w, x, y, z) \in T\} \quad \text{   and}\\
          \end{array}\]
   \[\begin{array}{c} A=\{ (w, w_{i,1}^{m-1}), (w_{i,1}^{m}, w), (w_{i,1}^{q}, w_{i,1}^{q+1}), (w_{i,1}^{m-1}, w_{i,1}^{0}), (w_{i,1}^{q}, w_{i,2}^{q}), (w_{i,3}^{q}, w_{i,2}^{q}), (w_{i,2}^{q}, w_{i,3}^{q}), \\(w_{i,3}^{q},w_{i,1}^{q}), 
    (w_{i,1}^{m}, w_{i,2}^{m}), (w_{i,3}^{m}, w_{i,2}^{m}), (w_{i,2}^{m}, w_{i,3}^{m}), (w_{i,3}^{m},w_{i,1}^{m}), \\
    (x_{i,1}^{1}, x_{i,2}^{1}), (x_{i,3}^{1}, x_{i,2}^{1}), (x_{i,2}^{1}, x_{i,3}^{1}), (x_{i,3}^{1},x_{i,1}^{1}), (x_{i,1}^{1},x_{i,1}^{2}) (x_{i,1}^{2}, x_{i,1}^{3}),\\
    (x, x_{i,1}^{m+1}), (x_{i,1}^{m+1}, x),  (x_{i,1}^{m+1}, x_{i,1}^2),  (x_{i,1}^{r}, x_{i,1}^{r+1}), (x_{i,1}^{r}, x_{i,2}^{r}), (x_{i,3}^{r}, x_{i,2}^{r}), (x_{i,2}^{r}, x_{i,3}^{r}), (x_{i,3}^{r},x_{i,1}^{r}),\\    
(y, y_{i,1}^{m-1}), (y_{i,1}^{m}, w), (y_{i,1}^{m-1}, y_{i,1}^{0}), (y_{i,1}^{q}, y_{i,1}^{q+1}), (y_{i,1}^{q}, y_{i,2}^{q}), (y_{i,3}^{q}, y_{i,2}^{q}), (y_{i,2}^{q}, y_{i,3}^{q}), (y_{i,3}^{q},y_{i,1}^{q}),\\ 
(y_{i,1}^{m}, y_{i,2}^{m}), (y_{i,3}^{m}, y_{i,2}^{m}), (y_{i,2}^{m}, y_{i,3}^{m}), (y_{i,3}^{m},y_{i,1}^{m}), \\
(z_{i,1}^{1}, z_{i,2}^{1}), (z_{i,3}^{1}, z_{i,2}^{1}), (z_{i,2}^{1}, z_{i,3}^{1}), (z_{i,3}^{1},z_{i,1}^{1}), (z_{i,1}^{m+1}, z_{i,1}^2), (z_{i,1}^{1},z_{i,1}^{2}), (z_{i,1}^{2}, z_{i,1}^{3}),\\
    (z, z_{i,1}^{m+1}), (z_{i,1}^{m+1}, z), (z_{i,1}^{r}, z_{i,1}^{r+1}), (z_{i,1}^{r}, z_{i,2}^{r}), (z_{i,3}^{r}, z_{i,2}^{r}), (z_{i,2}^{r}, z_{i,3}^{r}), (z_{i,3}^{r},z_{i,1}^{r}),\\ 
(w_{i,1}^{0}, x_{i,1}^{1}),(x_{i,1}^{2}, y_{i,1}^{0}), (y_{i,1}^{0}, z_{i,1}^{1}), (z_{i,1}^{2}, w_{i,1}^{0})
\mid q\in [1,m-2], r\in [3,m], t_i \in T \}.
\end{array}\] 
Note that every $w_{i,1}^q$, $y_{i,1}^q$ for $ q \in \{1,\ldots, m-2\} \cup \{m\}$ and $x_{i,1}^p, z_{i,1}^p$ for $p \in \{1\}\cup\{3, \ldots, m\}$  is connected in a cycle of size~$3$ to an international pairwise exchange with the same labels yet made distinct with the underline and overline. 

We have the countries $\N = \{1, 2\}$, where $$V_1 = 
    \{w, w_{i,1}^0, x_{i,1}^{m+1}, y, y_{i,1}^0, z_{i,1}^{m+1}, w_{i,3}^{q}, x_{i,3}^{p}, y_{i,3}^{q}, z_{i,3}^{p}   
    \mid  t_i \in T, q \in \{1,\ldots, m-2, m\}, p \in \{1\}\cup\{3, \ldots, m\}\},$$ 
and $V_2 = V \setminus V_1$. This is represented by circular and triangular nodes in Figure~\ref{fig:l=inftyk>=0c=1s>=2n=2}, respectively. 
The country-specific parameters are as stated in the lemma statement, with $\segv_2$ being some fixed $\sigma \in \mathbb{N}_{\geq 2}$ and $m = \sigma + 1$. 
Next, we give the $\Gamma$-cycles in our instance $\J$:
\[ 
\begin{array}{c}
 \gcyc= \{
\langle w, w_{i,1}^{m-1}, w_{i,1}^m \rangle, \, 
\langle x, x_{i,1}^{m+1} \rangle, 
\langle y, y_{i,1}^{m-1}, y_{i,1}^m \rangle, 
\langle z, z_{i,1}^{m+1} \rangle, \, \\
\langle w_{i,1}^0, w_{i,1}^1, \cdots, w_{i,1}^{m-1} \rangle, 
\langle x_{i,1}^2, \cdots, x_{i,1}^m, x_{i,1}^{m+1} \rangle, 
\langle y_{i,1}^0, y_{i,1}^1, \cdots, y_{i,1}^{m-1} \rangle,  
\langle z_{i,1}^2, \cdots, z_{i,1}^m, z_{i,1}^{m+1} \rangle, \\
\langle w_{i,1}^0, x_{i,1}^1, x_{i,1}^2, y_{i,1}^0, z_{i,1}^1, z_{i,1}^2 \rangle, 
\langle w_{i,3}^{q}, w_{i,2}^{q} \rangle, 
\langle w_{i,1}^q, w_{i,2}^{q}, w_{i,3}^{q} \rangle, 
\langle x_{i,3}^{p}, x_{i,2}^{p} \rangle, 
\langle x_{i,1}^p, x_{i,2}^{p},
x_{i,3}^{p} \rangle,
\langle y_{i,3}^{q}, y_{i,2}^{q} \rangle, \\
\langle y_{i,1}^q, y_{i,2}^{q}, y_{i,3}^{q} \rangle, 
\langle z_{i,3}^{p}, z_{i,2}^{p} \rangle, 
\langle z_{i,1}^p, z_{i,2}^{p},
z_{i,3}^{p} \rangle
\mid q \in \{1,\ldots, m-2\} \cup \{m\}, p \in \{1\}\cup\{3, \ldots, m\},  t_i \in T
\}. 
\end{array}
\]

To show that no other cycle in the instance is a $\Gamma$-cycle, take an arbitrary cycle $C$ spanning over more than one gadget. Thus, without loss of generality, $C$ must contain a path in the gadget corresponding to $t_i$, entering via $w$ and exiting via $x$. Hence, $C$ must contain the path 
$
[w, w_{i,1}^{m-1}, w_{i,1}^0, x_{i,1}^1, x_{i,1}^2, \cdots$, $x_{i,1}^m, x_{i,1}^{m+1}, x].
$
Observe that this path contains a $V_2$-segment of size $m$ (recall that $m = \sigma + 1$), namely $[x_{i,1}^1,\ldots,  x_{i,1}^m]$. Thus, $C$ cannot be a $\Gamma$-cycle given that $\segv_2 = \sigma$. A similar argument can be made when entering or exiting a gadget via different nodes. Therefore, all $\Gamma$-cycles are contained within a single gadget.

We now provide a proof of correctness by first assuming that we have an instance of {\sc P4DM} for which there is a perfect matching $ M $. We can then build a perfect $\Gamma$-cycle packing depending on whether each $ t_i = (w,x,y,z)\in T $ is included in $ M $. If $ t_i \in M $, we select the following cycles and add them to $\C$: 
\[
\begin{array}{l}
\langle w, w_{i,1}^{m-1}, w_{i,1}^m \rangle,  
\langle x, x_{i,1}^{m+1} \rangle, 
\langle y, y_{i,1}^{m-1}, y_{i,1}^m \rangle, 
\langle z, z_{i,1}^{m+1} \rangle, 
\langle w_{i,1}^0, x_{i,1}^1, x_{i,1}^2, y_{i,1}^0, z_{i,1}^1, z_{i,1}^2 \rangle, 
\langle w_{i,3}^{m}, w_{i,2}^{m} \rangle, \\
\langle w_{i,1}^q, w_{i,2}^{q}, w_{i,3}^{q} \rangle, 
\langle x_{i,3}^{1}, x_{i,2}^{1} \rangle, 
\langle x_{i,1}^p, x_{i,2}^{p},
x_{i,3}^{p} \rangle,
\langle y_{i,3}^{m}, y_{i,2}^{m} \rangle, 
\langle y_{i,1}^q, y_{i,2}^{q}, y_{i,3}^{q} \rangle 
\langle z_{i,3}^{1}, z_{i,2}^{1} \rangle, 
\langle z_{i,1}^p, z_{i,2}^{p},
z_{i,3}^{p} \rangle\in C, 
\end{array}
\]
for $ q \in \{1,\ldots, m-2\} \text{ and } p \in \{3, \ldots, m\}$. Thus, all the nodes within the gadget are covered by some $\Gamma$-cycle.  If a quadruple $t_i=(w,x,y,z)$ is not chosen in $M$, then the following cycles are added to $\C$:
\[
\begin{array}{c}
\langle w_{i,1}^0, w_{i,1}^1, \cdots, w_{i,1}^{m-1} \rangle, 
\langle x_{i,1}^2, \cdots, x_{i,1}^m, x_{i,1}^{m+1} \rangle, 
\langle y_{i,1}^0, y_{i,1}^1, \cdots, y_{i,1}^{m-1} \rangle,  
\langle z_{i,1}^2, \cdots, z_{i,1}^m, z_{i,1}^{m+1} \rangle, 
\langle w_{i,3}^{q}, w_{i,2}^q \rangle, \\
\langle w_{i,1}^m, w_{i,2}^{m}, w_{i,3}^{m} \rangle, 
\langle x_{i,3}^{p}, x_{i,2}^{p} \rangle, 
\langle x_{i,1}^1, x_{i,2}^{1},
x_{i,3}^{1} \rangle,
\langle y_{i,3}^{q}, y_{i,2}^{q} \rangle, 
\langle y_{i,1}^m y_{i,2}^{m}, y_{i,3}^{m} \rangle 
\langle z_{i,3}^{p}, z_{i,2}^{p} \rangle, 
\langle z_{i,1}^1, z_{i,2}^{1},
z_{i,3}^{1} \rangle\in \C,
\end{array}
\]
 for $ q \in \{1,\ldots, m-2\} \text{ and } p \in \{3, \ldots, m\}$. In this second case, observe that all of the gadget's vertices have been selected in some cycle, apart from $w$, $x$, $y$, and $z$.  As $M$ is a perfect matching, every vertex $u \in W\cup X \cup Y\cup Z$ will be covered by some gadget exactly once. Therefore, all vertices are covered by a cycle in $\gcyc$ and instance $\J$ has a perfect $\Gamma$-cycle packing.

Conversely, suppose there is a perfect $\Gamma$-cycle packing $\C$ of instance $\J$. We construct a perfect matching $M$ in {\sc P4DM} instance $\I$ as follows. First, recall that no $\Gamma$-cycle contains edges from different gadgets. Hence, each cycle in $\C$ belongs to a single gadget, and we argue about the cycles of gadget $t_i$ that are selected in perfect $\Gamma$-cycle packing $\C$ via the following two cases:
\begin{enumerate}
    \item Suppose $\langle w, w_{i,1}^{m-1}, w_{i,1}^m \rangle\in \C$. It is clear that $\langle w_{i,3}^{m}, w_{i,2}^m\rangle\in \C$ and $\langle  w_{i,1}^q, w_{i,2}^q, w_{i,3}^{q}\rangle\in \C$ must be selected in the perfect $\Gamma$-cycle packing $\C$ for each $q\in[1, m-2]$. For $w_{i,1}^0$ to be  covered, then $\langle w_{i,1}^0, x_{i,1}^1, x_{i,1}^2, y_{i,1}^0, z_{i,1}^1, z_{i,1}^2 \rangle\in \C$ and thus,  $\langle x_{i,3}^{1}, x_{i,2}^{1}\rangle, \langle z_{i,3}, z_{i,2}^{1}\rangle\in \C$. 
    In turn, $\langle x_{i,1}^q, x_{i,2}^{q}, x_{i,3}^{q}\rangle$, $\langle y_{i,1}^p, y_{i,2}^{p}, y_{i,3}^{p}\rangle$, $ \langle z_{i,1}^q, z_{i,2}^{q}, z_{i,3}^{q}\rangle, \in \C$ for vertices $x_{i,1}^p, y_{i,1}^q, z_{i,1}^p$ to be covered, for $q\in[1, m-2]$ and $p\in [3,m]$. 
    Then, for $x_{i,1}^{m+1}$, $y_{i,1}^{m-1}$, and $z_{i,1}^{m+1}$ to be covered, $\langle x, x_{i,1}^{m+1} \rangle, \langle y, y_{i,1}^{m-1}, y_{i,1}^m \rangle, \langle z, z_{i,1}^{m+1} \rangle \in \C$. Finally $\langle y_{i,3}^{m}, y_{i,2}^{m}\rangle\in \C$. Therefore, any perfect $\Gamma$-cycle packing $\C$ containing $\langle w, w_{i,1}^{m-1}, w_{i,1}^m \rangle$  must lead to $w,x,y,$ and $z$ must all be selected by the gadget corresponding to $t_i$.

    \item Suppose $\langle w, w_{i,1}^{m-1}, w_{i,1}^m \rangle\notin \C$. Then for $w_{i,1}^{m}$ to be selected, $\langle w_{i,1}^m, w_{i,2}^m,w_{i,3}^{m}\rangle\in\C$. Similarly, for $w_{i,1}^{m-1}$ to be selected, $\langle w_{i,1}^0, w_{i,1}^1, \cdots, w_{i,1}^{m-1} \rangle\in \C$, and then $\langle w_{i,2}^q,w_{i,3}^{q}\rangle\in\C$ for $q\in [1,m-2]$ such that there are no vertices connected to a $w$ vertex that is left uncovered. 
    As $w_{i,1}^0$ has already been selected, $\C$ must contain $\langle x_{i,1}^1, x_{i,2}^{1}, x_{i,3}^{1} \rangle, \langle z_{i,1}^1, z_{i,2}^{1}, z_{i,3}^{1} \rangle\in \C$ to cover $z_{i,1}^{1}$ and $x_{i,1}^{1}$. Moreover, to cover $y_{i,1}^0$,  $x_{i,1}^2$ and $z_{i,1}^2$,  the following cycles must be selected $\langle x_{i,1}^2, \cdots, x_{i,1}^m, x_{i,1}^{m+1} \rangle, \langle y_{i,1}^0, y_{i,1}^1, \cdots, y_{i,1}^{m-1} \rangle$, $\langle z_{i,1}^2, \cdots, z_{i,1}^m, z_{i,1}^m \rangle\in \C$. In turn, $\langle x_{i,2}^p, x_{i,3}^p\rangle, \langle y_{i,2}^q, y_{i,3}^{q} \rangle, \langle z_{i,2}^{p}, z_{i,3}^{p}\rangle \in \C$, to ensure that the remaining vertices connected to the middle cycles are selected, for $q\in [1,m-2]$ and $p\in [3,m]$. Finally, $\langle y_{i,1}^m,y_{i,2}^{m}, y_{i,3}^{m}\rangle$ must be selected in $\C$. 
    Therefore, we see that when $\langle w, w_{i,1}^{m-1}, w_{i,1}^m \rangle\notin \C$ and vertex $w$ is not covered by the gadget corresponding to $t_i$, then none of the vertices $w,x,y$, and $z$ are covered by $t_i$. 
\end{enumerate}

Therefore, for each $t_i \in T$ it must be the case that either:
\[
\begin{array}{c}\{
\langle w, w_{i,1}^{m-1}, w_{i,1}^m \rangle,  
\langle x, x_{i,1}^{m+1} \rangle, 
\langle y, y_{i,1}^{m-1}, y_{i,1}^m \rangle, 
\langle z, z_{i,1}^{m+1} \rangle, 
\langle w_{i,1}^0, x_{i,1}^1, x_{i,1}^2, y_{i,1}^0, z_{i,1}^1, z_{i,1}^2 \rangle, 
\langle w_{i,3}^{m}, w_{i,2}^{m} \rangle,\\
\langle w_{i,1}^q, w_{i,2}^{q}, w_{i,3}^{q} \rangle,
\langle x_{i,3}^{1}, x_{i,2}^{1} \rangle, 
\langle x_{i,1}^p, x_{i,2}^{p},
x_{i,3}^{p} \rangle,
\langle y_{i,3}^{m}, y_{i,2}^{m} \rangle, 
\langle y_{i,1}^q, y_{i,2}^{q}, y_{i,3}^{q} \rangle 
\langle z_{i,3}^{1}, z_{i,2}^{1} \rangle, 
\langle z_{i,1}^p, z_{i,2}^{p},
z_{i,3}^{p} \rangle
\}\subseteq \C\\
\end{array}\]
 \[\begin{array}{c}\{\text{or }\{\langle w_{i,1}^0, w_{i,1}^1, {\ldots}, w_{i,1}^{m-1} \rangle, 
\langle x_{i,1}^2, {\ldots}, x_{i,1}^m, x_{i,1}^{m+1} \rangle, 
\langle y_{i,1}^0, y_{i,1}^1, {\ldots}, y_{i,1}^{m-1} \rangle,  
\langle z_{i,1}^2, {\ldots}, z_{i,1}^m, z_{i,1}^{m+1} \rangle, 
\langle w_{i,3}^{q}, w_{i,2}^q \rangle,\\ 
\langle w_{i,1}^m, w_{i,2}^{m}, w_{i,3}^{m} \rangle, 
\langle x_{i,3}^{p}, x_{i,2}^{p} \rangle, 
\langle x_{i,1}^1, x_{i,2}^{1},
x_{i,3}^{1} \rangle,
\langle y_{i,3}^{q}, y_{i,2}^{q} \rangle, 
\langle y_{i,1}^m y_{i,2}^{m}, y_{i,3}^{m} \rangle 
\langle z_{i,3}^{p}, z_{i,2}^{p} \rangle, 
\langle z_{i,1}^1, z_{i,2}^{1},
z_{i,3}^{1} \rangle
\} \subseteq \C.
\end{array}
\]
In the former case, we add $t_i$ to $M$. Finally, $M$ is a perfect matching in $\I$ because $\C$ is a perfect $\Gamma$-cycle packing in $\J$, and in view of points~$(1)$ and ~$(2)$ above.

Therefore, {\sc Perfect $\Gamma$-Cycle Packing} is \NP-complete when $n \geq 2$,  $\intv=\infty$,  for all $i\in\N$, $\natv_i\in \mathbb{N}_{\geq 0}\cup\{\infty\},    \sv_i\in \mathbb{N}_{\geq 2}\cup\{\infty\}$, $\sv=\{\infty\}^n$, and for some $j \in \N$ $\segv_j\in \mathbb{N}_{\geq 2}\cup\{\infty\}$ with the remaining $j'\in\N\backslash\{j\}$ with $\segv_j\in \mathbb{N}_{\geq 1}\cup\{\infty\}$. 
\end{proof}

\begin{lemma}\label{lem:l=inftyc=inftys=1,infty}
    For a set of country-specific parameters $\Gamma= (n, \emph{\intv}, \emph{\natv},\emph{\segv}, \emph{\sv})$ with 
    $n \geq 2$,  
    $\emph{\intv}=\infty$, 
    $\emph{\natv}_i \in \mathbb{N}_{\geq 0}\cup\{\infty\}$ for all $i\in\N$,
    $\emph{\sv}_j\in \mathbb{N}$ for some $j\in \N$, and 
    $\emph{\sv}_i\in\mathbb{N}_{\geq 1}\cup\{\infty\}$ for all $i\in \N\backslash\{j\}$,
    $\emph{\segv}_i=\infty$ for all $i\in \N$,
    {\sc Perfect $\Gamma$-Cycle Packing} is \NP-complete.
\end{lemma}

\begin{proof}
\begin{figure*}[t]
    \centering
\resizebox{0.85\columnwidth}{!}{
\begin{tikzpicture}[triangle/.style = {regular polygon, regular polygon sides=3 }]

\node (xdt) at (1.5,2) {$\cdots$};
\node (ydt) at (9.5,2) {$\cdots$};
\node (zdt) at (17,2) {$\cdots$};
\begin{scope}[every node/.style={circle,thick,draw,minimum size=0.9cm}]
    \node (a1) at (-1,2) {$x_i^1$};
    \node (x) at (1.5,0) {$x$};
    \node[regular polygon, thick, regular polygon sides=3, draw,
     inner sep=0.005cm]  (a2) at (4, 2) {$x_i^{2\sigma}$};

    \node (a4) at (6.5,2) {$y_i^1$};
    \node (y) at (9.5,0) {$y$};
    \node[regular polygon, thick, regular polygon sides=3, draw,
     inner sep=0.005cm]  (a5) at (12.5,2) {$z_i^{2\sigma}$};

    \node (a7) at (14.5,2) {$z_i^1$};
    \node (z) at (17,0) {$z$};
    \node[regular polygon, thick, regular polygon sides=3, draw,
     inner sep=0.005cm]  (a8) at (19.5,2) {$z_i^{2\sigma}$};

    \node (a6) at (9.5,6) {$a_i^2$};
    \node (a3) at (2.5,7) {$a_i^1$};
    \node (a9) at (16,7) {$a_i^3$};

    \node[regular polygon, thick, regular polygon sides=3, draw,
     inner sep=0.005cm] (a11) at (8.5,7.5) {$a_i^5$};
    \node (a10) at (10.5,7.5) {$a_i^4$};

\end{scope}

 \begin{scope}[>={Stealth[black]},
               every edge/.style={draw,very thick}]
   
    \path [->] (a2) edge node {} (x);
    \path [->] (x) edge node {} (a1);
 
     \path [-] (a1) edge node {} (xdt);
     \path [->] (xdt) edge node {} (a2);
    \path [->] (a3) edge[] node {} (a1);
    \path [->] (a2) edge[] node {} (a3);

     \path [-] (a4) edge node {} (ydt);

     \path [->] (ydt) edge node {} (a5);

    \path [->] (z) edge node {} (a7);

    \path [->] (a8) edge node {} (z);

     \path [-] (a7) edge node {} (zdt);

     \path [->] (zdt) edge node {} (a8);

    \path [->] (a9) edge node {} (a7);
    \path [->] (a8) edge node {} (a9);

    \path [->] (a5) edge node {} (y);
    \path [->] (y) edge node {} (a4);

%center middle
    \path [->] (a5) edge node {} (a6);

    \path [->] (a6) edge node {} (a4);

%top cycle
    \path [->] (a3) edge node {} (a6);
    \path [->] (a6) edge node {} (a9);

    \path [->] (a9) edge node {} (a10);

    \path [->] (a11) edge node {} (a3);
    
    \path [->] (a10) edge node {} (a11);
    \path [->] (a11) edge[bend left] node {} (a10);
 
\end{scope}
\end{tikzpicture}}
\Description{}
    \caption{Gadget used in the proof of Lemma~\ref{lem:l=inftyc=inftys=1,infty} when $n=2$, $\intv=\infty$, $\segv=\{\infty\}^n$,  there is some $j \in\N$ (represented by the triangular vertices) such that for some fixed $\sigma \in \mathbb{N} $, $\sv_j=\sigma $ and the remaining $j' \in \N \backslash\{j\}$ have $\sv_{j'}\in \mathbb{N}_{\geq 1}\cup\{\infty\}$, and for all $i\in\N$, $ \natv_i\in \mathbb{N}_{\geq 0}\cup\{\infty\}$. 
    }
    \label{fig:l=inftyc=inftys=1,infty}
\end{figure*}
Membership of {\sc Perfect $\Gamma$-Cycle Packing} in \NP\ was shown in Lemma~\ref{lem:member}.  To show \NP-hardness, we transform  an instance $\I$ of \textsc{P3DM}~\cite{karp1972} to a constructed instance $\J$  of {\sc Perfect $\Gamma$-Cycle Packing} 
when $n=2$, $\intv=\infty$, $\segv=\{\infty\}^n$,  there is some $j \in\N$ such that $\sv_j\in \mathbb{N}$ and the remaining $j' \in \N \backslash\{j\}$ has $\sv_{j'}\in \mathbb{N}_{\geq 1}\cup\{\infty\}$, and for all $i\in\N$, $ \natv_i\in \mathbb{N}_{\geq 0}\cup\{\infty\}$. 

In our reduction, we first create a node for each $w \in X \cup Y \cup Z$. We create a gadget for every $t_i = (x, y, z) \in T$, and then all additional vertices and edges are given in a gadget, one for each $(x, y, z) = t_i \in T$, as depicted in Figure~\ref{fig:l=inftyc=inftys=1,infty}. 
Hence, the vertex set is defined as:
\[
V = X \cup Y \cup Z \cup \{x_i^1, \dots, x_i^{2\sigma}, y_i^1, \dots, y_i^{2\sigma}, z_i^1, \dots, z_i^{2\sigma}, a_i^1, \ldots, a_i^6 \mid  t_i \in T\}.
\]
The set of arcs is: 
\[
\begin{array}{c}
A = \{(x_i^{2\sigma}, x), (x, x_i^1), (a_i^1, x_i^1), (x_i^{2\sigma}, a_i^1), (x_i^{q}, x^{q+1}_i)
(y_i^{2\sigma}, y), (y, y_i^1),\\ (y_i^{2\sigma}, a_i^2),  (a_i^2, y_i^1), (y_i^{q}, y^{q+1}_i),
(z_i^{q}, z^{q+1}_i), (z_i^{2\sigma}, z), (z, z_i^1), (z_i^{2\sigma}, a_i^3), (a_i^3, z_i^1), \\
(a_i^1, a_i^2), (a_i^2, a_i^3), (a_i^3, a_i^4), (a_i^5, a_i^1), (a_i^5, a_i^4)\mid q \in [1, 2\sigma-1], t_i \in T\}.
\end{array}
\]
We define two countries, $\N = \{1,2\}$, with the vertex sets:
$V_2 = \{a_i^5, x_i^q, y_i^q, z_i^q \mid  t_i \in T \text{ and } q \in [1, \sigma]\}$ 
and $V_1 = V \setminus V_2$.
These are depicted in Figure~\ref{fig:l=inftyc=inftys=1,infty}, where vertices in countries $1$ and $2$ are represented by circular and triangular nodes, respectively. 
Moreover, we set $\sv_2 = \sigma$, and the remaining country-specific parameters are as stated in the lemma statement.

Next, we observe that the only $\Gamma$-cycles are those containing arcs from a single gadget, i.e., 
\[
\begin{array}{c}
\gcyc = \{ \langle x, x_i^1, \ldots, x_i^{2\sigma} \rangle, 
\langle y, y_i^1, \ldots, y_i^{2\sigma} \rangle,  
\langle z, z_i^1, \ldots, z_i^{2\sigma} \rangle,  
\langle a_i^1, x_i^1, \ldots, x_i^{2\sigma} \rangle,\\  
\langle a_i^2, y_i^1, \ldots, y_i^{2\sigma} \rangle,  
\langle a_i^3, z_i^1, \ldots, z_i^{2\sigma} \rangle,  
\langle a_i^1, a_i^2, a_i^3 , a_i^4, a_i^5 \rangle,
\langle a_i^4, a_i^5 \rangle  
\mid  t_i \in T \}.
\end{array}
\]
To show that these are the only $\Gamma$-cycles, consider a cycle $C$ that spans at least two gadgets. Thus, $C$ would have to enter and exit some gadget $t_i=(x,y,z)$. Without loss of generality, we assume that $C$ enters the gadget $t_i$ via the vertex $x$ and exits via the vertex $y$. 
Therefore, the path 
$[x,  x_i^1, \ldots, x_i^{2\sigma},  a_i^1, a_i^2$, $ y_i^1, \ldots, y_i^{2\sigma}, y]$ 
must be selected as part of the cycle $C$. However, this path contains $2\sigma$ $V_2$-segments, namely, the single-node segments $x_i^q$ and $y_i^q$ for even $q \leq 2\sigma$. Thus, $C \notin \gcyc$ and is not a $\Gamma$-cycle.

We next show that the instance $\I$ of \textsc{P3DM} has a perfect matching $M$ exactly when the constructed {\sc Perfect $\Gamma$-Cycle Packing} instance $\J$ has a perfect $\Gamma$-cycle packing. 
First, assume that $M$ is a perfect matching of \textsc{P3DM} instance $\I$. We construct a $\Gamma$-cycle packing in $\J$ given a perfect matching $M$. If a triple $t_i=(x,y,z)$ is chosen in $M$, then we add the $\Gamma$-cycles to $\C$:
$\langle x, x_i^1, \ldots, x_i^{2\sigma} \rangle, 
\langle y, y_i^1, \ldots, y_i^{2\sigma} \rangle,  
\langle z, z_i^1, \ldots, z_i^{2\sigma} \rangle$,
 and $\langle a_i^1, a_i^2, a_i^3, a_i^4, a_i^5 \rangle $. 
 Thus, all the nodes within the gadget are covered by some $\Gamma$-cycle in this case. Alternatively, if a triple $t_i=(x,y,z)$ is not chosen in $M$, then the following cycles are added to $\C$: 
$\langle a_i^1, x_i^1, \ldots, x_i^{2\sigma} \rangle,  
\langle a_i^2, y_i^1, \ldots, y_i^{2\sigma} \rangle, \langle a_i^3, z_i^1, \ldots, z_i^{2\sigma} \rangle$, and  $\langle a_i^4, a_i^5 \rangle\in C$. 
In this case, all vertices are selected except for $x$, $y$, and $z$. By assumption, $x$, $y$, and $z$ will be selected in a different gadget, given that $M$ is a perfect matching.  
Thus, when there is a perfect matching $M$, there exists a perfect $\Gamma$-cycle packing $\C$.

Conversely, suppose there is a perfect $\Gamma$-cycle packing $\C$  of instance $\J$. We construct a perfect matching $M$ in {\sc P3DM} instance $\I$ as follows. First, recall that no $\Gamma$-cycle contains edges from different gadgets. Hence, each cycle in $\C$ belongs to a single gadget, and we argue about the cycles of gadget $t_i$ that are selected in perfect $\Gamma$-cycle packing $\C$ via the following two cases:
\begin{enumerate}
    \item Suppose $\langle x, x_i^1, \ldots, x_i^{2\sigma} \rangle\in \C$. Thus, $a_i^1$ can be selected only by the cycle $\langle a_i^1, a_i^2, a_i^3, a_i^4, a_i^5\rangle$, and it is in $\C$. Thus, $\langle y, y_i^1, \ldots, y_i^{2\sigma}\rangle$, $\langle z, z_i^1, \ldots, z_i^{2\sigma} \rangle\in\C$ in order for $y_i^1, \ldots, y_i^{2\sigma}$ and $z_i^1, \ldots, z_i^{2\sigma}$ to be selected. Hence, we see that if in a perfect $\Gamma$-cycle packing $\langle x, x_i^1, \ldots, x_i^{2\sigma} \rangle\in \C$, then vertices $y$ and $z$ must also be selected in the same gadget.
    \item Suppose $\langle x, x_i^1, \ldots, x_i^{2\sigma} \rangle\notin \C$. Then $\langle a_i^1, x_i^1, \ldots, x_i^{2\sigma} \rangle\in \C$ to ensure that $x_i^1, \ldots, x_i^{2\sigma}$ are covered by some $\Gamma$-cycle. 
    For the vertices $a_i^4$ and $a_i^5$ to be selected, $\langle a_i^4, a_i^5\rangle\in\C$. 
    In turn, $\langle a_i^2, y_i^1, \ldots, y_i^{2\sigma} \rangle$,  and $ \langle a_i^3, z_i^1, \ldots, z_i^{2\sigma} \rangle\in \C$ to ensure that $a_i^2$ and $a_i^3$ are covered by a $\Gamma$-cycle. Therefore, when $\langle x, x_i^1, \ldots, x_i^{2\sigma} \rangle\notin \C$, the vertices $x,y$, and $z$ cannot be selected in by a cycle within gadget corresponding to $t_i$.  
\end{enumerate}

Therefore, it must be the case that either 
\[\begin{array}{c}
\{\langle x, x_i^1, \ldots, x_i^{2\sigma} \rangle, 
\langle y, y_i^1, \ldots, y_i^{2\sigma} \rangle,
\langle z, z_i^1, \ldots, z_i^{2\sigma} \rangle,\langle a_i^1, a_i^2, a_i^3, a_i^4, a_i^5 \rangle \}\subseteq \C, \text{ or }  \\
\{ \langle a_i^1, x_i^1, \ldots, x_i^{2\sigma} \rangle,  
\langle a_i^2, y_i^1, \ldots, y_i^{2\sigma} \rangle,
\langle a_i^3, z_i^1, \ldots, z_i^{2\sigma} \rangle,\langle a_i^4, a_i^5 \rangle\} \subseteq \C,
\end{array}\] for each $t_i \in T$. In the former case, we add $t_i$ to $M$. 
Finally, $M$ is a perfect matching in $\I$ because $\C$ is a perfect $\Gamma$-cycle packing in $\J$, and in view of points~$(1)$ and ~$(2)$ above.

Thus, we have shown {\sc Perfect $\Gamma$-Cycle Packing} to be \NP-complete when $n=2$ there is some $j \in\N$ such that $\sv_j\in \mathbb{N}$ and the remaining $i \in \N \backslash\{j\}$ have $\sv_i=\infty$, $\segv, \natv=\{\infty\}^n$. 
\end{proof}

\begin{lemma}\label{lem:lk=intyc=mathbbs=infty} 
    For a set of country-specific parameters $\Gamma= (n, \emph{\intv}, \emph{\natv},\emph{\segv}, \emph{\sv})$ with 
$n \geq 2$,  
$\emph{\intv}=\infty$, 
$\emph{\natv}_i \in  \mathbb{N}_{\geq 2}\cup\{\infty\}$ for all $i\in\N$,
$\emph{\sv}_i \in  \mathbb{N}_{\geq 2}\cup\{\infty\}$ for all $i\in \N$,
$\emph{\segv}_j\in  \mathbb{N}_{\geq 1}$ for some $j\in \N$, and
$\emph{\segv}_i\in\mathbb{N}_{\geq 1}\cup\{\infty\}$ for all $i\in \N\setminus \{j\}$,
{\sc Perfect $\Gamma$-Cycle Packing} is \NP-complete.
\end{lemma}
\begin{proof}
\begin{figure*}[t]
    \centering
\resizebox{0.75\columnwidth}{!}{
\begin{tikzpicture}[triangle/.style = {regular polygon, regular polygon sides=3 }]
\begin{scope}[every node/.style={minimum size=0.9cm}]
\node[] (wdt) at (-4,6.5) {$\cdots$};
\node[] (xdt) at (1,3.5) {$\cdots$};
\node[] (ydt) at (6,6.5) {$\cdots$};
\node[] (zdt) at (11.5,3.5) {$\cdots$};
\end{scope}
\begin{scope}[every node/.style={circle,thick,draw,minimum size=0.9cm}]
    \node (a1) at (-5,3) {$a_i^1$};
    \node[regular polygon, thick, regular polygon sides=3, draw,
     inner sep=0.005cm] (w) at (-5,0) {$w$};
    \node (a2) at (-5,10) {$a_i^2$};
        \node (w1) at (-4,8) {$w_i^1$};
    \node (wm) at (-4,5) {$w_i^m$};
    
    \node (a3) at (0,7) {$a_i^3$};
    \node (x) at (0,0) {$x$};
    \node[regular polygon, thick, regular polygon sides=3, draw,
     inner sep=0.005cm] (a4) at (0,10) {$a_i^4$};
    \node (x1) at (1, 5) {$x_i^1$};
     
    \node (xm) at (1, 2) {$x_i^m$};

    \node (a5) at (5,3) {$a_i^5$};
    \node[regular polygon, thick, regular polygon sides=3, draw,
     inner sep=0.005cm] (y) at (5,0) {$y$};
    \node (a6) at (5,10) {$a_i^6$};
    \node (y1) at (6,8) {$y_i^1$};
    \node (ym) at (6,5) {$y_i^m$};

    \node (a7) at (10.5,7) {$a_i^7$};
    \node (z) at (10.5,0) {$z$};
    \node (z1) at (11.5,5) {$z_i^1$};
     
    \node(zm) at (11.5,2) {$z_i^m$};
    \node[regular polygon, thick, regular polygon sides=3, draw,
     inner sep=0.005cm] (a8) at (10.5,10) {$a_i^8$};

\end{scope}

 \begin{scope}[>={Stealth[black]},
               every edge/.style={draw,very thick}]

    \path [->] (w) edge[bend left] node {} (a1);
    \path [->] (a1) edge[bend left] node {} (w);
    \path [->] (a1) edge[bend left] node {} (a2);
    \path [->] (a2) edge[bend left] node {} (w1);
    \path [-] (w1) edge node {} (wdt);
    \path [->] (wdt) edge node {} (wm);
    \path [->] (wm) edge[bend left] node {} (a1);

    \path [->] (x) edge[bend left] node {} (a3);
    \path [->] (a3) edge[bend left] node {} (a4);
    \path [->] (a4) edge[bend left] node {} (a3);
    \path [->] (a3) edge[bend left] node {} (x1);
    \path [-] (x1) edge node {} (xdt);
    \path [->] (xdt) edge node {} (xm);
    \path [->] (xm) edge[bend left] node {} (x);

    \path [->] (y) edge[bend left] node {} (a5);
    \path [->] (a5) edge[bend left] node {} (y);
    \path [->] (a5) edge[bend left] node {} (a6);
    \path [->] (a6) edge[bend left] node {} (y1);
    \path [-] (y1) edge node {} (ydt);
    \path [->] (ydt) edge node {} (ym);
    \path [->] (ym) edge[bend left] node {} (a5);

    \path [->] (z) edge[bend left] node {} (a7);
    \path [->] (a7) edge[bend left] node {} (a8);
    \path [->] (a8) edge[bend left] node {} (a7);
    \path [->] (a7) edge[bend left] node {} (z1);
    \path [-] (z1) edge node {} (zdt);
    \path [->] (zdt) edge node {} (zm);
    \path [->] (zm) edge[bend left] node {} (z);

    \path [->] (a2) edge[] node {} (a4);
    \path [->] (a4) edge[] node {} (a6);
    \path [->] (a6) edge[] node {} (a8);
    \path [->] (a8) edge[bend right=10] node {} (a2);

    \path [->] (wm) edge[bend left] node {} (w1);

    \path [->] (xm) edge[bend left] node {} (x1);
    \path [->] (ym) edge[bend left] node {} (y1);
    \path [->] (zm) edge[bend left] node {} (z1);
\end{scope}
\end{tikzpicture}}
\Description{}
    \caption{Gadget used in the proof of Lemma~\ref{lem:lk=intyc=mathbbs=infty} showing \NP-hardness 
     when $n\geq 2$, $\intv=\infty$  there exists a $j \in \N$ (represented by the triangular country) such that $\segv_j \in  \mathbb{N}_{\geq 1}$  while the remaining country $j'\in \N\backslash\{j\}$ has $\segv_{j'}\in \mathbb{N}_{\geq 1}\cup\{\infty\}$, and for every $i\in\N$ has $\sv_i\in \mathbb{N}_{\geq 2}\cup\{\infty\}$ and $\natv_i\in\mathbb{N}_{\geq 2}\cup\{\infty\}$. 
    Note that $m= \segv_j+1$. 
    }
    \label{fig:l=inftyk=inftyc=mathbbs=inftyn=2}
\end{figure*}

Membership of {\sc Perfect $\Gamma$-Cycle Packing} in \NP\ was shown in Lemma~\ref{lem:member}. To show \NP-hardness, we provide a reduction from {\sc P4DM}. Specifically, we transform an arbitrary instance $\I$ of {\sc P4DM} to a constructed instance $\J$  of {\sc Perfect $\Gamma$-Cycle Packing}  when $n=2$, $\intv=\infty$  there exists a $j \in \N$ such that $\segv_j \in  \mathbb{N}_{\geq 1}$  while the remaining country $j'\in \N\backslash\{j\}$ has $\segv_{j'}\in \mathbb{N}_{\geq 1}\cup\{\infty\}$, and for every $i\in\N$ has $\sv_i\in \mathbb{N}_{\geq 2}\cup\{\infty\}$ and $\natv_i\in \mathbb{N}_{\geq 2}\cup\{\infty\}$.

In our reduction, we create one vertex for every $u \in W\cup X \cup Y \cup Z$. The remaining vertices and all edges are added via gadgets. One gadget is created for each $t_i = (w, x, y, z) \in T$, as depicted in Figure~\ref{fig:l=inftyk=inftyc=mathbbs=inftyn=2}. For the gadget corresponding to $t_i$, we add the vertices $\{w_i^1, \cdots, w_i^{m}, x_i^1, \cdots, x_i^{m}, y_i^1, \cdots, y_i^{m}$, $z_i^1$, $\cdots$, $ z_i^{m}$, $a_i^1, \ldots a_i^8\}$. Hence, the vertex set is:
\[V = X \cup Y \cup Z \cup \{a_i^1, \cdots, a_i^8, w_i^1, \cdots, w_i^{m}, x_i^1, \cdots, x_i^{m}, y_i^1, \cdots, y_i^{m}, z_i^1, \cdots, z_i^{m} \mid  t_i \in T\}.\] 
In our reduction, we have $\N = \{1, 2\}$, with the vertices belonging to countries~1 and~2 depicted by the 
circular and triangular vertices in Figure~\ref{fig:l=inftyk=inftyc=mathbbs=inftyn=2}, respectively. Thus, $V_2 =  W \cup Y \cup \{a_i^{4},  a_i^{8} \mid  t_i = (w, x, y, z) \in T\}$, and $V_1 = V \setminus V_2$.
We let $\segv_2=m-1$ for some fixed $m \in \mathbb{N}$ while the proof holds for any $\segv_1\in\mathbb{N}_{\geq 1}\cup\{\infty\}$, and the remaining country-specific parameters are unbounded, as stated in the lemma. For each $t_i \in T$, we add the edges: 
\[\begin{array}{c}
A=\{(w, a_i^1), (a_i^1, a_i^2), (a_i^2, w_i^1), (w_i^q, w_i^{q+1}), (w_i^m, w_i^1), (w_i^m, a_i^1),  (a_i^1, w),
(x, a_i^3), \\ (a_i^3, a_i^4), (a_i^4, a_i^3), (a_i^3, x_i^1),  (x_i^q, x_i^{q+1}),  (x_i^m, x), (x_i^m, x_i^1),
(y, a_i^5), (a_i^5, a_i^6), \\(a_i^6, y_i^1), (y_i^q, y_i^{q+1}), (y_i^m, y_i^1), (y_i^m, a_i^5),  (a_i^5, y),
(z, a_i^7), (a_i^7, a_i^8), (a_i^8, a_i^7), (a_i^7, z_i^1),  (z_i^q, z_i^{q+1}), \\(z_i^m, z_i^1), (z_i^m, z), 
(a_i^{2}, a_i^{4}), (a_i^{4}, a_i^{6}), (a_i^{6}, a_i^{8}), (a_i^{8}, a_i^{2})
\mid q\in [1, m-1],  t_i = (w, x, y, z) \in T\}.\end{array}\]

Hence, we have constructed four international cycles of length 2, one international cycle of length 4, four national $V_2$-cycles of length~$m+2$, and four national $V_2$-cycles of length~$m$.
Next, we show that every $\Gamma$-cycle is contained within a single gadget and that 
\[\begin{array}{c}
\gcyc = \{ \langle w, a_i^1\rangle, \langle a_i^1, a_i^2, w_i^1, \ldots, w_i^m\rangle, \langle  w_i^1, \ldots, w_i^m\rangle, \langle a_i^3, a_i^4\rangle, \\ \langle x, a_i^3, x_i^1, \ldots, x_i^m\rangle, \langle x_i^1, \ldots, x_i^m\rangle, 
\langle y, a_i^5\rangle,  \langle a_i^5, a_i^6, y_i^1, \ldots, y_i^m\rangle, \langle y_i^1, \ldots, y_i^m\rangle, \\ \langle a_i^7, a_i^8\rangle, \langle z, a_i^7, z_i^1, \ldots, z_i^m\rangle, \langle z_i^1, \ldots, z_i^m\rangle, \langle a_i^{2},  a_i^4,a_i^{6},  a_i^8  \rangle
 \mid  t_i = (w,x, y, z) \in T\}.
\end{array}\] 
To show there are no other $\Gamma$-cycles, assume that there were some $\Gamma$-cycle $C$ that spans more than one gadget, including a gadget corresponding to $t_i$. Moreover, without loss of generality, assume that the cycle $C$ enters the gadget via the vertex $w$ and leaves via $x$ (the other possibilities of entering and exiting follow the same reasoning). This requires $C$ to contain the path $[w, a_i^1, a_i^2, a_i^4,  a_i^3, x_i^1, \ldots, x_i^{m}, x]$. However, this path contains the $V_2$-segment $[a_i^3, x_i^1, \ldots, x_i^{m}, x]$  of length~$m+2$ which is longer than $\segv_2=m$. Therefore, $C$ cannot be a $\Gamma$-cycle, and $\gcyc$ consists only of the cycles listed above.

We next show that the instance $\I$ of \textsc{P3DM} has a perfect matching $M$ exactly when the constructed {\sc Perfect $\Gamma$-Cycle Packing} instance $\J$ has a perfect $\Gamma$-cycle packing. 
First, assume that $M$ is a perfect matching of \textsc{P3DM} instance $\I$. We construct a $\Gamma$-cycle packing in $\J$ given a perfect matching $M$. If a quadruple $t_i=(w, x,y,z)$ is chosen in $M$, then we add the $\Gamma$-cycles to $\C$:
$$\langle w, a_i^1\rangle, 
\langle  w_i^1, \ldots, w_i^m\rangle, 
\langle x, a_i^3, x_i^1, \ldots, x_i^m\rangle, 
\langle y, a_i^5\rangle,  
\langle y_i^1, \ldots, y_i^m\rangle,
\langle z, a_i^7, z_i^1, \ldots, z_i^m\rangle,\langle a_i^{2},  a_i^4,a_i^{6},  a_i^8  \rangle\in \C.$$ Thus, all the nodes within the gadget are covered by some $\Gamma$-cycle in this case. Alternatively, if a quadruple $t_i=(w,x,y,z)$ is not chosen in $M$, then the following cycles are added to $\C$:
$$\langle a_i^1, a_i^2, w_i^1, \ldots, w_i^m\rangle,\langle a_i^3, a_i^4\rangle,\langle x_i^1, \ldots, x_i^m\rangle,
 \langle a_i^5, a_i^6, y_i^1, \ldots, y_i^m\rangle,  \langle a_i^7, a_i^8\rangle,\langle z_i^1, \ldots, z_i^m\rangle \in\C.$$
Note that when $t_i\notin M$ all nodes of the gadget, excluding $w$, $x$, $y$, and $z$, are covered by some $\Gamma$-cycle. Since $M$ is a perfect matching, every vertex $u \in W\cup X \cup Y \cup Z$ will be covered by some gadget exactly once. Therefore, all vertices are covered by a cycle in $\gcyc$, and instance $\J$ has a perfect $\Gamma$-cycle packing $\C$.

Conversely, suppose there is a perfect $\Gamma$-cycle packing $\C$ of instance $\J$. We construct a perfect matching $M$ in {\sc P3DM} instance $\I$ as follows. First, recall that no $\Gamma$-cycle contains edges from different gadgets. Hence, each cycle in $\C$ belongs to a single gadget, and we argue about the cycles of gadget $t_i$ that are selected in perfect $\Gamma$-cycle packing $\C$ via the following two cases:
\begin{enumerate}
    \item Suppose $ \langle w, a_i^1\rangle\in \C$. Then, for $w_i^1, \ldots, w_i^m$ to be selected, $\langle w_i^1, \ldots, w_i^m\rangle \in \C$. In turn, for $a_i^2$ to then be selected, $\langle a_i^2, a_i^4, a_i^6, a_i^8\rangle\in \C$. For $a_i^3$ to be selected in $\C$, it must be the case that $\langle x, a_i^3, x_i^1, \ldots, x_i^m\rangle\in\C$. For $y_i^1, \ldots, y_i^m$ to be selected, $\langle y_i^1, \ldots, y_i^m\rangle\in\C$, and in turn, for $a_i^5$ to be selected, $\langle y, a_i^5\rangle\in\C$. Finally, for $a_i^7$ to be selected, $\langle z, a_i^7, z_i^1, \ldots, z_i^m\rangle\in\C$. 
    Hence, when $ \langle w, a_i^1\rangle\in \C$ and $w$ is selected by the gadget $t_i$, then the vertices $w, x,y$, and $z$ must all be selected.
    \item Suppose $\langle w, a_i^1\rangle\notin \C$. For $a_i^1$ to be selected by $\C$, then $\langle a_i^1, a_i^2, w_i^1, \ldots, w_i^m\rangle\in\C$. 
    Thus, for $a_i^4, a_i^6$, and $a_i^8$ to be selected then $\langle a_i^3, a_i^4\rangle$, $ \langle a_i^5, a_i^6, y_i^1, \ldots, y_i^m\rangle$, $\langle a_i^7, a_i^8\rangle\in\C$. Finally, for $x_i^1, \ldots,x_i^m$, and $z_i^1, \ldots,z_i^m$ to be selected, then $\langle x_i^1, \ldots, x_i^m\rangle$, $\langle z_i^1, \ldots, z_i^m\rangle\in\C$. 
    Therefore, we see that when $\langle w, a_i^1\rangle\notin \C$ and vertex $w$ is not covered by the gadget corresponding to $t_i$, then none of the vertices $w,x,y$, and $z$ are covered by $t_i$. 
\end{enumerate}

Therefore, for each $t_i \in T$ it must be the case that either:
\[
\begin{array}{c}
\{\langle w, a_i^1\rangle, 
\langle  w_i^1, \ldots, w_i^m\rangle, 
\langle x, a_i^3, x_i^1, \ldots, x_i^m\rangle, 
\langle y, a_i^5\rangle,\\  
\langle y_i^1, \ldots, y_i^m\rangle,
\langle z, a_i^7, z_i^1, \ldots, z_i^m\rangle,\langle a_i^{2},  a_i^4,a_i^{6},  a_i^8  \rangle\}\subseteq \C, \quad  \text{ or  } \\ 
 \{ \langle a_i^1, a_i^2, w_i^1, \ldots, w_i^m\rangle,\langle a_i^3, a_i^4\rangle,\langle x_i^1, \ldots, x_i^m\rangle,
 \langle a_i^5, a_i^6, y_i^1, \ldots, y_i^m\rangle,  \langle a_i^7, a_i^8\rangle,\langle z_i^1, \ldots, z_i^m\rangle\} \subseteq \C.\end{array}\]
In the former case, we add $t_i$ to $M$.
Finally, $M$ is a perfect matching in $\I$ because $\C$ is a perfect $\Gamma$-cycle packing in $\J$, and in view of points~$(1)$ and ~$(2)$ above.

Therefore, {\sc Perfect $\Gamma$-Cycle Packing} is \NP-complete  when $n=2$, $\intv=\infty$  there exists a $j \in \N$ such that $\segv_j \in  \mathbb{N}_{\geq 1}$  while the remaining countries $j'\in \N\backslash\{j\}$ have $\segv_{j'}\in \mathbb{N}_{\geq 1}\cup\{\infty\}$, and for every $i\in\N$ has $\sv_i\in \mathbb{N}_{\geq 2}\cup\{\infty\}$ and $\natv_i\in \mathbb{N}_{\geq 2}\cup\{\infty\}$.  
\end{proof}

\begin{lemma}\label{lem:l=inftyk=2c=inftys=inftyn=2}
For a set of country-specific parameters $\Gamma= (n, \emph{\intv}, \emph{\natv},\emph{\segv}, \emph{\sv})$ with 
$n \geq 2$,  
$\emph{\intv}=\infty$, 
$\emph{\natv}_j \in \mathbb{N}_{\geq 0}$ for some $j\in\N$, 
and $\emph{\natv}_i \in\mathbb{N}_{\geq 0}\cup\{\infty\}$ for all $i\in\N\backslash\{j\}$,
$\emph{\sv}_i =\infty$ for all $i\in \N$,
$\emph{\segv}_i=\infty$ for all $i\in \N$,
{\sc Perfect $\Gamma$-Cycle Packing} is \NP-complete.
\end{lemma}
\begin{proof}
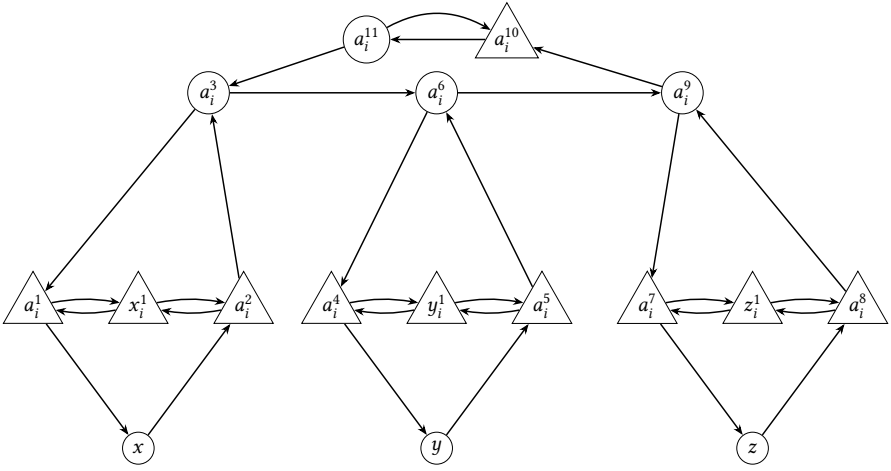
\begin{figure*}[t]
    \centering
\resizebox{0.85\columnwidth}{!}{
\begin{tikzpicture}[triangle/.style = {regular polygon, regular polygon sides=3}]

\begin{scope}[every node/.style={circle,thick,draw,minimum size=0.9cm,font=\huge}]
    \node[regular polygon, thick, regular polygon sides=3, draw,
     inner sep=0.005cm] (a11) at (-2,-1) {$a_i^1$};
    \node (x) at (1,-5) {$x$};
    \node [regular polygon, thick, regular polygon sides=3, draw,
     inner sep=0.005cm] (a12) at (4, -1) {$a_i^2$};
     \node [regular polygon, thick, regular polygon sides=3, draw,
     inner sep=0.005cm] (x1) at (1, -1) {$x_i^1$};

    \node[regular polygon, thick, regular polygon sides=3, draw,
     inner sep=0.005cm] (a31) at (15.5,-1) {$a_i^7$};
    \node (z) at (18.5,-5) {$z$};
    \node[regular polygon, thick, regular polygon sides=3, draw,
     inner sep=0.005cm]  (a32) at (21.5, -1) {$a_i^8$};
     \node [regular polygon, thick, regular polygon sides=3, draw,
     inner sep=0.005cm] (z1) at (18.5, -1) {$z_i^1$};

    \node[regular polygon, thick, regular polygon sides=3, draw,
     inner sep=0.005cm] (a21) at (6.5,-1) {$a_i^4$};
    \node (y) at (9.5,-5) {$y$};
    \node[regular polygon, thick, regular polygon sides=3, draw,
     inner sep=0.005cm]  (a22) at (12.5, -1) {$a_i^5$};
     \node [regular polygon, thick, regular polygon sides=3, draw,
     inner sep=0.005cm] (y1) at (9.5, -1) {$y_i^1$};

    \node  (a6) at (9.5,5) {$a_i^6$};
    \node(a3) at (3,5) {$a_i^3$};
    
    \node (a9) at (16.5,5) {$a_i^9$};

    \node[regular polygon, thick, regular polygon sides=3, draw,
     inner sep=0.005cm]  (t1) at (11.5,6.5) {$a_i^{10}$};
    \node  (t2) at (7.5,6.5) {$a_i^{11}$};
\end{scope}
 \begin{scope}[>={Stealth[black]},
               every edge/.style={draw,very thick}]
   
    \path [->] (x) edge node {} (a12);
    \path [->] (a11) edge node {} (x);
    \path [->] (a11) edge[bend left=10] node {} (x1);
    \path [->] (x1) edge[bend left=10] node {} (a11);
    \path [->] (a12) edge[bend left=10] node {} (x1);
    \path [->] (x1) edge[bend left=10] node {} (a12);
        
    \path [->] (a12) edge node {} (a3);
    \path [->] (a3) edge node {} (a11);

    \path [->] (y) edge node {} (a22);
    \path [->] (a21) edge[bend left=10] node {} (y1);
    \path [->] (y1) edge[bend left=10] node {} (a21);
    \path [->] (a22) edge[bend left=10] node {} (y1);
    \path [->] (y1) edge[bend left=10] node {} (a22);

    \path [->] (a21) edge node {} (y);

    \path [->] (a22) edge node {} (a6);
    \path [->] (a6) edge node {} (a21);

    \path [->] (z) edge node {} (a32);
    \path [->] (a31) edge[bend left=10] node {} (z1);
    \path [->] (z1) edge[bend left=10] node {} (a31);
    \path [->] (a32) edge[bend left=10] node {} (z1);
    
    \path [->] (z1) edge[bend left=10] node {} (a32);
    \path [->] (a31) edge node {} (z);

    \path [->] (a32) edge node {} (a9);
    \path [->] (a9) edge node {} (a31);

    \path [->] (a6) edge node {} (a9);
    \path [->] (a9) edge node {} (t1);
    \path [->] (a3) edge node {} (a6);

    \path [->] (t1) edge node {} (t2);
    \path [->] (t2) edge[bend left] node {} (t1);

    \path [->] (t2) edge node {} (a3);
\end{scope}
\end{tikzpicture}}
\Description{}
    \caption{Gadget used in the proof of Lemma~\ref{lem:l=inftyk=2c=inftys=inftyn=2} showing \NP-hardness when $\N=\{1,2\}$, $ \intv=\infty, \sv, \segv=\{\infty\}^2$, and for some fixed $\natv_2=0$ where country~$2$ is represented by the triangular vertices while $\natv_{1}\in \mathbb{N}_{\geq 0}\cup\{\infty\} $.}
    \label{fig:l=inftyk=0c=inftys=inftyn=2}
\end{figure*}

Membership of {\sc Perfect $\Gamma$-Cycle Packing} in \NP\ was shown in Lemma~\ref{lem:member}. To show \NP-hardness, we split the proof into three cases corresponding to \emph{(i)} $\natv_j=0$, \emph{(ii)} $\natv_j=2$ , and \emph{(iii)} $\natv_j\in \mathbb{N}_{\geq 3}$ for some $j\in \N$.  The proof of case \emph{(iii}) 
follows from Theorem~\ref{t-dicho}. To see this, assume that all vertices belong to country $j$ in the instance $\J$ constructed by Theorem~\ref{t-dicho}; then a perfect $\Gamma$-cycle packing in $\J$ is a perfect $\natv_j$-cycle packing. We now address cases (i) and (ii). 

\paragraph{Case (i).} 
We transform an arbitrary instance $\I$ of \textsc{P3DM} to a constructed instance $\J$  of {\sc Perfect $\Gamma$-Cycle Packing} when $n=2$,  there is some $j\in\N$ such that $\natv_j=0$, $\intv=\infty, \sv=\{\infty\}^2$ and $\segv=\{\infty\}^2$. 
In our reduction, we create one vertex for every $w \in  X \cup Y \cup Z$. The remaining vertices and all edges are added via gadgets. One gadget is created for each $t_i=(x,y,z) \in T$, as depicted in Figure~\ref{fig:l=inftyk=0c=inftys=inftyn=2}. We add the following vertices $\{a_i^1, \cdots, a_i^9, x_i^1, y_i^1, z_i^1\}$ for the gadget corresponding to $t_i$. Moreover, we add the following vertices and edges:
\[\begin{array}{c}
V = X \cup Y\cup Z\cup \{ a_i^1, a_i^2, a_i^3, a_i^4, a_i^5, a_i^6, a_i^7, a_i^8, a_i^9, a_i^{10}, a_i^{11}, x_i^1, y_i^1, z_i^1 \mid t_i \in T\} \quad \text{  and  }\\
A = \{ (x, a_i^2),  (a_i^1, x),  a_i^1, x_i^1),  (x_i^1, a_i^1),  (a_i^2, x_i^1),  (x_i^1, a_i^2),  (a_i^2, a_i^3),  (a_i^3, a_i^1),  
(y, a_i^5),  (a_i^5, y), \\ (a_i^4, y_i^1),  
(y_i^1, a_i^4), (a_i^5, y_i^1),  
(y_i^1, a_i^5), (a_i^4, y),  
(a_i^5, a_i^6), (a_i^6, a_i^4),  
(z, a_i^8),  (a_i^7, z_i^1),  
(z_i^1, a_i^7), (a_i^8, z_i^1),  \\
(a_i^7, z),  (a_i^8, a_i^9),  
(a_i^9, a_i^7),  (a_i^6, a_i^9),  
(a_i^9, a_i^{10}), (a_i^{10}, a_i^{11}), (a_i^{11}, a_i^{10}),  (a_i^{10}, a_i^{3}),   (a_i^3, a_i^6)  \mid t_i\in T
\}.
\end{array}
\]

In our reduction, we have $\N=\{1,2\}$, with the vertices belonging to countries~$1$ and $2$ depicted by the circular and triangular vertices in Figure~\ref{fig:l2kinty}, respectively. Thus, $V_1=\{x,y,z,a_i^3,a_i^6, a_i^9, a_i^{11} \mid  t_i=(x,y,z)\in T\}$ and $V_2 =V\backslash V_1$.
We let $\natv_1\in \mathbb{N}_{\geq 0}\cup\{\infty\}$,  $\natv_2=0$, $\segv=(\infty,\infty)$ and $\sv=(\infty,\infty)$ and, thus, our restrictions for the country-specific parameters, as given in the lemma statement, are being upheld.
Hence, observe that the following  cycles are $\Gamma$-cycles: 
\[\begin{array}{c}
\{
\langle x, a_i^2, x_i^1, a_i^1\rangle, 
\langle a_i^3, a_i^1, x_i^1, a_i^2 \rangle,  
\langle y, a_i^5, y_i^1, a_i^4 \rangle,
\langle a_i^4, y_i^1, a_i^5, a_i^6\rangle,\\
\langle z, a_i^8, z_i^1, a_i^9\rangle, 
\langle a_i^7, z_i^1, a_i^8, a_i^9\rangle,  
\langle  a_i^3, a_i^6, a_i^9, a_i^{10},a_i^{11}\rangle, 
\langle  a_i^{10},a_i^{11}\rangle
\mid t_i \in T\}\subset \gcyc.
\end{array}\] 
Note that there are other $\Gamma$-cycles within the instance. Yet, we will show that these are the only cycles that will be selected in any perfect $\Gamma$-cycle packing. 

Next, we give a proof of correctness and first assume that there is a perfect matching $M$ of instance $\I$ of \textsc{P3DM} and construct a perfect $\Gamma$-cycle packing $\C$ of instance $\J$. 
If a triple $t_i=(x,y,z)$ is chosen in $M$, then the following cycles are added to $\C$: 
\[\langle x, a_i^2, x_i^1, a_i^1\rangle, 
\langle y, a_i^5, y_i^1, a_i^4\rangle,
\langle z, a_i^8, z_i^1, a_i^9\rangle, 
\langle a_i^9, a_i^3, a_i^6,  a_i^{10},a_i^{11}\rangle  \in \C
 \]
Thus, all the nodes within the gadget are covered by some $\Gamma$-cycle.  If a triple $t_i=(x,y,z)$ is not chosen in $M$, then the following cycles are added to $\C$: 
 \[
\langle a_i^3, a_i^1, x_i^1, a_i^2 \rangle, 
\langle a_i^4, y_i^1, a_i^5, a_i^6\rangle,
\langle a_i^7, z_i^1, a_i^8, a_i^9\rangle, 
\langle a_i^{10},a_i^{11}\rangle
\in \C
 \]
Observe in this second case that all of the gadget's vertices have been selected in some cycle, apart from $x$, $y$, and $z$.  As $M$ is a perfect matching, every vertex $w \in X \cup Y\cup Z$ will be covered by some gadget exactly once. Therefore, all vertices are covered by a cycle in $\gcyc$ and instance $\J$ has a perfect $\Gamma$-cycle packing. 

Conversely, suppose there is a perfect $\Gamma$-cycle packing $\C$ of instance $\J$. We construct a perfect matching $M$ in {\sc P3DM} instance $\I$ as follows. We argue about the cycles of gadget $t_i$ that are selected in perfect $\Gamma$-cycle packing $\C$ via the following two cases:
\begin{enumerate}
    \item Suppose $\langle x, a_i^2, x_i^1, a_i^1\rangle\in \C$. For the vertex $a_i^3$ to be covered by a $\Gamma$-cycle then $\langle a_i^3, a_i^6, a_i^9,  a_i^{10},a_i^{11} \rangle$ must also be in $\C$. Following this, we see that $\langle y, a_i^5, y_i^1, a_i^4 \rangle$ and $\langle z, a_i^8, z_i^1, a_i^9\rangle $ must also be selected to ensure that $ a_i^5, y_i^1, a_i^4 a_i^8, z_i^1$, and $a_i^9$ are all selected.  
    \item Suppose $\langle x, a_i^2, x_i^1, a_i^1\rangle\notin \C$. Consider the $\Gamma$-cycles that could select $a_i^2$ that contain the path $[x, a_i^2]$. As  $\langle x, a_i^2, x_i^1, a_i^1\rangle\notin\C$, this entails that either \emph{(a)} the cycle $\langle x, a_i^2, x_i^3, a_i^1 \rangle$ or \emph{(b)} the cycle must span over more than one gadget and contain the path $[x, a_i^2, a_i^3, a_i^6]$. 
    \begin{enumerate}
        \item If the cycle $\langle x, a_i^2, x_i^3, a_i^1 \rangle\in\C$, this would mean that the node $x_i^1$ would be left uncovered by $\C$. Therefore, $\langle x, a_i^2, x_i^3, a_i^1 \rangle\notin\C$.
        \item If the path $[x, a_i^2, a_i^3, a_i^6]$ were contained within a cycle in the perfect $\Gamma$-cycle packing, this would entail that the cycle $\langle a_i^1, x_i^1\rangle$ is selected by $\C$ as well.  However, this cannot be the case, as $\langle a_i^1, x_i^1\rangle$ is a national $V_2$-cycle, yet $\natv_2=0$. Hence, this path cannot be contained in any cycle within a perfect $\Gamma$-cycle packing. 
    \end{enumerate}
    Therefore, the only remaining cycle that can cover $a_i^2$ when $\langle x, a_i^2, x_i^1, a_i^1\rangle\notin \C$ is $\langle a_i^1, x_i^1, a_i^2, a_i^3\rangle\in\C$. 
    For $ a_i^{10}$ and $a_i^{11}$ to be selected, $\langle  a_i^{10},a_i^{11}\rangle\in\C$. 
    Then, for $a_i^6$ to be covered by a $\Gamma$-cycle, either $\langle a_i^6, a_i^4, y_i^1, a_i^5\rangle\in\C$ or $\langle a_i^6, a_i^4, y, a_i^5\rangle\in\C$. In the latter case, this would leave the node $y_i^1$ uncovered by any $\Gamma$-cycle. Therefore, $\langle a_i^6, a_i^4, y_i^1, a_i^5\rangle\in\C$. Similar reasoning leads to $\langle a_i^7, z_i^1, a_i^8, a_i^9\rangle \in \C$. 
\end{enumerate}

 % \noindent 
 Therefore, it must be the case that either 
 \[\begin{array}{c}
 \{\langle x, a_i^2, x_i^1, a_i^1\rangle, 
\langle y, a_i^5, y_i^1, a_i^4 \rangle,
\langle z, a_i^8, z_i^1, a_i^9\rangle, 
\langle a_i^9, a_i^3, a_i^6,  a_i^{10},a_i^{11}\rangle\}\subseteq \C \text{  or}\\
\{\langle a_i^3, a_i^1, x_i^1, a_i^2 \rangle,  
\langle a_i^4, y_i^1, a_i^5, a_i^6\rangle,
\langle a_i^7, z_i^1, a_i^8, a_i^9\rangle, 
 \langle a_i^{10},a_i^{11}\rangle \} \subseteq \C,
 \end{array}\] for each $t_i \in T$. In the former case, we add $t_i$ to $M$. 
Finally, $M$ is a perfect matching in $\I$ because $\C$ is a perfect $\Gamma$-cycle packing in $\J$, and in view of points~$(1)$ and ~$(2)$ above. 

Therefore, {\sc Perfect $\Gamma$-Cycle Packing} is \NP-complete when $n =2$, $\intv=\infty$, $\natv_j=0$ for some $j \in \N$, and $\segv=\sv=(\infty,\infty)$.

\paragraph{Case (ii).}
We transform an arbitrary instance $\I$ of \textsc{P3DM} to a constructed instance $\J$  of {\sc Perfect $\Gamma$-Cycle Packing} when $n=2$,  there is some $j\in\N$ such that $\natv_j=2$, $\intv=\infty, \sv=\{\infty\}^2$ and $\segv=\{\infty\}^2$. 

\begin{figure*}[t]
    \centering
\resizebox{0.95\columnwidth}{!}{
\begin{tikzpicture}[triangle/.style = {regular polygon, regular polygon sides=3}]
\begin{scope}[every node/.style={circle,thick,draw,minimum size=0.9cm,font=\huge}]
    \node[regular polygon, thick, regular polygon sides=3, draw,
     inner sep=0.005cm] (a11) at (-4,-3) {$a_i^1$};
    \node (x) at (0,-9) {$x$};
    \node [regular polygon, thick, regular polygon sides=3, draw,
     inner sep=0.005cm] (a12) at (4, -3) {$a_i^2$};
  
    \node[regular polygon, thick, regular polygon sides=3, draw,
     inner sep=0.005cm] (a31) at (18.5,-3) {$a_i^7$};
    \node (z) at (22.5,-9) {$z$};
    \node[regular polygon, thick, regular polygon sides=3, draw,
     inner sep=0.005cm]  (a32) at (27, -3) {$a_i^8$};
     
    \node[regular polygon, thick, regular polygon sides=3, draw,
     inner sep=0.005cm] (a21) at (7.5,-3) {$a_i^4$};
    \node (y) at (11.5,-9) {$y$};
    \node[regular polygon, thick, regular polygon sides=3, draw,
     inner sep=0.005cm]  (a22) at (15.5, -3) {$a_i^5$};

    \node  (a6) at (11.5,3) {$a_i^6$};
    \node(a3) at (0,3) {$a_i^3$};
    
    \node (a9) at (23,3) {$a_i^9$};

    \node[regular polygon, thick, regular polygon sides=3, draw,
     inner sep=0.005cm]  (y1) at (11.5,-5) {$y_i^1$};
    \node[regular polygon, thick, regular polygon sides=3, draw,
     inner sep=0.005cm]  (y2) at (11.5,-3) {$y_i^2$};
    \node[regular polygon, thick, regular polygon sides=3, draw,
     inner sep=0.005cm]  (y3) at (11.5,-1) {$y_i^3$};

     \node[regular polygon, thick, regular polygon sides=3, draw,
     inner sep=0.005cm]  (x1) at (0,-5) {$x_i^1$};
    \node[regular polygon, thick, regular polygon sides=3, draw,
     inner sep=0.005cm]  (x2) at (0,-3) {$x_i^2$};
    \node[regular polygon, thick, regular polygon sides=3, draw,
     inner sep=0.005cm]  (x3) at (0,-1) {$x_i^3$};

     \node[regular polygon, thick, regular polygon sides=3, draw,
     inner sep=0.005cm]  (z1) at (23,-5) {$z_i^1$};
    \node[regular polygon, thick, regular polygon sides=3, draw,
     inner sep=0.005cm]  (z2) at (23,-3) {$z_i^2$};
    \node[regular polygon, thick, regular polygon sides=3, draw,
     inner sep=0.005cm]  (z3) at (23,-1) {$z_i^3$};

     \node[regular polygon, thick, regular polygon sides=3, draw,
     inner sep=0.005cm]  (t1) at (13.5,5.5) {$a_i^{10}$};
    \node  (t2) at (9.5,5.5) {$a_i^{11}$};
\end{scope}
 \begin{scope}[>={Stealth[black]},
               every edge/.style={draw,very thick}]

    \path [->] (a11) edge node {} (x);
    \path [->] (a11) edge node {} (x3);
    \path [->] (x) edge node {} (a12);
    \path [->] (a12) edge node {} (x1);
    \path [->] (x1) edge[bend left] node {} (x2);
    \path [->] (x1) edge node {} (a11);
    \path [->] (x2) edge[bend left] node {} (x1);
     \path [->] (x3) edge[bend left] node {} (x2);
    \path [->] (x2) edge[bend left] node {} (x3);
    \path [->] (x1) edge[bend right=50] node {} (x3);
    \path [->] (x3) edge node {} (a12);
        
    \path [->] (a12) edge node {} (a3);
    \path [->] (a3) edge node {} (a11);

    \path [->] (y) edge node {} (a22);

    \path [->] (a21) edge node {} (y);

    \path [->] (a22) edge node {} (a6);
    \path [->] (a6) edge node {} (a21);

   \path [->] (a22) edge node {} (y1);
    \path [->] (y1) edge[bend left] node {} (y2);
    \path [->] (y1) edge node {} (a21);
    \path [->] (y2) edge[bend left] node {} (y1);
     \path [->] (y3) edge[bend left] node {} (y2);
    \path [->] (y2) edge[bend left] node {} (y3);
    \path [->] (y1) edge[bend right=50] node {} (y3);
    \path [->] (y3) edge node {} (a22);
    \path [->] (a21) edge node {} (y3);

    \path [->] (z) edge node {} (a32);
    \path [->] (a31) edge node {} (z);

    \path [->] (a32) edge node {} (a9);
    \path [->] (a9) edge node {} (a31);

   \path [->] (a32) edge node {} (z1);
    \path [->] (z1) edge[bend left] node {} (z2);
    \path [->] (z1) edge node {} (a31);
    \path [->] (z2) edge[bend left] node {} (z1);
     \path [->] (z3) edge[bend left] node {} (z2);
    \path [->] (z2) edge[bend left] node {} (z3);
    \path [->] (z1) edge[bend right=50] node {} (z3);
    \path [->] (a31) edge node {} (z3);
    \path [->] (z3) edge node {} (a32);

    \path [->] (a6) edge node {} (a9);
    \path [->] (a9) edge node {} (t1);
    \path [->] (a3) edge node {} (a6);
    \path [->] (t1) edge node {} (t2);
    \path [->] (t2) edge[bend left] node {} (t1);
    \path [->] (t2) edge node {} (a3);

\end{scope}
\end{tikzpicture}}
\Description{}
    \caption{Gadget used in the proof of Lemma~\ref{lem:l=inftyk=2c=inftys=inftyn=2} showing \NP-hardness 
     when $\N=\{1,2\}$, $ \intv=\infty, \sv, \segv=\{\infty\}^2$, and for some fixed $\natv_2=2$ where country~$2$ is represented by the triangular vertices while $\natv_{1}\in \mathbb{N}_{\geq 0}\cup\{\infty\} $.
    }
    \label{fig:l=inftyk=2c=inftys=inftyn=2}
\end{figure*}
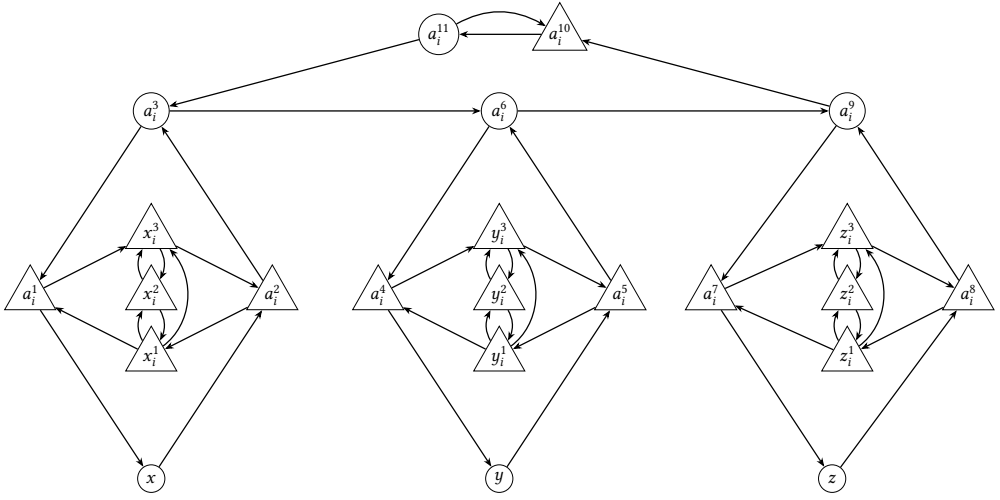

In our reduction, we create one vertex for every $w \in  X \cup Y \cup Z$. The remaining vertices and all edges are added via gadgets. One gadget is created for each $t_i=(x,y,z) \in T$, as depicted in Figure~\ref{fig:l=inftyk=0c=inftys=inftyn=2}. We add the following vertices $\{a_i^1, \cdots, a_i^9, x_i^1,x_i^2, x_i^3,y_i^1,y_i^2,y_i^3, z_i^1,z_i^2,z_i^3\}$ for the gadget corresponding to $t_i$. Hence, we have the following vertices and edges:
\[\begin{array}{c}
V = X \cup Y\cup Z \{ a_i^1, \ldots, a_i^{11}, x_i^1,x_i^2, x_i^3,y_i^1,y_i^2,y_i^3, z_i^1,z_i^2,z_i^3 \mid t_i \in T\} \quad \text{  and} \\
A = \{ (x, a_i^2),  (a_i^1, x),  (a_i^1, x_i^3),  (x_i^3, a_i^2),  (x_i^1, a_i^2),  (a_i^2, x_i^1), (x_i^2, x_i^3),  (x_i^3, x_i^2), (x_i^1, x_i^2), \\  (x_i^2, x_i^1), (x_i^1, x_i^3),  (a_i^2, a_i^3),  (a_i^3, a_i^1),  
(y, a_i^5),  (a_i^4, y),  (a_i^4, y_i^3),  (y_i^3, a_i^5),
(y_i^1, a_i^4), \\ (a_i^5, y_i^1),  
(y_i^2, y_i^3), (y_i^3, y_i^2),  (y_i^1, y_i^2),  (y_i^2, y_i^1), (y_i^1, y_i^3),  (a_i^5, a_i^6),  (a_i^6, a_i^4),  
(z, a_i^8),\\  (a_i^7, z),  (a_i^7, z_i^3),  (z_i^3, a_i^8),
(z_i^1, a_i^7),  (a_i^8, z_i^1),  
(z_i^2, y_i^3),  (z_i^3, z_i^2),  (z_i^1, y_i^2),  (z_i^2, z_i^1),\\ (z_i^1, y_i^3),  (a_i^8, a_i^9),  (a_i^9, a_i^7),  
 (a_i^6, a_i^9),  
(a_i^9, a_i^3),  (a_i^3, a_i^6) (a_i^{10}, a_i^{11}), (a_i^{11}, a_i^{10}) \mid t_i\in T
\}.
\end{array}
\]

In our reduction, we have $\N=\{1,2\}$, with the vertices belonging to countries~$1$ and $2$ depicted by the circular and triangular vertices in Figure~\ref{fig:l=inftyk=2c=inftys=inftyn=2}, respectively. Thus, $V_1=\{x,y,z,a_i^3,a_i^6, a_i^9, a_i^{11}\mid  t_i=(x,y,z)\in T\}$ and $V_2 =V\backslash V_1$.
We let $\natv_1\in \mathbb{N}_{\geq 0}\cup\{\infty\}$,  $\natv_2=2$, $\segv=(\infty,\infty)$ and $\sv=(\infty,\infty)$ and, thus, our restrictions for the country-specific parameters are as given in the Lemma statement. 
The following  cycles are $\Gamma$-cycles and abide by the country-specific parameters: 

\[\begin{array}{c}
\{
\langle x, a_i^2, x_i^1, a_i^1\rangle, 
\langle  x_i^2, x_i^3\rangle, 
\langle a_i^3, a_i^1, x_i^3, a_i^2 \rangle, 
\langle x_i^2, x_i^1\rangle, 
\langle y, a_i^5, y_i^1, a_i^4 \rangle,
\langle  y_i^2,y_i^3\rangle,
\langle  y_i^2,y_i^1 \rangle, 
\langle  z_i^3, z_i^2\rangle, 
\langle z_i^1, z_i^2\rangle, \\
\langle z, a_i^8, z_i^1,a_i^9\rangle,
\langle a_i^4, y_i^3, a_i^5, a_i^6\rangle,
\langle a_i^7,z_i^3, a_i^8, a_i^9\rangle,  
\langle  a_i^3, a_i^6, a_i^9, a_i^{10}, a_i^{11}\rangle , 
\langle  a_i^{10}, a_i^{11}\rangle 
\mid t_i \in T\}\subset \gcyc.
\end{array}\]

Next, we give a proof of correctness and first assume that there is a perfect matching $M$ of instance $\I$ of \textsc{P3DM} and construct a perfect $\Gamma$-cycle packing $\C$ of instance $\J$. 
 If a triple $t_i=(x,y,z)$ is chosen in $M$, then the following cycles are added to $\C$: 
\[\begin{array}{c}
\langle x, a_i^2, x_i^1, a_i^1\rangle, 
\langle  x_i^2, x_i^3\rangle, 
\langle y, a_i^5, y_i^1, a_i^4 \rangle,
\langle  y_i^2,y_i^3\rangle,
\langle z, a_i^8, z_i^1,a_i^9\rangle, 
\langle  z_i^3, z_i^2\rangle,  
\langle  a_i^3, a_i^6, a_i^9, a_i^{10}, a_i^{11}\rangle 
\in \C.
\end{array}\] 
Thus, all the nodes within the gadget are covered by some $\Gamma$-cycle.  If a triple $t_i=(x,y,z)$ is not chosen in $M$, then the following cycles are added to $\C$: 
\[\begin{array}{c} 
\langle a_i^3, a_i^1, x_i^3, a_i^2 \rangle,  
\langle x_i^2, x_i^1\rangle, 
\langle a_i^4, y_i^3, a_i^5, a_i^6\rangle,
\langle  y_i^2,y_i^1 \rangle, 
\langle a_i^7,z_i^3, a_i^8, a_i^9\rangle,  
\langle z_i^1, z_i^2\rangle , 
\langle a_i^{10}, a_i^{11} \rangle 
\in \C.
\end{array}\] 
Observe that, in the second case, all of the gadget's vertices have been selected in some cycle, apart from $x$, $y$, and $z$.  As $M$ is a perfect matching, every vertex $w \in X \cup Y\cup Z$ will be covered by some gadget exactly once. Therefore, all vertices are covered by a cycle in $\gcyc$ and instance $\J$ has a perfect $\Gamma$-cycle packing.

Conversely, suppose there is a perfect $\Gamma$-cycle packing $\C$  of instance $\J$. We construct a perfect matching $M$ in {\sc P3DM} instance $\I$ as follows. We argue about the cycles of gadget $t_i$ that are selected in perfect $\Gamma$-cycle packing $\C$ via the following two cases:
\begin{enumerate}
    \item Suppose $\langle x, a_i^2, x_i^1, a_i^1\rangle\in \C$. This means that $\langle x_i^2, x_i^3\rangle\in\C$, to ensure that these vertices are covered by a $\Gamma$-cycle. 
    In turn, $a_i^3$ must then be covered by $\langle  a_i^3, a_i^6, a_i^9, a_i^{10}, a_i^{11} \rangle\in\C$. 
    Then observe that the only $\Gamma$-cycle to cover $a_i^4$ is then $\langle  y, a_i^5, y_i^1, a_i^4 \rangle$, all other cycles are national $V_2$-cycles with~$4$ vertices,  exceeding country~$2$'s national cycle limit $\natv_2=2$.  
    Similar reasoning entails that $\langle z, a_i^8, z_i^1,a_i^9\rangle\in \C$. 
    Thus,  when $\langle x, a_i^2, x_i^1, a_i^1\rangle\in \C$ and $x$ is selected by the gadget corresponding to $t_i$, then $y$ and $z$ must also be selected by cycles from the same gadget.
    \item Suppose $\langle x, a_i^2, x_i^1, a_i^1\rangle\notin \C$. Then, the possible ways that $a_i^2$ can then be covered is either  \emph{(a)} $\langle a_i^1, x_i^3, a_i^2, x_i^1 \rangle$ and $\langle x_i^3, a_i^2, x_i^1, x_i^2 \rangle$, \emph{(b)} $\langle a_i^3, a_i^1, x, a_i^2 \rangle$, \emph{(c)} a cycle $C$ containing the path $[x, a_i^2, a_i^3, a_i^6]$, or \emph{(d)} $\langle a_i^3, a_i^1, x_i^3, a_i^2 \rangle$. 
    \begin{enumerate}
        \item We first consider the cycles $\langle a_i^1, x_i^3, a_i^2, x_i^1 \rangle$ and $\langle x_i^3, a_i^2, x_i^1, x_i^2 \rangle$. Note that these are not $\Gamma$-cycles, as they are both $V_2$-cycles that exceed $\natv_2=2$. 
        \item Next consider if $a_i^2$ were covered by $\langle a_i^3, a_i^1, x, a_i^2 \rangle$ in $\C$. This would leave the vertices $x_i^1$, $x_i^2$, and $x_i^3$  to be covered by $\langle x_i^1, x_i^2, x_i^3\rangle$, a $V_2$-cycle of length~$3$, exceeding $\natv_2=2$. Therefore,  $\langle a_i^3, a_i^1, x, a_i^2 \rangle\notin\C$. 
        \item Next consider if $a_i^2$ were covered by a cycle $C$ containing the path $[x, a_i^2, a_i^3, a_i^6]$. This would entail that the cycle $\langle a_i^1, x_i^3, x_i^2, x_i^1 \rangle$ must be selected in $\C$; however, this is a national $V_2$-cycle of length~$4$, exceeding $\natv_2=2$.
    \end{enumerate}
    Therefore, the only option for $a_i^2$ to be covered is \emph{(d)} $\langle a_i^3, a_i^1, x_i^3, a_i^2 \rangle\in\C$. Thus, $a_i^{10}$ and $a_i^{11}$ must then be selected  by $\langle a_i^{10},a_i^{11}\rangle\in\C$. Next, we address how $a_i^6$ can be covered in $\C$. It cannot be in a cycle that spans more than one gadget, given that $a_i^3$ has previously been selected. Therefore, $a_i^6$ must either be covered by $\langle a_i^6, a_i^4, y, a_i^5 \rangle$ or  $\langle a_i^6, a_i^4, y_i^3, a_i^5 \rangle$. In the former case, this would leave $y_i^1, y_i^2$, and $y_i^3$ left to be covered by $\langle y_i^1, y_i^2, y_i^3\rangle$, a $V_2$-cycle of length~$3$, exceeding $\natv_2=2$. Therefore, $\langle a_i^6, a_i^4, y_i^3, a_i^5 \rangle\in\C$. Consequently, $\langle  y_i^2,y_i^1 \rangle\in\C$ to cover these vertices. A similar argument holds to entail that $\langle a_i^9, a_i^7, z_i^3, a_i^8 \rangle, \langle  z_i^2,z_i^1 \rangle\in\C$. 
    Hence, when $\langle x, a_i^2, x_i^1, a_i^1\rangle\notin \C$ and thus, $x$ is not covered with gadget $t_i$, then $y$ and $z$ cannot be covered by $t_i$ without leaving some vertex uncovered that only exists in gadget $t_i$. 
\end{enumerate}
 % \noindent 
 Therefore, it must be the case that either 
 \[\begin{array}{c}\{\langle x, a_i^2, x_i^1, a_i^1\rangle, 
\langle  x_i^2, x_i^3\rangle, 
\langle y, a_i^5, y_i^1, a_i^4 \rangle,
\langle  y_i^2,y_i^3\rangle,
\langle z, a_i^8, z_i^1,a_i^9\rangle, 
\langle  z_i^3, z_i^2\rangle,  
\langle  a_i^3, a_i^6, a_i^9, a_i^{10}, a_i^{11}\rangle \}
\subset \C \text{  or}\\
\{\langle a_i^3, a_i^1, x_i^3, a_i^2 \rangle,  
\langle x_i^2, x_i^1\rangle, 
\langle a_i^4, y_i^3, a_i^5, a_i^6\rangle,
\langle  y_i^2,y_i^1 \rangle, 
\langle a_i^7,z_i^3, a_i^8, a_i^9\rangle,  
\langle z_i^1, z_i^2\rangle , 
\langle  a_i^{10}, a_i^{11}\rangle\}
\subset \C,
\end{array}\] for each $t_i \in T$. In the former case, we add $t_i$ to $M$. 
Finally, $M$ is a perfect matching in $\I$ because $\C$ is a perfect $\Gamma$-cycle packing in $\J$, and in view of points~$(1)$ and ~$(2)$ above. 

Therefore, {\sc Perfect $\Gamma$-Cycle Packing} is \NP-complete when $n =2$, $\intv=\infty$, $\natv_j=2$ for some $j \in \N$, and $\segv=\sv=(\infty,\infty)$.
\end{proof}

\subsection{Establishing Theorem~\ref{thm:dico}}\label{app:Cases} 

We will use the results summarised in Table~\ref{tab:dicho} to prove that we have a dichotomy with respect to the country-specific parameters between the polynomial-time solvable and \NP-hard cases of {\sc Max $\Gamma$-Cycle Packing}. 
Given Theorem~\ref{t-dicho}, we assume that $n \geq 2$ and we show for each value of $\intv$ that we have a dichotomy. 

\begin{lemma}\label{lem:l=0}
 Let $\Gamma=(n,\intv,\natv,\segv,\sv)$ be given where $\intv=0$. 
 {\sc Max $\Gamma$-Cycle Packing} is polynomial-time solvable when $\natv_i\in \{0,2,\infty\}$ for all $i \in \N$. Otherwise, it is \NP-hard.
\end{lemma}
\begin{proof}
The polynomial-time solvability of {\sc Max $\Gamma$-Cycle Packing} when $\intv=0$ and $\natv_i\in \{0,2,\infty\}$ for all $i\in \N$ follows from Theorem~\ref{t-dicho}. 
As there are no international cycles, we can perform cycle packing on the individual national KEP pools, which can each be solved in polynomial time since $\natv_i \in \{0,2,\infty\}$ for all $i \in \N$.

If $\intv=0$ and $\natv_i= 3$ for some $i \in\N$, then {\sc Max $\Gamma$-Cycle Packing} being \NP-hard follows from Theorem~\ref{t-dicho}. 
\end{proof}

\begin{lemma}
 Let $\Gamma=(n,\intv,\natv,\segv,\sv)$ be given where $\intv=2$. {\sc Max $\Gamma$-Cycle Packing} is polynomial-time solvable when $\natv_i \in \{0,2\}$ for all $i \in \N$. Otherwise, it is \NP-hard. 
   \end{lemma}
\begin{proof}
     We first show that {\sc Max $\Gamma$-Cycle Packing} is polynomial-time solvable when $\intv=2$ and for all $i \in \N $, $\natv_i \in \{0,2\}$. As $\intv, \natv_i \leq 2$ for all $ i \in \N$, all $\Gamma$-cycles are of length at most~$2$. 
     Hence, Edmonds' algorithm~\cite{micali1980v} can be used, as shown by \citet{roth2005pairwise} in the KEP setting when there are only cycles of length~$2$. Note that this algorithm runs in polynomial time. The algorithm converts the compatibility graph into an undirected graph where an edge represents a pairwise exchange between two vertices. Note that for countries $j \in\N$ with $\natv_j=0$, undirected edges are not added for their national pairwise exchanges. 

    This polynomial time solvable case, given that $\intv=2$,  is when  $\natv_i \in \{0,2\}$ for all $i \in \N$. Therefore, in all remaining cases, there must be some $i \in \N$ such that $\natv_i \geq 3$.
    First, assume $n=2$.
    Recall that Lemma~\ref{lem:n=2l=2k=inftyAnycs} show that {\sc Max $\Gamma$-Cycle Packing} is \NP-hard  when $\intv=2$ and there is some $i \in \N$ such that $\natv_i\in \mathbb{N}_{\geq 3}\cup\{\infty\}$. This is reiterated in Row~\ref{row:n=2l=2k=inftyAnycs} of Table~\ref{tab:dicho}, showing that when $\intv=2$, and there is some $i \in \N$ such that $\natv_i \in \mathbb{N}_{\geq 3}\cup\{\infty\}$ while there are no other restrictions placed on $\natv_j, \segv, \sv$, i.e., $\natv_j\in \mathbb{N}_{\geq 0}\cup\{\infty\}$ for $j \in \N\backslash\{i\}$, $\segv\in\mathbb{N}_{\geq 1}\cup\{\infty\}^n$ and $\sv\in\mathbb{N}_{\geq 1}\cup\{\infty\}^n$. Hence, all other cases of {\sc Max $\Gamma$-Cycle Packing} are \NP-hard when $\intv=2$.
    We can extend this to any $n\geq 3$ by adding countries that are isolated vertices in the computability graph.   
\end{proof}

\begin{lemma}
    Let $\Gamma=(n,\intv,\natv,\segv,\sv)$ be given where $\intv=3$. {\sc Max $\Gamma$-Cycle Packing} is polynomial-time solvable when $n=2$, $\segv=(1,1)$ and $\natv_i \in \{0,2\}$ for all $i \in \N$. Otherwise, it is \NP-hard.
\end{lemma}
\begin{proof}
    We first show that {\sc Max $\Gamma$-Cycle Packing} is polynomial-time solvable when $\intv=3$,  $n=2$, $\segv=(1,1)$ and $\natv_i \in \{0,2\}$ for all $i \in \N$. This is due to all $\Gamma$-cycles being of length~$2$. National $\Gamma$-cycles are of length at most~$2$ given that $\natv_i \in \{0,2\}$ for all $i \in \N$. We next show that all international $\Gamma$-cycles are of length~$2$ when $n=2$, $\segv=(1,1)$ and $\intv=3$. Consider an international cycle with a segment of size~$1$ from country~$1$ followed by a segment of size~$1$ from country~$2$. This international cycle cannot be extended to a cycle of length~$3$ given that $\segv=(1,1)$. Adding another segment of size~$1$ from either country to this international cycle will create a segment of length~$2$ when the segments are joined to make a cycle. As all $\Gamma$-cycles are of size~$2$, we can use Edmonds' algorithm~\cite{micali1980v} as shown by \citet{roth2005pairwise} in the KEP setting.

    To show the dichotomy in complexity when restricted to $\intv=3$, we show that {\sc Max $\Gamma$-Cycle Packing} becomes \NP-hard when increasing each country-specific parameter from the polynomial-time solvable case.  

    We first focus on increasing the number of countries to $n \geq 3$ and show that the problem is then \NP-hard. Recall that Lemma~\ref{lem:Gl=3c=3s=1n=3} shows \NP-hardness when $n=3$, $\intv=3$, $\segv=\{1\}^n$, and $\natv_i \in \{0,2\}$ for all $i \in \N$. We can extend this to any $n\geq 4$ by adding countries that are isolated vertices in the computability graph.
    Observe Row~\ref{row:Gl=3c=3s=1n=3} of Table~\ref{tab:dicho}, showing that  Lemma~\ref{lem:Gl=3c=3s=1n=3} holds for any $\natv_i\in \mathbb{N}_{\geq 0}\cup\{\infty\}$ and $\segv_i, \sv_i\in \mathbb{N}_{\geq 1}\cup\{\infty\}$, and therefore in all other cases when $\intv=3$ and the number of countries is increased $n \geq 3$, we see that the problem becomes \NP-hard.

    We next increase the national cycle limits such that at least one country has a national cycle limit of at least length~$3$. 
    Recall that   
    Lemma~\ref{lem:n=2l=2k=inftyAnycsSEGV} shows that when $\intv=3$, $n = 2$, $\segv=(1,1)$ and there is some $i \in \N$ such that $\natv_i=3$, 
    then {\sc Max $\Gamma$-Cycle Packing} is \NP-hard. Moreover, as shown in Row~\ref{row:n=2l=2k=inftyAnycsSEGV} of Table~\ref{tab:dicho}, Lemma~\ref{lem:n=2l=2k=inftyAnycsSEGV} also holds for any any $\natv$ such that  for some $i \in \N$ we have that $\natv_i \in \mathbb{N}_{\geq 3}\cup\{\infty\} $ and for the other country $j \in \N\backslash \{i\}$ we have that $\natv_j\in \mathbb{N}_{\geq 0}\cup\{\infty\}$. Thus, it holds for any variation of national cycle limits and, therefore, in all other cases when $n=2$, $\intv=3$, $\segv=(1,1)$, any $\sv$, and any $\natv$ when one country has a national limit of at least $3$ or unbounded we see that {\sc Max $\Gamma$-Cycle Packing} is \NP-hard. 
    
    Finally, we address altering the international segment sizes such that at least one country has its international segment size to be at least~$2$. 
    Recall that Lemma~\ref{lem:n=2l=2k=2Anycs} shows that {\sc Max $\Gamma$-Cycle Packing} is \NP-hard when $\N=\{1,2\}$, $\intv=3$, $\natv_1, \natv_2 \in\mathbb{N}_{\geq 0}\cup\{\infty\}$, $\sv_1, \sv_2\in \{\mathbb{N}_{\geq 
    1}, \infty\}$, $\segv_1 \in \mathbb{N}_{\geq 2}\cup\{\infty\}$ and $\segv_2 \in \mathbb{N}_{\geq 1}\cup\{\infty\}$ (also given in Row~\ref{row:n=2l=2k=2Anycs} of  Table~\ref{tab:dicho}). 
    Therefore, we have shown that for all other cases the problem is \NP-hard when altering the international segment sizes from this polynomial time solvable case. 
    
    Thus, we have shown a dichotomy in the complexity of {\sc Max $\Gamma$-Cycle Packing} when  $\intv=3$.
\end{proof}

\begin{lemma}
     Let $\Gamma=(n,\intv,\natv,\segv,\sv)$ be given where $\intv\in \mathbb{N}_{\geq 4}$. {\sc Max $\Gamma$-Cycle Packing} is polynomial-time solvable when $n=2$, $\segv=(1,1)$, $\natv_i \in \{0,2\}$ for all $i \in \N$ and there is some $j \in\N$ such that $\sv_j=1$. Otherwise, it is \NP-hard.\label{lem:l>4}
\end{lemma}
\begin{proof}
    We first show that {\sc Max $\Gamma$-Cycle Packing} is polynomial-time solvable when $\intv\in\mathbb{N}_{\geq 4}$, $n=2$, $\segv=(1,1)$, there is some $j \in\N$ such that $\sv_j=1$ and $\natv_i\in \{0,2\}$ for all $i \in \N$. This is due to all $\Gamma$-cycles having at most~$2$ vertices.  
    First, observe that all national $\Gamma$-cycles are of length at most~$2$ due to $\natv_i \in \{0,2\}$ for all $i \in \N$. Secondly, all international $\Gamma$-cycles  also have at most~$2$ vertices. Given that $\segv=(1,1)$, every international $\Gamma$-cycle consists of alternating international segments of size~$1$ from the two countries.  
    As there is some $j \in \N$ with $\sv_j=1$, all international $\Gamma$-cycles will consist of $2$ vertices (i.e., one cycle of length~$2$, built from two segments of size~$1$, one from each country). Although the second country may have a higher segment number, i.e., for country $i\in \N$ such that $i\neq j$ with  $\sv_i > 1$, another $V_i$-segment cannot be added without another $V_j$-segment. Thus, all $\Gamma$-cycles are of length at most~$2$, and we can solve the problem in polynomial time again using Edmonds' algorithm~\cite{micali1980v}.

    We next show that the remaining cases of {\sc Max $\Gamma$-Cycle Packing} when $\intv\in \mathbb{N}_{\geq 4}$ are \NP-hard.

    We first show that the problem becomes \NP-hard when the number of countries is increased such that $n\geq3$.   
    Observe Row~\ref{row:Gl=3c=3s=1n=4}   of Table~\ref{tab:dicho}, showing that  Lemma~\ref{lem:Gl=3c=3s=1n=4} holds  when $n=3$, $\intv\in \mathbb{N}_{\geq 4}$, for any $\natv_i\in \mathbb{N}_{\geq 0}\cup\{\infty\}$, $\segv_i\in \mathbb{N}_{\geq 1}\cup\{\infty\}$ for all $i \in \N$, and for some $j \in \N$ we have that $\sv_j=1$  and no restriction on the remaining $j' \in \N\backslash\{j\}$ having $\sv_{j'}\in \mathbb{N}_{\geq 1}\cup\{\infty\}$.
    We can extend this to any $n\geq 4$ by adding countries that are isolated vertices in the computability graph.
    Therefore, in all other cases when $\intv\in \mathbb{N}_{\geq 4}    $ and the number of countries is increased $n \geq 3$, we see that {\sc Max $\Gamma$-Cycle Packing} becomes \NP-hard.

    Next, we address increasing the international segment size vector such that at least one country has a segment size of at least~$2$. Thus, we need to show that {\sc Max $\Gamma$-Cycle Packing} is \NP-hard when  $n=2$, $\intv\in\mathbb{N}_{\geq 4}$, and at least one country has an international segment number of one i.e., for some $i \in \N$ $\segv_i=1 $  while $i' \in \N\backslash\{i\}$ has $\segv_{i'} \in \mathbb{N}_{\geq 1}\cup\{\infty\}$, and least one country has an international segment size of least~$2$, without a restriction on the other countries, i.e., for some $j \in \N$, $\segv_j \in \mathbb{N}_{\geq 2}\cup\{\infty\}$  while $j' \in \N\backslash\{j\}$ has $\segv_{j'} \in \mathbb{N}_{\geq 1}\cup\{\infty\}$. This was shown in Lemma~\ref{lem:n=2l=2k=2Anycs4} as seen in Row~\ref{row:n=2l=2k=2Anycs4} of Table~\ref{tab:dicho}. 
    
    Next, we address increasing the national cycle limit such that at least one country has a national cycle limit of $3$. We assume that $n=2$, $\intv\in\mathbb{N}_{\geq 4}$, $\segv=(1,1)$ and there is some $j \in\N$ such that $\sv_j=1$ while the other country $j'\in\N$ has $\sv_{j'}\in\mathbb{N}_{\geq 1}\cup\{\infty\}$, and there is some $i \in \N$ such that $\natv_i \in  \mathbb{N}_{\geq 3}\cup\{\infty\}$ while the other country $i'\in\N$ has $\sv_{i'}\in\mathbb{N}_{\geq 0}\cup\{\infty\}$.  Lemma~\ref{lem:n=2l=2k=inftyAnycsSV} (and Row~\ref{row:n=2l=2k=inftyAnycsSV} of Table~\ref{tab:dicho}) shows it to be \NP-hard. 

    Finally, we increase the international segment numbers such that no country has an international segment number of~$1$. In Lemma~\ref{lem:l=4k>=2c=1s=2n=2} we show that {\sc Max $\Gamma$-Cycle Packing} is \NP-hard when $\N=\{1,2\}$, $\intv\in\mathbb{N}_{\geq 4}$, $\segv=(1,1)$, $\natv_1, \natv_2 \in \{0,2\}$, and $ \sv_1, \sv_2 \in  \mathbb{N}_{\geq 2}\cup\{\infty\}$ (see Row~\ref{row:l=4k>=2c=1s=2n=2} of Table~\ref{tab:dicho}). That is, for any increase in international segment numbers such that no country has a segment number of~1, then {\sc Max $\Gamma$-Cycle Packing} becomes \NP-hard. 

    Hence, we have shown that in all other cases {\sc Max $\Gamma$-Cycle Packing} is \NP-hard given that $\intv\in \mathbb{N}_{\geq 4}$. 
 \end{proof}

\begin{lemma}
    Let $\Gamma=(n,\intv,\natv,\segv,\sv)$ be given where $\intv=\infty$. {\sc Max $\Gamma$-Cycle Packing} is polynomial-time solvable in the following three cases: 
    \begin{enumerate}
    \item when $n= 2$, $\natv_i \in \{0,2\}$ for all $i \in \N$, $\segv=(1,1)$ and there is some $j\in \N$ such that $\sv_j=1$;
    \item when $\natv=\{0\}^n$, $\segv=\{1\}^n$, $\sv=\{\infty\}^n$;
    \item  and, when $\natv{=}\{\infty\}^n$, $\segv{=}\{\infty\}^n$, $\sv{=}\{\infty\}^n$.
    \end{enumerate}
Otherwise, it is \NP-hard.
\end{lemma}

\begin{proof}

We first show that the three cases listed in the Lemma statement are all polynomial-time solvable. Following this, we show \NP-hardness for the remaining cases. 

\paragraph{Case~1.} We first show that {\sc Max $\Gamma$-Cycle Packing} is solvable in polynomial time when $\intv=\infty$, $n= 2$, $\natv_i\in\{0,2\}$ for all $i \in \N$, $\segv=\{1\}^n$ and there exists a $j\in \N$ such that $\sv_j=1$. This follows the same reasoning and algorithm as the polynomial-time solvable case in Lemma~\ref{lem:l>4}. 
    
\paragraph{Case~2.} Our second polynomial time solvable case of {\sc Max $\Gamma$-Cycle Packing} is when $\intv=\infty$, $n\geq2$, $\natv=\{0\}^n$, $\segv=\{1\}^n$, $\sv=\{\infty\}^n$. As $\natv=\{0\}^n$ and $\segv=\{1\}^n$, we see that any $\Gamma$-cycle will not contain any national arcs. Hence, a cycle of any length that uses only international arcs will be a $\Gamma$-cycle.  Thus, the same folklore trick described by \citet{abraham2007origin} for the unbounded KEP setting (as given in Theorem~\ref{t-dicho}) can be used here. However, in the modified graph, we only consider international arcs.

\paragraph{Case~3.} Our final polynomial-time solvable case of {\sc Max $\Gamma$-Cycle Packing} is when all of the country-specific parameters are unbounded, i.e., $\intv=\infty$, $n \geq 2$, $\natv=\{\infty\}^n$, $\segv=\{\infty\}^n$, and $\sv=\{\infty\}^n$. The same folklore trick described by \citet{abraham2007origin} for the unbounded KEP setting (as given in Theorem~\ref{t-dicho}) can be applied in our setting. 

We next show these are the only polynomial-time solvable cases when $\intv=\infty$. 

\paragraph{Extending Case 1.} This follows from the proof of Lemma~\ref{lem:l>4}. The remaining hardness cases extending these polynomial-time solvable case are covered by Rows~\ref{row:Gl=3c=3s=1n=4}, \ref{row:n=2l=2k=inftyAnycsSV}, \ref{row:n=2l=2k=2Anycs4}, and~\ref{row:l=4k>=2c=1s>=2n=2} in Table~\ref{tab:dicho}.
We can extend all these hardness results to larger values of $n$ by adding extra countries as isolated vertices in the computability graph.

\paragraph{Extending Case 2.}
Returning to Case~2 where  $\natv=\{0\}^n$, $\segv=\{1\}^n$,  and $\sv=\{\infty\}^n$, we next show that altering any of these parameters leads to {\sc Max $\Gamma$-Cycle Packing} being \NP-hard.

First, we increase the length of the national cycle of at least one country. Hence, we assume that there is some $j \in \N$ such that $\natv_j\in \mathbb{N}_{\geq 2}\cup\{\infty\}$, while $\intv=\infty$, $n=2$, $\segv=\{1\}^n$, $\sv=\{\infty\}^n$.   Lemma~\ref{lem:l=inftyk=0c=1s=inftyn=2}, which we can extend to $n\geq 2$ by adding new countries as isolated vertices in the computability graph, shows that {\sc Max $\Gamma$-Cycle Packing} is \NP-hard when $\intv=\infty$, $\N=\{1, 2\}$, $\segv=(1, 1)$, $\sv=(\infty, \infty)$ and $\natv_1\in \mathbb{N}_{\geq 2}\cup\{\infty\}$ and $\natv_2\in \mathbb{N}_{\geq 0}\cup\{\infty\}$ (see Row~\ref{row:l=inftyk=0c=1s=inftyn=2} of Table~\ref{tab:dicho}). That is, for any other national cycle limit vector $\natv$, {\sc Max $\Gamma$-Cycle Packing} becomes \NP-hard.

Next, we address the effect of an increasing $\segv$; thus, assuming that $\intv=\infty$, $\natv=\{0\}^n$, $\sv=\{\infty\}^n$, and $\segv_i \in \mathbb{N}_{\geq 2}$ for some $i \in \N$ and $\segv_j \in \mathbb{N}_{\geq 1}\cup\{\infty\}$ for all $j \in \N\backslash\{i\}$. 
In Lemma~\ref{lem:l=inftyk>=0c=1s>=2n=2} (see Row~\ref{row:l=inftyk>=0c=1s>=2n=2} in Table~\ref{tab:dicho}), {\sc Max $\Gamma$-Cycle Packing} was shown to be \NP-hard under these conditions and showing that for any change of $\segv$ will entail \NP-hardness. We can extend both lemmas to be valid for larger values of $n$ by adding new countries as isolated vertices in the computability graph.

Finally, we address changing the segment number of at least one country from being unbounded to being a natural number, i.e., there exists some $j$ such that $\sv_j\in \mathbb{N}$ while $\natv=\{0\}^n$, $\segv=\{1\}^n$. 
First suppose $n=2$. Recall that in Case~1, if one country's segment number is~$1$, the problem becomes solvable in polynomial time. Now suppose $n\geq 3$. In this case we apply Lemma~\ref{lem:n=2l=2k=2Anycs} (see Row~\ref{row:n=2l=2k=2Anycs} in Table~\ref{tab:dicho}).
Thus, we must show that {\sc Max $\Gamma$-Cycle Packing} is \NP-hard when $\intv=\infty$, $\natv=\{0\}^n$, $\segv=\{1\}^n$ and for some $i \in \N$ we have that $\sv_i\in  \mathbb{N}_{\geq 2} $ while for all $j \in \N\backslash\{i\}$, we have that $\sv_j\in  \mathbb{N}_{\geq 2}\cup\{\infty\}$. Hence, {\sc Max $\Gamma$-Cycle Packing} was shown to be \NP-hard in the remainder of the cases when varying the segment numbers in Lemma~\ref{lem:l=4k>=2c=1s>=2n=2} (see Row~\ref{row:l=4k>=2c=1s>=2n=2} in Table~\ref{tab:dicho}). 
Again, we can extend the results used in these lemmas to be valid for larger values of $n$ by adding new countries as isolated vertices in the computability graph.

\paragraph{Extending Case 3.}
Lastly, we modify parameters relating to the third polynomial time solvable case of {\sc Max $\Gamma$-Cycle Packing} where $\intv=\infty$, $\natv=\{\infty\}^n$, $\segv=\{\infty\}^n$, $\sv=\{\infty\}^n$ with $n\geq2$.

We first look at restricting at least one country's segment number to be a natural number, hence assuming that $\intv=\infty$, $\natv=\{\infty\}^n$, $\segv=\{\infty\}^n$, there is some $i \in \N$ such that $\sv_i \in \mathbb{N}_{\geq 1}$ and the remaining countries $j \in \N\backslash\{i\}$ have $\sv_j \in \mathbb{N}_{\geq 1}\cup\{\infty\}$. 
{\sc Max $\Gamma$-Cycle Packing} was shown to be \NP-hard under these conditions in Lemma~\ref{lem:l=inftyc=inftys=1,infty} (see Row~\ref{row:l=inftyc=inftys=1,infty} in Table~\ref{tab:dicho}). 

 Next, we assume that at least one country's segment size is restricted to be a natural number, hence assuming that $\intv=\infty$, $\natv=\{\infty\}^n$, $\sv=\{\infty\}^n$, there is some $i \in \N$ such that $\segv_i \in \mathbb{N}_{\geq 1}$ and the remaining countries $j \in \N\backslash\{i\}$ we have $\segv_j \in \mathbb{N}_{\geq 1}\cup\{\infty\}$. Recall that {\sc Max $\Gamma$-Cycle Packing} was shown to be \NP-hard under these conditions in Lemma~\ref{lem:lk=intyc=mathbbs=infty} (see Row~\ref{row:lk=intyc=mathbbs=infty} in Table~\ref{tab:dicho}). 

Finally, we address altering the national cycle limits, and at least one country has reduced its national cycle limit to a natural number. 
Thus, we need to show that $n \geq 2$, $\intv=\infty$, $\sv, \segv=\{\infty\}^n$, for some $i \in \N$ we have that $\natv_i \in \mathbb{N}_{\geq 0}$ and the remaining countries $j \in \N\backslash\{i\}$ have any national cycle limit $\natv_j \in \mathbb{N}_{\geq 0}\cup\{\infty\}$.
Lemma~\ref{lem:l=inftyk=2c=inftys=inftyn=2} shows this case of {\sc Max $\Gamma$-Cycle Packing} to be \NP-hard (see Row~\ref{row:l=inftyk=2c=inftys=inftyn=2} in Table~\ref{tab:dicho}). 
We note that can extend all lemmas in this case to be valid for larger values of $n$ by adding new countries as isolated vertices in the computability graph.

We conclude that the three stated cases of {\sc Max $\Gamma$-Cycle Packing} in this lemma are polynomial-time solvable and that all the remaining cases are \NP-hard. 
\end{proof}

Combining the previous lemmas with Theorem~\ref{t-dicho} for the case where $n=1$ provides us with our dichotomy theorem, which we restate below.

\dicho*

\end{document}